\documentclass[12pt, reqno,titlepage, hyperfootnotes=false]{amsart}
\usepackage{lmodern}
\usepackage[T1]{fontenc}
\usepackage{mathtools}
\usepackage[a4paper]{geometry}
\usepackage{color}
\usepackage{array}
\usepackage{bm}
\usepackage{amstext}
\usepackage{amsthm}
\usepackage{amssymb}
\usepackage{ragged2e}
\usepackage{multibib}

\newcites{supp}{References}

\usepackage{graphicx}
\usepackage{nicefrac}
\usepackage{setspace}
\usepackage{microtype}
\usepackage{xcolor}
\usepackage[textsize=tiny]{todonotes}
\usepackage[authoryear]{natbib}
\usepackage{xr}
\allowdisplaybreaks
\usepackage[unicode=true,pdfusetitle,
 bookmarks=true,bookmarksnumbered=false,bookmarksopen=false,
 breaklinks=false,pdfborder={0 0 0},pdfborderstyle={},backref=false,colorlinks=true]
 {hyperref}

\makeatletter

\numberwithin{equation}{section}
\numberwithin{figure}{section}

\usepackage{amsfonts}
\usepackage{amstext}
\usepackage{amsthm}

\usepackage{hyperref}
\hypersetup{
    colorlinks = true,
    citecolor = black
}

\usepackage{threeparttable}
\usepackage{graphicx}
\usepackage{enumerate}
\usepackage{listings}
\usepackage{subcaption}
\usepackage{caption}
\DeclareMathOperator*{\argmin}{arg\,min}

\DeclareMathOperator*{\argmax}{arg\,max}

\usepackage{float}
\newtheoremstyle{compact}
{4pt}   
{4pt}   
{\itshape}
{}
{\bfseries}
{.}
{.5em}
{}

\theoremstyle{compact}
\newtheorem{prop}{Proposition}
\newtheorem{thm}{Theorem}

\newtheorem{lem}{Lemma}
\newtheorem*{lem*}{Lemma}
\newtheorem*{thm*}{Theorem}

\newtheorem{assumption}{Assumption}

\theoremstyle{definition}

\newtheorem{example}{Example}[section]

\def\E{\mathbb{E}}
\def\D{\mathcal{D} }
\def\N{\mathcal{N} }
\def\R{\mathbb{R} }
\def\P{\mathbb{P}}
\def\G{\mathcal{G}}

\usepackage[textsize=tiny]{todonotes}
\newcommand{\red}[1]{\textcolor{red}{#1}}

\g@addto@macro\normalsize{%
	\setlength\abovedisplayskip{4pt}%
	\setlength\belowdisplayskip{4pt}%
	\setlength\abovedisplayshortskip{2pt}%
	\setlength\belowdisplayshortskip{2pt}%
}
\makeatletter
{
	\renewcommand\subsection{\@startsection{subsection}{2}{\z@}%
		{0pt}
		{0.4ex}
		{\normalfont\normalsize\bfseries}}
}
\makeatother

\renewcommand{\footnotesize}{\scriptsize}

\begin{document}
\pagenumbering{gobble} 
\title{Empirical Bayes for compound adaptive experiments}
\author{Karun Adusumilli\textsuperscript{*}}
\author{Jiaying Gu\textsuperscript{\textdagger}}
\author{Junfan Tao$^\ddagger$}

\begin{abstract}
We investigate Empirical Bayes (EB) methods in the context of compound adaptive experiments, where the arm distribution in each experiment follows a normal distribution with an unknown mean that we seek to estimate. There are two main EB strategies: $g$-modeling, which estimates the prior by maximizing the marginal likelihood, and $f$-modeling, which derives posterior means directly from the empirical distribution of the observations. We show that $g$-modeling continues to be a valid EB procedure even when it incorrectly assumes that data are collected exogenously; its validity does not depend on the particular sampling algorithm or on whether sample sizes are endogenous. In practice, one can apply standard 
$g$-modeling techniques by acting as though the data were exogenously sampled. We extend regret guarantees from exogenous sampling to adaptively generated data. By contrast, naively applying the Tweedie formula based on the marginal density of the observed data, as in standard $f$-modeling, can produce biased rules under adaptive sampling. We corroborate the robustness of $g$-modeling through simulations with widely used adaptive algorithms and demonstrate its applicability using a real-world dataset consisting of multiple sequential experiments.\end{abstract}

\thanks{\textit{This version}: \today{}. \\
\textsuperscript{*}Department of Economics, University
of Pennsylvania; akarun@sas.upenn.edu.\\
\textsuperscript{\textdagger}Department of Economics, University of Toronto;
jiaying.gu@utoronto.ca.\\
$^{\ddagger}$Institute of Economic Research, Kyoto University; 
tao.junfan.7j@kyoto-u.ac.jp.
}

\maketitle

\newpage 
\pagenumbering{arabic}

\section{Introduction \protect\label{sec:Introduction}}

Recent years have seen remarkable advancements in the theory and application of adaptive experiments. These cutting-edge designs are now widely adopted across a range of disciplines, including online advertising \citep{russo2017tutorial}, dynamic pricing \citep{ferreira2018online}, drug discovery \citep{wassmer2016group}, public health \citep{athey2021shared}, and economic interventions \citep{kasy2019adaptive}. Compared to traditional randomized trials, adaptive experiments offer a more efficient and flexible framework for balancing welfare, ethical, and economic considerations. Their potential has been widely recognized---e.g., since the launch of the Critical Path Initiative in 2006, the FDA has actively promoted the adoption of adaptive designs in clinical trials to reduce costs and mitigate risks for participants.

While a rich literature has introduced innovative algorithms for implementing adaptive designs, fundamental challenges remain in post-experiment estimation. For example, technology firms routinely run thousands of adaptive experiments, commonly known as Online Controlled Experiments (OCEs), to assess and improve digital products and services. In a standard OCE, visitors to a website are randomly assigned either to a control group or to one of $K$ treatment groups, and the allocation probabilities are continually adjusted in response to interim outcomes. Estimating treatment effects in such adaptive frameworks, however, poses distinctive challenges. Conventional methods, like simple differences in sample means, break down because the adaptive data collection process violates the assumptions underlying traditional statistical procedures. In practice, Bayesian methods are therefore often used, where a prior is placed on the treatment effect, and the posterior mean is reported. This strategy, however, makes the results sensitive to the choice of prior and can yield biased estimates when the prior is misspecified.

These challenges raise important methodological questions: How can treatment effects be estimated efficiently and reliably in adaptive experimental settings? Can information across multiple experiments be aggregated to enhance estimation accuracy? Moreover, can we learn the prior distribution of the treatment effects directly from the data?

To address these questions, this paper investigates the use of Empirical Bayes (EB) methods in the context of compound adaptive experiments, where multiple adaptive experiments share a common but unknown distribution of effect sizes. As noted by \cite{efron2014two}, there are two main strategies for EB estimation: $g$-modeling, which estimates the prior by maximizing the marginal likelihood, and $f$-modeling, which computes the posterior means directly from the sample distribution of the observations. Our key contribution is to establish that $g$-modeling remains valid even under adaptive sampling, whereas $f$-modeling fails.

Remarkably, we show that \( g \)-modeling requires no knowledge of the data-generating algorithm. Standard \( g \)-modeling procedures can be applied as if the data were exogenously drawn and the sample means were normally distributed - even when neither assumption holds in reality. Moreover, each treatment arm can be analyzed independently, even when some or all arms are part of a common adaptive design. Despite these apparent misspecifications, we establish that \( g \)-modeling retains similar risk guarantees in adaptive settings as it does under independently and identically distributed (i.i.d.) data.  

These two key properties---algorithm independence and the ability to treat each arm as an independent experiment---substantially extend the applicability of EB methods. For instance, \( g \)-modeling can estimate mean treatment effects within a single adaptive experiment, even when the design includes a large number of treatment arms. Such settings are common in practice; e.g., \citet{chapelle2011empirical} describe an application in display advertising where bandit algorithms allocate traffic across 5,910 possible ads, each functioning as a distinct treatment. 

Of course, algorithm independence is not unique to $g$-modeling. As established by the likelihood principle (see Chapter 7 in \citealp{berger2013statistical}), any Bayesian estimation strategy also remains valid regardless of the sampling algorithm. Since $g$-modeling is simply Bayesian estimation with an estimated prior, its posterior estimates naturally inherit this robustness. The novel insight of this paper, however, is that the prior itself can be estimated in an algorithm-independent manner.

To explain why $g$-modeling consistently estimates the prior despite the misspecification of the likelihood, we introduce a novel interpretation of $g$-modeling as a moment-matching procedure. Specifically, we show that $g$-modeling aligns the sample moments of the posterior distribution with those of the estimated prior. At the population level, the prior is identified through moments implied by the law of iterated expectations, which requires the prior moments to equal the expected posterior moments. The moment-matching perspective suggests that $g$-modeling effectively solves for the prior by leveraging the sample analogs of these identifying conditions. Consequently, for parametric prior classes, the validity of $g$-modeling can be established using standard generalized method of moments arguments.

When the class of candidate priors is left unspecified, as in non-parametric maximum likelihood (NPMLE) procedure, we show that $g$-modeling estimates a probability distribution satisfying a self-consistency property: it equals the average posterior distribution corresponding to itself. This is again a sample version of the population requirement that the prior be equal to the expected posterior.

We also formally demonstrate the regret consistency of the NPMLE procedure. This relies on a novel extension of the seminal results of \cite{jiang2009general} and  \cite{jiang2020general}, and may be of independent of interest. 

As is standard in the EB literature, we derive our results on regret consistency assuming Gaussian outcomes. However, we argue that even with general parametric models the Gaussian likelihood naturally emerges as an approximation to the true likelihood within a local asymptotic regime. Effectively, the $g$-modeling framework continues to be valid if we use the sample means of the score functions in place of the sample means of the outcomes.

Beyond treatment effect estimation, knowledge of the prior can improve future decision making and algorithmic design. For instance, \cite{azevedo2020b} utilized EB methods to estimate the prior across multiple static Online Controlled Experiments (OCEs). Their findings revealed that the distribution of treatment effects is typically fat-tailed, leading them to propose a ``lean'' experimentation strategy--conducting a larger number of experiments, each with smaller sample sizes. Our methods allow us to extend this analysis to adaptive experiments.

As an illustration of our methods, we apply the $g$-modeling strategy to the ASOS digital experiments dataset \cite{liu2021datasets}, which comprises 61 adaptive experiments run by a business unit within ASOS, a fashion retail platform. The specific algorithms used to generate this data are proprietary and not publicly disclosed, making this an ideal setting to evaluate the robustness and applicability of our approach. Our results appear to suggest a thick tailed distribution of treatment effects, though some caution is warranted given the small sample of experiments in our analysis.

\subsection{Related literature}
There is a substantial and growing literature on the design of adaptive experiments; see \cite{lattimore2020bandit} and \cite{wassmer2016group} for comprehensive surveys. In contrast, research on the estimation of treatment effects following adaptive experiments remains relatively limited. While the standard sample mean estimator is consistent at parametric rates, it is typically asymptotically biased and fails to be normally distributed. One approach to addressing this issue is inverse propensity weighting, as proposed by \cite{hadad2021confidence}, which restores asymptotic normality. However, this method requires knowledge of the data-generating algorithm and applies only to specific classes of non-deterministic algorithms, excluding widely used methods such as Upper Confidence Bound (UCB) and Bayes-optimal algorithms. 

Alternative approaches, such as those in \cite{deshpande2018accurate} and \cite{nie2018adaptively}, attempt to de-bias the sample mean estimator in adaptive experiments, producing asymptotically normal estimates but often at the cost of increased variance. Meanwhile, a large applied literature eschews frequentist methods in favor of Bayesian estimation. While Bayesian estimators are Bayes-optimal by construction, their performance is highly sensitive to the choice of prior, raising concerns about robustness in practical applications.

We contribute to this literature by developing methods to estimate the prior from the data and extending the Empirical Bayes framework to multiple adaptive experiments, thereby improving both the efficiency and robustness of estimation.

While the foundations of EB trace back to \cite{robbins1951asymptotically}, the methodology has recently experienced a resurgence, particularly in economics, driven by applications in labor economics and related fields. For comprehensive overviews, see \cite{efron2012large}, \cite{walters2024empirical}, \cite{koenker2024empirical} and \cite{koenker_gu_2026}. Beyond treatment effect estimation, the EB framework has broad applicability to various compound decision problems, including multiple testing \citep{efron2012large} and ranking \citep{gu2023invidious}.

\cite{efron2014two} describes $g$-modeling and $f$-modeling approaches to EB estimation. This paper provides a novel interpretation of $g$-modeling as a moment-matching procedure. This perspective is based on \cite{adusumilli2020unobserved}, itself an extension of the seminal work of \cite{neal1998view} on the variational interpretation of the EM algorithm. 

The non-parametric maximum likelihood estimator (NPMLE), first introduced by \cite{kiefer1956consistency}, is a flexible approach to \( g \)-modeling that estimates the prior without imposing parametric assumptions on its form. In recent years, the NPMLE has gained traction, driven by computationally efficient algorithms proposed by \cite{koenker2014convex} and strengthened by new theoretical analyses from \cite{jiang2009general}, \cite{jiang2020general}, and \cite{polyanskiy2020self}. \cite{gilraine2020new} apply this method to non-parametrically estimate distributions of teacher value-added. In this paper, we develop new theoretical results demonstrating that the regret consistency of the NPMLE carries over to adaptive experimental settings.

Our results indicate that standard \( g \)-modeling approaches, which assume a Gaussian likelihood for the sample means, remain valid even for adaptively sampled data, despite the fact that the sample mean may no longer be Gaussian or even a sufficient statistic. Because sample sizes are random, this mis-specified likelihood naturally introduces heteroskedasticity across experiments. Recently, \cite{chen2022empirical} proposed the CLOSE framework for analyzing such heteroskedastic models. However, our findings indicate that CLOSE loses its validity under our setup of compound adaptive experiments. In our setup, sample sizes are independent of the unknown parameters when conditioned on observed data . The breakdown of CLOSE occurs for two reasons. First, CLOSE assumes a specific relationship between parameters and sample sizes, which may not hold under adaptivity. Second, and more critically, CLOSE treats sample sizes as parameters, whereas in adaptive experiments, they are functions of the data. As we demonstrate in Section \ref{subsec:CLOSE}, this misalignment causes CLOSE to apply Bayesian updating twice to the same data, leading to estimates that are insufficiently shrunk.

\section{Compound Adaptive Experiments\protect\label{sec:Setup}}

\subsection{Motivating example: The ASOS Digital Experiments dataset} \protect\label{subsec:Motivating example}

Online Controlled Experiments (OCEs) are web-based randomized controlled trials designed to evaluate and improve digital products and services. These experiments often incorporate adaptive stopping rules, enabling data collection to stop when pre-specified criteria are met. Technology companies routinely conduct thousands of OCEs daily to refine user experiences and optimize product performance. In a typical OCE, website users are randomly assigned to either a control group or one of \( K \) treatment groups. OCEs with \( K = 1 \) are commonly referred to as A/B tests.

Between 2019 and 2020, the global fashion retailer ASOS.com conducted 
$78$ such OCEs, documented in the ASOS Digital Experiments Dataset \citep{liu2021datasets}. This dataset includes daily sample means for the treatment and control groups, alongside the number of samples per treatment arm. The retailer used an adaptive stopping criterion to determine the end of each experiment, but the algorithms used are proprietary and unknown (to us).

For simplicity, we assume each ASOS experiment involves a single treatment arm (i.e., it is an A/B test).\footnote{In reality, 17 experiments included multiple treatments; we drop these experiments, resulting in $n = 61$ experiments in total. See Section \ref{sec:Empirical illustration} for more details.} In each experiment, let $j$ index the periods of experimentation, and in each period a single treatment and control observation is drawn with equal proportion, but we only observe the difference $Y_{j,i} = Y_{j,i}^{(1)} - Y_{j,i}^{(0)}$ between the treated and control observation. The rationale for this assumption is discussed in Section \ref{subsec: Generalizing}. We assume that the outcome differences follow a normal distribution, \( Y_{j,i} \sim \mathcal{N}(\tilde{\theta}_i, \omega_i^2) \), where \( \omega_i^2 \) is known. The Gaussianity assumption is relaxed in Section \ref{subsec:local asymptotics}. 

Let $A_{j,i} \in \{0,1\}$ indicate whether experiment $i$ stops in period $j$, with $A_{j,i} = 1$ denoting a stopping decision. The stopping rule is governed by a policy $\pi_{j,i}$, which maps past information to a probability of stopping. Formally, stopping occurs when $ A_{j,i} = \mathbb{I}\{U_{j,i} < \pi_{j,i}\}$, where $U_{j,i} \sim \textrm{Uniform}[0,1]$ is an exogenous random variable (i.e., independent of all past data) encoding policy randomization.

Define the information set at period $j$ as $\mathcal{I}_{j,i} \equiv \sigma\{ Y_{1,i}, U_{1,i}, \dots, Y_{j-1,i}, U_{j-1,i} \}.$ The stopping policy is thus a function 
\[
\pi_{j,i}: \mathcal{I}_{j,i} \to [0,1].
\]
The policies are implicitly restricted to depend only on past outcome differences and not on the outcome levels themselves. As discussed in Section \ref{subsec: Generalizing}, this restriction holds for most stopping rules commonly used in A/B testing.  

The stopping time for experiment \( i \) is denoted by \( N_i \), and the realized sample from each experiment is given by \( \mathcal{D}_i \equiv \mathcal{I}_{N_i + 1,i} \). The sample mean of the treatment effects is defined as  
\[
\tilde{Z}_i = \frac{1}{N_i} \sum_{j=1}^{N_i} Y_{j,i}.
\]

In the ASOS dataset, the observed sample means are small (on the order of \( 10^{-3} \)), while the realized \( N_i \) values are large (on the order of \( 10^5 \)). Such relative magnitudes are common in online experimentation. \cite{deng2013improving} survey A/B testing practices in digital environments and report that treatment effects often amount to less than 1\% of the expected control outcome. Despite their small size, these effects can generate substantial revenue gains when implemented at scale, motivating companies to conduct large experiments to detect them with statistical precision. 

To facilitate accurate estimation of these small treatment effects, we rescale them as $\theta_i = \sqrt{N} \tilde{\theta}_i$, where
\begin{equation}  \label{eq:definition of N}
N := \frac{1}{n}\sum_{i=1}^n \mathbb{E}[N_i].	
\end{equation}
Correspondingly, we also rescale $N_i$ and $\tilde{Z_i}$ as 
\begin{align} \protect \label{Definition of t_i, Z_i}
	\tau_i & := \delta_\tau(\D_i) = N_i/N, \textrm{ and} \nonumber \\
	   Z_i &:= \delta_z(\mathcal{D}_i)  = \sqrt{N} \tilde{Z}_i.  
\end{align}	   
This ensures that the orders of magnitude of $\tau_i, Z_i$ are stable even as $N$ increases.\footnote{The scaling factor $N$ is used only as a conceptual tool to justify focusing on small treatment effects. Our methods do not depend on knowing its value.} 

To motivate our Empirical Bayes (EB) procedures, we start by discussing the estimation of $\theta_i$ in a Bayesian setting where each $\theta_i$ is an independent random draw from a known prior $G_0$. 
\subsubsection*{The likelihood principle and Bayes estimation}
We can write the likelihood of the observed history $\D_i$ as 
\begin{align}  \protect \label{eq:derivation of likelihood}
&p(\D_i | \theta_i)	 	 = \prod_{j=1}^{N_i}  p\Large(A_{j,i}, U_{j,i}, Y_{j,i}| \mathcal{I}_{j,i}, \theta_i \Large) \nonumber \\
	 	& = \left[\prod_{j=1}^{N_i-1} p(A_{j,i} = 0, U_{j,i}|\mathcal{I}_{j,i}, \theta_i)\right]
	 	  p(A_{N_i,i} = 1, U_{N_i,i}|\mathcal{I}_{N_i,i}, \theta_i)\cdot \prod_{j=1}^{N_i} p(Y_{j,i}|A_{j-1,i} =0, U_{j,i}, \mathcal{I}_{j,i}, \theta_i) \nonumber\\
	 	& = \left[ \prod_{j=1}^{N_i-1} 1\{U_{j,i}\geq \pi_{j,i}  \} \right]1\{U_{N_i,i}< \pi_{N_i,i} \}\cdot
	 	 \prod_{j=1}^{N_i} p(Y_{j,i}|\theta_i) ,
\end{align}  	
	 where the last equality follows from the facts: (1) $A_{j,i}$ is determined by the policy $\pi_{j,i}$ as a function solely of the past history and exogenous randomization $U_{j,i}$ (the actions and the policy can not depend on $\theta_i$ as it is unknown); (2) $U_{j,i}$ is uniformly distributed, implying its density is 1; and (3) the distribution of $Y_{j,i}$ given $\theta_i$ is independent of the past history of observations because, conditional on continuing the experiment, the outcomes are just a random draw from $\N(\theta_i/\sqrt{N}, \omega_i^2)$.

	Making use of the normality of outcomes, we can represent $p(\D_i \vert \theta_i)$ as 
	\begin{equation}\label{likelihoodPrinciple}
	  p(\D_i |\theta_i)= \pi_i(\D_i) \cdot \frac{1}{\sigma_i} \varphi \big(\frac{Z_i - \theta_i}{\sigma_i} \big) 
    \end{equation} 
where $\varphi(\cdot)$ represents the standard normal density, 
$$
\sigma_i^2 := \delta_\sigma(\D_i) = \frac{\omega_i^2}{\tau_i},
$$ and 
	$$
	\pi_i(\D_i) := p(\D_i \vert 0) \cdot  \sqrt{2\pi\sigma_i^2}  \exp \left\{\frac{Z_i^2}{2\sigma_i^2}  \right\}.
	$$
Note that $\pi_i(\D_i)$ is independent of $\theta_i$.

Suppose instead that the samples were drawn exogenously. In a such a scenario, we would have 
$$
Z_i \mid \theta_i \sim \mathcal{N}(\theta_i, \sigma_i^2),
$$
leading to the likelihood 
\[
\frac{1}{\sigma_i} \varphi \bigg(\frac{Z_i - \theta_i}{\sigma_i} \bigg)
\]
for the observed data. We refer to the above as the `working likelihood' for our Empirical Bayes procedures. Equation (\ref{likelihoodPrinciple}) implies that the working likelihood differs from the true likelihood $p(\D_i\vert \theta_i)$ by a multiplicative constant $\pi_i(\D_i)$ that is independent of $\theta_i$. 

Consider the Bayesian approach that assigns a prior \( G_0 \) to \( \theta_i \). Since the true and working likelihoods differ only by a multiplicative factor free from $\theta_i$, using either leads to the same posterior distribution over \( \theta_i \). This is because of the likelihood principle: Bayesian updating depends solely on the observed data and does not require detailed knowledge of the data-generating algorithm.

 This independence from the underlying sampling process extends naturally to estimation. When the goal is to estimate the scaled treatment effect \( \theta_i \) under squared error loss, the optimal Bayes estimator is the posterior expectation, \( \mathbb{E}_{G_0}[\theta_i | \mathcal{D}_i] \), where \( \mathbb{E}_{G_0}[\cdot| \mathcal{D}_i] \) denotes the expectation conditioned on the observed data \( \mathcal{D}_i \) and the prior \( G_0 \). Using (\ref{likelihoodPrinciple}), the posterior expectation of \( \theta_i \) conditional on \( \mathcal{D}_i \) can expressed as:  
\[
\mathbb{E}_{G_0}[\theta_i \mid \mathcal{D}_i] = \frac{\int \theta_i \frac{1}{\sigma_i} \varphi \left(\frac{Z_i - \theta_i}{\sigma_i} \right) \, dG_0}{\int \frac{1}{\sigma_i} \varphi \left(\frac{Z_i - \theta_i}{\sigma_i} \right) \, dG_0}.
\]
We observe that the posterior mean is identical to the one that would be obtained if the data were assumed to have been generated exogenously, irrespective of the actual adaptive nature of the sampling process.

The Bayesian approach to estimation is particularly appealing because it is independent of the data collection algorithm. However, it is highly sensitive to the choice of a subjective prior. To address this limitation, our EB methods, introduced in Section \ref{sec:Empirical Bayes}, offer a data-driven and robust alternative by estimating the prior directly from observed data. We also demonstrate that the $g$-modeling-based EB procedure retains a key advantage of Bayesian estimation: its independence from the data generation algorithm.

Before delving into the EB methodology, we first generalize the problem setup. This is followed by a detailed discussion of the crucial common prior assumption for $(\theta_1, \dots, \theta_n)$.

\subsection{Generalizing the setup: Compound adaptive experiments} \protect\label{subsec: Generalizing}

Consider a setting in which an analyst observes \( n \) experiments, each associated with an unknown parameter \( \theta_1, \dots, \theta_n \). Let \( \mathcal{D}_i \) denote the data generated from experiment \( i \) and \( \mathcal{D} \) the combined data across all experiments. Our methodology applies to any setting where the likelihood takes the form  
\begin{equation}\label{likelihoodPrinciple:generalization}
	  p(\mathcal{D} \mid \theta_1,\dots,\theta_n) = c(\mathcal{D}) \cdot  \prod_{i=1}^n  \frac{1}{\sigma_i} \varphi \left(\frac{Z_i - \theta_i}{\sigma_i} \right),
\end{equation}  
where \( Z_i \) and \( \sigma_i \) are scalar statistics that depend only on \( \mathcal{D}_i \).  

Equation \eqref{likelihoodPrinciple:generalization} implies that the likelihood is proportional to the usual working likelihood for Empirical Bayes methods, where the data from each experiment \( i \) follows a normal model,  
\begin{equation} \label{eq:working_likelihood}
Z_i \mid \theta_i \sim \mathcal{N}(\theta_i, \sigma_i^2),
\end{equation}
with observations across experiments being mutually independent. The formulation \eqref{likelihoodPrinciple:generalization} turns out to be broad enough to encompass nearly all known classes of compound adaptive experiments. Below, we discuss specific examples that illustrate its applicability.

\begin{example}[Multi-arm adaptive experiments]\label{example 1}

Consider a setting in which an analyst observes one or more adaptive experiments, each involving multiple treatment arms. Our methodology treats each treatment arm as an independent experiment, regardless of whether some or all arms belong to the same adaptive experiment. Let \( i \) index the treatment arms, and model the outcomes from arm \( i \) as  $Y_{j,i} \sim \mathcal{N} \left(\tilde{\theta}_i, \omega_i^2\right)$, 
where \( \tilde{\theta}_i \) represents the unknown mean outcome for arm \( i \), and \( \omega_i^2 \) is assumed to be known.

Let \( N_i \) denote the number of times arm \( i \) has been sampled, and define the observed data for arm \( i \) as \( \mathcal{D}_i := \{ Y_{1,i}, \dots, Y_{N_i,i} \} \), the set of all outcomes generated from that arm. Define \( \tilde{Z}_i \) as the sample mean of outcomes from arm \( i \).  

Our framework accommodates a broad range of adaptive sampling algorithms, allowing them to vary across arms. Moreover, the sampling algorithm for a given arm \( i \) may depend on data from other arms or experiments. For instance, if Arms 1 and 2 belong to the same adaptive experiment, data from Arm 1 may influence how often Arm 2 is sampled. Alternatively, if the arms correspond to separate experiments, data from an earlier experiment may inform the algorithm used in a subsequent experiment. Our setup allows for both scenarios. The only substantive restriction is that the choice of sampling algorithms cannot directly depend on the true treatment effects \( \tilde{\theta}_1, \dots, \tilde{\theta}_n \).  

Define $N$ as in (\ref{eq:definition of N}). In practice, \( N \) is typically large because the mean effects \( \tilde{\theta}_i \) are often close to each other, requiring large samples to reliably distinguish between treatment arms. As before, we employ a local-to-zero reparameterization, $\theta_i := \sqrt{N} \tilde{\theta}_i$, to account for such small difference, and also rescale \( N_i \) and \( \tilde{Z}_i \) to obtain the normalized variables \( \tau_i \) and \( Z_i \) as in \eqref{Definition of t_i, Z_i}. Our primary objective is to estimate the scaled treatment effects \( \theta_1, \dots, \theta_n \), postulated to be random draws from a prior \( G_0 \).  


By applying arguments similar to those in \eqref{eq:derivation of likelihood}, we establish that the overall conditional data density follows the structure of \eqref{likelihoodPrinciple:generalization}.

\begin{prop} \protect \label{Prop:likelihood principle}
In the context of compound adaptive experimentation with multiple treatment arms indexed by \( i \), suppose the sampling algorithm for any given arm \( i \) is conditionally independent of \( (\theta_1, \dots, \theta_n) \), given the data from all other arms and an exogenous randomization. Then, the overall conditional data density \( p(\mathcal{D} \mid \theta_1, \dots, \theta_n) \) takes the form  
\begin{equation}\label{likelihoodPrinciple:MAB}
p(\mathcal{D} \mid \theta_1, \dots, \theta_n) = c(\D) \cdot \prod_{i=1}^n \frac{1}{\sigma_i} \varphi \left(\frac{Z_i - \theta_i}{\sigma_i} \right),
\end{equation}  
for some function \( c(\D) \) that does not depend on $\theta_1, \dots, \theta_n$.  
\end{prop}

See Appendix \ref{sec:Appendix:A} for the formal proof. 
\end{example}

\begin{example}[Panel data with missingness and attrition]

Consider a panel data setting where outcomes are modeled as \( Y_{j,i} = \theta_i + \epsilon_{j,i} \), with \(\epsilon_{j,i} \sim \mathcal{N}(0, \omega_i^2)\). Here, the goal is to estimate the parameters \(\theta_i\). Observations may be subject to missingness or attrition, assumed to occur at random given the history of past observations. The probabilities of missingness and attrition are unspecified and can vary across both time periods \(j\) and individuals \(i\). This scenario aligns with the structure of an adaptive multi-arm experiment, implying that Proposition \ref{Prop:likelihood principle} applies to this setting as well.
\end{example}

\begin{example}[Correcting for p-hacking]
P-hacking refers to a range of practices by which researchers selectively report statistically significant results. Here, we focus on specific forms of p-hacking that arise when the experimental protocol is mis-characterized or incompletely reported---even as all relevant data are disclosed.

Examples include optional stopping, where data collection continues until statistical significance is achieved, and selective subgroup analysis, where only significant results from specific subgroups are reported. These forms of p-hacking are equivalent to adaptive experiments with an unknown sampling algorithm, placing them squarely within the scope of our framework.\footnote{Other forms of p-hacking include selectively replacing or failing to disclose relevant data. These possibilities fall outside the remit of our methods.} 

Bayesian methods are robust to selective reporting of this kind, but are sensitive to the choice of prior. Our Empirical Bayes approach enables estimation of the prior from a meta-analysis of potentially p-hacked studies. The resulting Bayes estimators inherently adjust for the distortions introduced by p-hacking.
\end{example} 
\begin{example}[Multiple A/B tests] 

The multiple A/B testing problem, introduced in Section \ref{subsec:Motivating example}, is a special case of multi-arm adaptive experiments in which each experiment consists of exactly two arms: a treatment and a control. While these experiments can therefore be analyzed similarly to Example \ref{example 1}, the assumption that mean effects across arms are drawn from the same prior may be less credible. This is because the control arm typically represents the status quo and is not necessarily exchangeable with the treatment arm. Moreover, in some cases, the same control arm may be used across multiple experiments, further complicating standard exchangeability assumptions.

A complete analysis of multiple A/B testing would require specifying a joint prior on the mean effects of both the treatment and control arms. However, multivariate EB procedures are more complex, and for tractability, we adopt a simpler approach, as in Section \ref{subsec:Motivating example}, based on the following assumptions:
\begin{enumerate}
\item The arms are sampled in equal (or known) proportions.
\item Only the outcome differences between the treatment and control arms are observed in each period.
\item The stopping time depends only on past treatment effect differences.
\end{enumerate}
The first assumption is standard for A/B tests and is known to hold in the ASOS example. The second assumption implies that we set aside information about the mean outcome in each individual arm and instead concentrate solely on the differences between arms. This strategy is standard in EB analysis, even though it entails some loss of efficiency because only part of the available data is used for estimation. It can be supported by an equivariance requirement stating that the decision rule should remain unchanged if a constant is added to all outcomes. The third assumption is unverifiable in the ASOS context, but in practice, almost all A/B testing stopping rules are based solely on the difference in sample means. This includes Wald's seminal sequential probability ratio test (SPRT; \citealt{wald1945statistical}), various group-sequential methods from the clinical trials literature \citep{wassmer2016group}, and the minimax-optimal stopping rule of \cite{adusumilli2022sample}. Moreover, it can be shown that stopping times based only on the difference in sample means form an asymptotically complete class.

Under these assumptions we can simplify the analysis by focusing on the difference in means and placing a prior directly on the treatment effects, as we did in Section \ref{subsec:Motivating example}. It then follows straightforwardly that an analog of Proposition \ref{Prop:likelihood principle} holds, where \( Z_i \) now represents the scaled difference in sample means at the conclusion of experiment \( i \).

\end{example}

\subsection{On the common prior assumption}

Throughout this paper, we adopt the standard Empirical Bayes (EB) assumption that \(\theta_1, \dots, \theta_n\) are i.i.d draws from a common prior. This assumption implies that any specific realization \((\theta_1, \dots, \theta_n)\) is just as likely as any of its permutations \((\theta_{\kappa(1)}, \dots, \theta_{\kappa(n)})\), ensuring exchangeability of the parameter vectors. As \cite{efron2012large} notes, this is closely tied to the question of `comparability': ``\textit{Empirical Bayes methods involve each case learning from the experience of others... To make this believable, the `others' have to be similar in nature to the case at hand, or at least not obviously dissimilar}''. 

The common prior seems particularly natural in the context of analyzing multiple arms within the same adaptive experiment: after all, there is no a priori reason to expect that the ordering of the arms carries any intrinsic significance. However, when considering multiple adaptive experiments, the validity of this assumption depends on the context. For instance, it would be highly questionable to pool adaptive experiments on ad-targeting with those on drug discovery. In the context of the ASOS dataset, the various experiments were run by a single business unit within the company. Moreover, it is reasonable to assume that the algorithms used were independent of the parameters, given the observed data (as required for Proposition \ref{Prop:likelihood principle}). Therefore, in this setting, we consider the common prior assumption to be much more justifiable. 

The term \textit{compound decision problem} refers to a setting in which an experimenter seeks to estimate \(\theta_1, \dots, \theta_n\) by minimizing an average frequentist mean-squared error criterion (i.e., averaged across the $n$ experiments, see Section \ref{sec:Theoretical results}). The compound decision problem is not associated with a prior and is generally different from the EB problem. However, the fundamental theorem of compound decisions \citep{robbins1951asymptotically} establishes that these two problems are equivalent if the experiments are exchangeable---meaning the conditional likelihood for experiment \(i\) has the same form across all experiments---and if the prior \(G_0\) is taken to be the empirical distribution \(n^{-1} \sum_i \delta_{\theta_i}\).  

In our setting, the sampling algorithms are allowed to differ across experiments, so exchangeability of experiments may not hold.\footnote{Note, however, that this is separate from the exchangeability of parameters, as required for Empirical Bayes.} However, if all arms belong to the same adaptive experiment or if identical algorithms were used across experiments with constant outcome variances (\(\omega_i^2 = \omega^2\) for all \(i\)), then exchangeability of  experiments is preserved too, and our methods also solve the compound decision problem.

\section{Empirical Bayes Methodology for Estimating the Prior\protect\label{sec:Empirical Bayes}}

Continuing with the general setup of Section \ref{subsec: Generalizing}, let $G_0$ denote the common prior for $\theta_1,\dots ,\theta_n$. For any given candidate prior \( G \), the marginal density of the data \( \D \) is given by 
$$
 p_{G}(\D) = \int p(\D|\theta_1, \dots, \theta_n)dG^{(n)}(\theta _1,\dots,\theta_n),
$$
where $G^{(n)}(\cdot)$ denotes the product prior over $\theta_1,\dots ,\theta_n$ corresponding to the marginal $G$. Using (\ref{likelihoodPrinciple:generalization}), $p_G(\D)$ can be written as:
\begin{equation} \label{eq:marginal_density}
   p_{G}(\mathcal{D}) = c(\D) \prod_{i=1}^n \int  \frac{1}{\sigma_i} \varphi \bigg(\frac{Z_i - \theta_i}{\sigma_i} \bigg) \, dG(\theta_i ) = c(\D) \prod_{i=1}^n  f_{G, \sigma_i}(Z_i),
\end{equation}
where
\begin{equation} \protect \label{eq:definition_of_f_G}
f_{G, \sigma_i}(Z_i) := \frac{1}{\sigma_i} \int  \varphi \bigg(\frac{Z_i - \theta_i}{\sigma_i} \bigg) \, dG(\theta_i).
\end{equation}
The term \( f_{G, \sigma_i}(Z_i) \) represents the marginal density of \( Z_i \) derived from the working likelihood. 

The core idea of the $g$-modeling approach is to estimate the true prior \( G_0 \) by maximizing the marginal likelihood \( p_{G}(\D) \) over a candidate family of priors \( \mathcal{G} \). Different choices of \( \mathcal{G} \) define different $g$-modeling procedures. For example, \( \mathcal{G} \) could represent the class of all Gaussian priors with mean 0, resulting in a parametric $g$-modeling approach. Alternatively, \( \mathcal{G} \) could remain unrestricted, leading to the non-parametric maximum likelihood estimation (NPMLE) approach. In all cases, the $g$-modeling framework estimates \( G_0 \) as:
\[
\hat{G} = \argmax_{G\in \mathcal{G}}\frac{1}{n}  \ln p_{G}(\D) = \argmax_{G\in \mathcal{G}} \frac{1}{n} \sum_{i} \ln f_{G,\sigma_i}(Z_i).
\]

Importantly, \( f_{G,\sigma_i}(Z_i) \) is independent of the structure of the adaptive experiment. This independence implies that, from a computational perspective, $g$-modeling does not require any knowledge of how the dataset was collected.

Our key insight is that this independence also supports statistical guarantees for $g$-modeling. Under the assumptions stated in Section \ref{sec:Theoretical results}, maximizing the Gaussian working likelihood yields a regret-consistent plug-in Bayes rule even when using the `wrong' marginal likelihood \( \prod_i f_{G,\sigma_i}(Z_i) \), which presumes the data were generated exogenously.

Why does $g$-modeling remain valid under this mis-specification? As shown below, this robustness stems from a novel interpretation of $g$-modeling as a moment-matching procedure. Specifically, we show that $g$-modeling aligns the moments of the posterior with those of the estimated prior. Because the posterior is algorithm independent, this ensures validity of g-modeling in adaptive settings even under a mis-specified likelihood. The moment matching interpretation builds on the work of \cite{neal1998view}, as further developed in \cite{adusumilli2020unobserved}.

\subsection{$g$-modeling as moment matching\protect\label{sec:g modeling moment matching}}

The Donsker-Varadhan variational formula states that 
\begin{equation}\label{DV_formula}
	  \ln f_{G,\sigma_i}(Z_i) = \max_{q_i(\cdot)} \bigg\{ \E_{q_i(\cdot)}\left[\ln \bigg\{ \frac{1}{\sigma_i} \varphi \big(\frac{Z_i - \theta_i}{\sigma_i} \big) \bigg\} \right] - \textrm{KL} \big(q_i(\cdot) \parallel G \big)  \bigg\}, 
 \end{equation}
 where $q_i(\cdot)$ denotes some arbitrary probability distribution and $\textrm{KL}(P \parallel Q)$ denotes the Kullback-Leibler (KL) divergence between two probability measures $P$ and $Q$. The optimal value of $q_i(\cdot)$ in (\ref{DV_formula}) is just the posterior distribution $q^*_{i,G}(\cdot)$ corresponding to the prior $G$ and the working likelihood $
\varphi \big(\frac{Z_i - \theta_i}{\sigma_i} \big)/\sigma_i.
 $ Based on the variational formula, we can rewrite the $g$-modeling optimization problem as
 \begin{align}\label{EM_interpretation}
	  & \max_{G \in \mathcal{G}} \frac{1}{n}\sum_i \ln f_{G,\sigma_i}(Z_i) =\max_{G \in \mathcal{G}} \max_{\{q_i(\cdot)\}_i } \frac{1}{n} \sum_i\bigg\{ \E_{q_i (\cdot)}\left[\ln \bigg\{ \frac{1}{\sigma_i} \varphi \big(\frac{Z_i - \theta_i}{\sigma_i} \big) \bigg\} \right] - \textrm{KL} \big(q_i(\cdot) \parallel G \big)  \bigg\}.
\end{align} 
The two max-operations have an EM interpretation. Conditional on the choice of the `posteriors' $\{q_i(\cdot)\}_i$,  maximization over the prior $G \in \mathcal{G}$ is equivalent to the M-step in an EM algorithm. The maximization over $\{q_i(\cdot)\}_i$ given $G$, which involving computing the posterior corresponding to the prior $G$, is equivalent to the E-step. 

At the $g$-modeling estimate \(\hat{G}\), the E-step and M-step updates constitute a fixed point. At this fixed point, the E-step ensures that the distributions \(\{q_i(\cdot)\}_i\) are the posterior distributions \(\{q^*_{i,\hat{G}}(\cdot)\}_i\) corresponding to \(\hat{G}\). Consequently, the M-step update at this fixed point implies that 
\begin{align}
\hat{G} &= \argmin_{G \in \mathcal{G}}  \frac{1}{n} \sum_{i=1}^{n}\textrm{KL}\left(q^*_{i,\hat{G}}(\cdot)\mid\mid G \right) \nonumber \\
  & =\argmin_{G \in \mathcal{G}}-\int\left(\frac{1}{n}\sum_{i=1}^{n}q^*_{i,\hat{G}}(\theta)\right)\ln g(\theta)d\nu(\theta) \nonumber \\
 & =\argmin_{G \in \mathcal{G}}\textrm{KL}\left(\bar{q}_{\hat{G}}(\cdot)\mid\mid G\right)
 \label{M_step_problem}
\end{align}
where $g(\cdot)$ denotes the density of $G\in \G$ with respect to some dominating measure $\nu$, and 
\[
\bar{q}_G(\cdot):=\frac{1}{n}\sum_{i=1}^{n}q^*_{i,G}(\cdot)
\]
denotes the average posterior given some prior $G$. We now explore the properties of the solution in two distinct settings: (1) when $\G$ is an exponential family, and (2) when $\G$ is unrestricted.

\subsubsection{Exponential family of priors}
For an exponential family of priors \( \mathcal{G} \), the solution to (\ref{M_step_problem}) is identified by matching the moments of the sufficient statistics \( u(\theta) \) between the prior and the average posterior distribution. This property, known as \emph{moment matching}, is a fundamental characteristic of exponential families \citep{bishop2006pattern}. Specifically, this condition implies that  
\[
\mathbb{E}_{\hat{G}}[u(\theta)] = \mathbb{E}_{\bar{q}_{\hat{G}}}[u(\theta)],
\]  
where \( \mathbb{E}_{\hat{G}} \) denotes the expectation with respect to the prior \( \hat{G} \), and \( \mathbb{E}_{\bar{q}_{\hat{G}}} \) denotes the expectation with respect to the average posterior \( \bar{q}_{\hat{G}} \).  

For the Gaussian prior family \( \mathcal{G} = \{\mathcal{N}(0, \gamma^{-1}) : \gamma > 0 \} \), with a zero mean, the sufficient statistic is \( u(\theta) = \theta^2 \). The posterior distribution of \( \theta_i \) given \( \gamma \) and data is  
\[
\theta_i \mid \mathcal{D} \sim \mathcal{N}\left(\frac{Z_i}{1 + \gamma \sigma_i^2}, \frac{\sigma_i^2}{1 + \gamma \sigma_i^2}\right).  
\]  
Moment matching for \( \theta^2 \) between the prior and the average posterior implies that the Empirical Bayes estimate \( \hat{\gamma} \) of the prior precision \( \gamma \) satisfies the equation  
\begin{align}
\protect \label{eq:sample_moment_Gaussian_prior}
\frac{1}{n} \sum_{i=1}^n m(Z_i, \sigma_i; \hat{\gamma}) &=0, \quad \textrm{where}  \\
m(Z_i, \sigma_i; \gamma) &:= \left( \frac{Z_i}{1 + \gamma \sigma_i^2} \right)^2 + \frac{\sigma_i^2}{1 + \gamma \sigma_i^2} - \frac{1}{\gamma} \nonumber
\end{align}  

The above expression is the sample counterpart of the law of iterated expectations involving $\theta_i^2$, which states that the true prior \( G_0 := \mathcal{N}(0, \gamma_0^{-1}) \) satisfies 
$\mathbb{E}_{G_0}\left[\mathbb{E}_{G_0}[\theta_i^2 \mid \mathcal{D}]\right] = \mathbb{E}_{G_0}[\theta_i^2]$. 
 This condition leads to the following necessary population moment condition for the true prior precision \( \gamma_0 \):  
\begin{align}
\protect \label{eq:population_moment_Gaussian_prior}
\mathbb{E}_{G_0^{(n)}}\left[\frac{1}{n} \sum_{i=1}^n m(Z_i, \sigma_i; \gamma_0)\right] = 0,
\end{align}  
where $\mathbb{E}_{G_0^{(n)}}[\cdot]$ denotes the marginal over $\D$ given the product prior $G_0^{(n)}$. Clearly, (\ref{eq:sample_moment_Gaussian_prior}) serves as the method of moments counterpart to the population moment condition (\ref{eq:population_moment_Gaussian_prior}). The population moment is only a necessary score condition. Showing that it vanishes at $\gamma_0$ does not prove uniqueness of its root. Identification instead follows from Lemma \ref{lem:prior-identification} in Appendix \ref{sec:Appendix:A}, which shows that the joint distribution of $(Z_i, \sigma_i)$ identifies the entire mixing distribution and hence identifies $\gamma_0$ in the Gaussian subfamily.

\subsubsection{Unconstrained family of priors}

When the class \(\mathcal{G}\) is unrestricted, as in the NPMLE procedure, (\ref{M_step_problem}) implies that \(\hat{G}\) satisfies a self-consistency property: it equals the average posterior distribution corresponding to itself. Formally, \(\hat{G} = \bar{q}_{\hat{G}}\). 

This expression represents the sample counterpart of the `martingale' property of Bayesian updating, which asserts that the expected posterior must equal the prior. This in turn follows from the law of iterated expectations, which states that for any measurable function \(h(\cdot)\),  $\mathbb{E}_{G_0}\left[\mathbb{E}_{G_0}[h(\theta_i) \mid \mathcal{D}]\right] = \mathbb{E}_{G_0}[h(\theta_i)].$  
Of course this is again a necessary fixed-point characterization, not a standalone identification result. For example, every point mass is unchanged by Bayesian updating when it is itself used as the candidate prior, irrespective of the true data law. Identification is established in Lemma \ref{lem:prior-identification} in Appendix \ref{sec:Appendix:A}.

\subsubsection{Taking stock}

The above discussion explains why \(g\)-modeling is still valid even when we substitute the true marginal likelihood \( p_G(\mathcal{D}) \) with the working marginal likelihood \( \prod_i f_{G,\sigma_i}(Z_i) \). In any adaptive experiment, Bayesian updating ensures that the sequence of posterior distributions forms a martingale (in the space of probability measures). Consequently, the prior equals the expected posterior, even under adaptive stopping; this is a basic informational constraint--Bayes consistency--that every experiment must obey. Viewed this way, \( g \)-modeling amounts to applying a sample-based version of Bayesian consistency, which delivers identifying restrictions for the true prior via the law of iterated expectations. Because the posterior is invariant to the choice of algorithm by the likelihood principle, the moment-matching interpretation implies that \( g \)-modeling remains valid irrespective of the specific algorithm employed.

For parametric families \( \mathcal{G} \), the regret consistency of \( \hat{G} \) is an immediate consequence of standard generalized method-of-moments (GMM) reasoning. Demonstrating regret consistency for the nonparametric maximum likelihood estimator (NPMLE) is more involved, since it corresponds to satisfying an uncountable collection of moment conditions (see Section \ref{sec:Theoretical results} for the rigorous analysis). However, the basic intuition is essentially the same. 

\subsection{Tweedie's formula and the failure of naive $f$-modeling}

The celebrated Tweedie's formula establishes a connection between the posterior expectation of $\theta_i$ and the derivative of the marginal working likelihood. Specifically, recalling the definition of $f_{G,\sigma_i}(Z_i)$ from (\ref{eq:definition_of_f_G}), some straightforward algebra gives
\begin{align} \label{Tweedie's formula}
	\sigma_i^2  \nabla_z\ln f_{G,\sigma _i}(z) \big\rvert_{z=Z_i} &= \frac{\int (\theta_i - Z_i) \varphi \left(\frac{Z_i - \theta_i}{\sigma_i} \right)dG(\theta_i)}{\int \varphi \left(\frac{Z_i - \theta_i}{\sigma_i} \right)dG(\theta_i)} =\E_{G}[\theta_i|\D] - Z_i. 
\end{align}

Tweedie's formula lies at the core of the Empirical Bayes \( f \)-modeling approach, which seeks to estimate the `true' posterior mean \( \mathbb{E}_{G_0}[\theta_i \mid \mathcal{D}] \) using a non-parametric approximation to \( \nabla_z \ln f_{G_0,\sigma_i}(z) \). In the classical exogenous sampling framework, \(f_{G_0,\sigma_i}(z) \) equals the marginal density, $p(Z_i)$, of \( Z_i \), which serves as a sufficient statistic for the data. Consequently, \( \nabla_z \ln f_{G_0,\sigma_i}(z) \) can be non-parametrically estimated  as the derivative of the log-marginal density $\ln p(Z_i)$ of  \( Z_i \). 

This approach, however, breaks down in the context of adaptive experimentation. Under adaptive sampling, $ f_{G_0,\sigma_i}(z)$ no longer equals the true marginal density of $Z_i$ because the true conditional distribution of $Z_i$ given $\theta_i$ is not really $\mathcal{N}(\theta_i, \sigma_i^2)$. In fact, as shown in \cite{adusumilli2021risk}, $Z_i$ even ceases to be a sufficient statistic for the data in adaptive settings. At a minimum, both $Z_i$ and $N_i$ are needed for Bayesian updating. As a result, the derivative of $\ln p(Z_i)$ need not estimate \( \nabla_z \ln f_{G_0,\sigma_i}(Z_i) \), rendering the \( f \)-modeling approach invalid under adaptive sampling.

\subsection{Heteroskedastic models and CLOSE} \protect \label{subsec:CLOSE}
Recently, \cite{chen2022empirical} proposed the CLOSE framework for analyzing heteroskedastic models of the form (\ref{eq:working_likelihood}). This approach assumes a location-scale dependence between \(\theta_i\) and \(\sigma_i\):  
\[
\frac{\theta_i - \E[\theta_i \vert \sigma_i]}{\textrm{sd}[\theta_i \vert \sigma_i]} \Bigg\vert  \sigma_i \sim G, \quad \text{independently of } \sigma_i.
\]
This effectively constrains \(p(\theta_i | \sigma_i)\), the conditional density of \(\theta_i\) given \(\sigma_i\).  

In our setting, however, \(\sigma_i\) depends on the data through \(\tau_i\), implying that \(p(\theta_i | \sigma_i)\) represents a partial posterior, where only a subset of the data (\(\tau_i\)) has been used for Bayesian updating. In adaptive experiments, \(\tau_i\) is determined in a complex manner by \(\theta_i\), so there is no reason to expect this assumption to hold.  

More critically, the CLOSE framework first estimates \(p(\theta_i | \sigma_i)\) and then treats it as a prior for an additional Bayesian update over the data. However, since computing \(p(\theta_i | \sigma_i)\) already involves updating over part of the data, this results in Bayesian updating being applied twice. Repeated updating on the same data will ultimately recover the marginal distribution of \(Z_i\). We would therefore expect CLOSE to lead to under-shrunk in the context of compound adaptive experiments, as compared to standard \(g\)-modeling.\footnote{This concern does not apply to the typical use-case of CLOSE, where $\tau_i$ is determined \textit{prior} to the experiment and may be correlated with $\theta_i$. In that scenario $\sigma_i$ is independent of the data given $\theta_i$.} 

\subsection{Efficiency of $g$-modeling}

Because $g$-modeling is a maximum-likelihood method over the space of priors, we expect it to be efficient when nothing is known about the adaptive algorithms in use. In fact, information about these algorithms would enhance the estimation of $G_0$ only if the algorithms were explicit, known functions of the true $G_0$. In realistic settings, however, characterizing how the algorithms depend on $G_0$ would be cumbersome, if not infeasible. Thus, any loss of efficiency arising from ignorance of the algorithms is likely to be negligible.

\section{Theoretical Results\protect\label{sec:Theoretical results}}

In this section, we examine the theoretical properties of $g$-modeling under two distinct scenarios: (1) when the class of priors is restricted to be Gaussian, and (2) when the class of priors is unrestricted, as in the NPMLE approach. For both scenarios, our primary focus is to evaluate the performance of $g$-modeling methods in terms of the average Bayes risk criterion, comparing their efficacy against an oracle benchmark that has full knowledge of the true prior distribution, $G_0$. 

We start with a formal definition of compound Bayes risk and regret.

\subsection{Compound Bayes risk and regret}

Recall that our setup consists of $n$ experiments with individual treatment effects $\bm{\theta} := (\theta_1, \dots ,\theta_n)$ that we aim to estimate. We view the parameters $\theta_i$ as i.i.d draws from an unknown prior $G_0$.

Let $\delta_i := \delta_{\theta}(\D_i)$ denote an estimator of $\theta_i$, and $\bm{\delta} := (\delta_1, \dots, \delta_n)$ the collection of estimators for each experiment. We define compound frequentist risk as 
$$
R(\bm{\delta}, \bm{\theta}) = \E \left[ \frac{1}{n} \sum_{i=1}^n  \vert \delta_i - \theta_i \vert^2 \ \right\vert \bm{\theta} \bigg].
$$
The compound Bayes risk criterion integrates $R(\bm{\delta}, \bm{\theta})$ over the joint prior $\bm{\theta} \sim G_0^{(n)}$:
$$
R(\bm{\delta}, G_0) = \E_{G_0^{(n)}} \left[ \frac{1}{n} \sum_{i=1}^n  \vert \delta_i - \theta_i \vert^2 \ \right].
$$ 
 Following \cite{jiang2020general},  we evaluate estimators $\bm{\delta}$ using the compound Bayes risk criterion. Consider an oracle who knows the true prior $G_0$. The oracle estimator of $\theta_i$ is clearly $\delta_i^* := \E_{G_0}[\theta_i \vert \D_i]$ and write $\bm \delta^* :=(\delta_1^*, \dots, \delta_n^*)$. The difference in compound Bayes risk between the oracle estimator and the proposed estimator $\bm{\delta}$ is known as regret:
$$
\mathcal{R}(\bm{\delta},G_{0})=R(\bm{\delta},G_{0})-R(\bm{\delta}^{*},G_{0}).
$$
From the form of $\delta_{i}^{*}$, some simple algebra indicates that 
\begin{equation}
\protect \label{eq:regret formula}
\mathcal{R}(\bm{\delta},G_{0}) = \E_{G_0^{(n)}}\left[ \frac{1}{n} \sum_{i=1}^n \left\vert \delta_{i}-\delta_{i}^{*}\right\vert ^{2}\right].
\end{equation}
We call an estimator $\bm{\delta}$ regret consistent if the regret goes to 0 asymptotically as $n\to \infty$. Note that the MLE estimator of $\theta_i$ is given by $\delta_i^\textrm{mle} = Z_i$. 

\subsection{Regret consistency of $g$-modeling with Gaussian priors\protect\label{sec:Theoretical results: Gaussian}}

We now establish results on regret consistency for \( g \)-modeling with a Gaussian prior family. Specifically, consider the Gaussian prior family defined as \(\mathcal{G} = \{\mathcal{N}(0, \gamma^{-1}): \gamma > 0 \}\), and assume that the true prior, \( G_0 \), belongs to this family, with \( G_0 = \mathcal{N}(0, \gamma_0^{-1}) \). The empirical Bayes (EB) estimate, \(\hat{\gamma}\), is then obtained by solving the corresponding sample moment equation provided in \eqref{eq:sample_moment_Gaussian_prior}. The \( g \)-modeling estimator for \(\theta_i\) is the posterior mean under the estimated prior:
\[
\hat{\delta}_i^{\text{EB}} = \frac{Z_i}{1 + \hat{\gamma} \sigma_i^2}.
\]

\subsubsection{Leave-one-out estimation}

We are able to derive a remarkably simple proof of regret consistency when the experiments are independent of each other, i.e., $(Z_i, \tau_i)$ are independent across $i$, and \(\gamma_0\) is estimated using a leave-one-out (LOO) methodology. Let \(\hat{\gamma}^{(-i)}\) denote the leave-one-out estimate of \(\gamma\), computed by excluding the \(i\)-th observation from the sample. The corresponding leave-one-out EB estimator is then denoted:
\[
\tilde{\delta}_i^{\text{EB}} = \frac{Z_i}{1 + \hat{\gamma}^{(-i)} \sigma_i^2}.
\]

To assess the performance of this estimator, we consider its regret ratio relative to the MLE estimator. Assume that experiments are indepdendent of each other,
\begin{align}
\label{eq:regret_ratio}
\frac{\mathcal{R}(\tilde{\bm{\delta}}^{\textrm{EB}},G_{0})}{\mathcal{R}(\bm{\delta}^{\textrm{mle}},G_{0})} & = \nicefrac{\sum_i \E_{G_{0}}\left[ \left\vert \frac{Z_{i}}{1 +\hat{\gamma}^{(-i)}\sigma _i^2 }-\frac{Z_{i}}{1+\gamma_{0}\sigma_i^2 }\right\vert ^{2}\right]}
{\sum_i \E_{G_{0}}\left[ \left\vert Z_i -\frac{Z_{i}}{1+\gamma_{0}\sigma_i^2 }\right\vert ^{2}\right]}
\nonumber \\
 & =\nicefrac{\sum_i \E_{G_{0}}\left[\frac{Z_{i}^{2}\left(\hat{\gamma}^{(-i)}-\gamma_{0}\right)^{2} \sigma_i^4 }{\left(1 +\hat{\gamma}^{(-i)}\sigma _i^2 \right)^{2}\left(1+\gamma_{0}\sigma_i^2 \right)^{2}}\right]}
 {\sum_i \E_{G_{0}}\left[\frac{Z_{i}^{2}\gamma_{0}^{2} \sigma_i^4}{\left(1+\gamma_{0} \sigma_i^2 \right)^{2}}\right]} \nonumber \\
 & \le \nicefrac{\sum_i \E_{G_{0}}\left[\frac{Z_{i}^{2}\left(\hat{\gamma}^{(-i)}-\gamma_{0}\right)^{2} \sigma_i^4 }{\left(1+\gamma_{0}\sigma_i^2 \right)^{2}}\right]}
 {\sum_i \E_{G_{0}}\left[\frac{Z_{i}^{2}\gamma_{0}^{2} \sigma_i^4}{\left(1+\gamma_{0} \sigma_i^2 \right)^{2}}\right]} 
 \le \sup_{1\le i \le n} \E_{G_{0}}\left[\left(\frac{\hat{\gamma}^{(-i)}-\gamma_{0}}{\gamma_{0}}\right)^{2}\right],
 \end{align}
where the first equality is due to (\ref{eq:regret formula}), the inequality is due to $\hat{\gamma}^{(-i)} > 0$, and the last inequality follows from the fact $\hat{\gamma}^{(-i)}$ is independent of $Z_{i},\sigma_{i}$ under the assumption that the experiments are independent of each other.

Remarkably, (\ref{eq:regret_ratio}) is both algorithm-independent and non-asymptotic. Since $\hat{\gamma}$ solves a sample moment condition based on $n$ observations, standard regularity conditions ensure that for all $i$:
\begin{equation}
\label{eq:gamma_bound}
\underset{1 \leq i \leq n}{\sup} \mathbb{E}_{G_0}\left[\left(\frac{\hat{\gamma}^{(-i)} - \gamma_0}{\gamma_0}\right)^2\right] = O(n^{-1}).
\end{equation}
Substituting this result into (\ref{eq:regret_ratio}), we find:
\begin{equation}
\label{eq:regret_bound_gaussian}
\mathcal{R}(\tilde{\bm{\delta}}^{\textrm{EB}}, G_0) = O(n^{-1}) \cdot \mathcal{R}(\bm{\delta}^{\textrm{mle}}, G_0),
\end{equation}
implying that the regret of the leave-one-out EB estimator is a vanishingly small fraction of the MLE regret as $n \to \infty$. 

In fact, in many practical scenarios, the regret of the MLE estimator, 
$$
\mathcal{R}(\bm{\delta}^{\textrm{mle}}, G_0) = \frac{1}{n} \sum_i \E_{G_{0}}\left[\frac{Z_{i}^{2}\gamma_{0}^{2} \sigma_i^4}{\left(1+\gamma_{0} \sigma_i^2 \right)^{2}}\right] < \frac{1}{n} \sum_i \E_{G_{0}}[Z_i^2] 
$$
is finite whenever $n^{-1} \sum_i \mathbb{E}_{G_0}[Z_i^2] = O(1)$. Under this additional moment condition, (\ref{eq:regret_bound_gaussian}) further establishes that the regret of $\tilde{\delta}_i^{\textrm{EB}}$ decays at an $O(n^{-1})$ rate. 

\subsubsection{General results}

The assumption that experiments are independent may be violated in practice, and the computational cost of leave-one-out estimation can also be substantial. Still, even when these conditions fail, we can derive the following bound on the regret of the empirical Bayes (EB) estimator:
\[
\begin{aligned}
\mathcal{R}(\hat{\bm{\delta}}^{\textrm{EB}},G_{0}) 
& =  \frac{1}{n} \sum_{i=1}^n \E_{G_{0}}\!\left[ \left| \frac{Z_{i}}{1 + \hat{\gamma}\sigma_i^2 } - \frac{Z_{i}}{1 + \gamma_0\sigma_i^2 } \right|^2 \right] \le  \E_{G_{0}}\!\left[ \left( \frac{1}{n} \sum_{i=1}^n \frac{Z_i^2}{\sigma_i^4} \right) \left( \frac{1}{\hat{\gamma}} - \frac{1}{\gamma_0} \right)^2 \right].
\end{aligned}
\]
As a result, we can still show that $\mathcal{R}(\hat{\bm{\delta}}^{\textrm{EB}},G_{0}) = O(n^{-1})$, 
albeit under stronger regularity conditions (we skip the formal statement for brevity).

\subsection{Regret consistency of NPMLE}\label{subsec:regret_consistency_NPMLE}

To study the regret consistency of NPMLE, we focus on the setting where the experiments are independent, so that \((Z_i,\tau_i)\) are independent across \(i\). Our treatment of the NPMLE follows \cite{jiang2009general} and \cite{jiang2020general}. Define  
\[
\ell_n(G) = \frac{1}{n} \sum_i \ln f_{G,\sigma_i}(Z_i),
\]
with $\sigma_i^2 := \omega_i^2/\tau_i$ and let the NPMLE of \(G_0\) be given by  
\begin{equation} \label{NPMLE}
\hat{G}_n \in \underset{G \in \mathcal{G}(\mathcal{A}_n)}{\argmax}\; \ell_n(G),
\end{equation}
where \(\mathcal{G}(\mathcal{A}_n)\) is the collection of probability measures supported on a finite, possibly data dependent set \(\mathcal{A}_n \subseteq \mathbb{R}\). In applications, the finiteness of \(\mathcal{A}_n\) facilitates the use of convex optimization \citep{koenker2014convex} to compute the NPMLE.

Let \(\bm{\delta}^{NPEB}\) denote the plug-in estimator of the optimal Bayes rule obtained by substituting the NPMLE \(\hat{G}_n\) for the true prior \(G_0\). Concretely,  
\[
\delta_i^{NPEB}
= \frac{\int \theta \,\frac{1}{\sigma_i}\, \varphi\!\left(\frac{Z_i - \theta}{\sigma_i}\right) d\hat{G}_n}
       {\int \frac{1}{\sigma_i}\, \varphi\!\left(\frac{Z_i - \theta}{\sigma_i}\right) d\hat{G}_n}.
\]

The central theoretical claim of this section is that \(\bm{\delta}^{NPEB}\) achieves regret consistency. To establish this, we impose the assumptions stated below. Let \(\bm{Z} := (Z_1,\dots,Z_n)\) and \(\bm{\tau} := (\tau_1,\dots,\tau_n)\). We write \(\lesssim\) to indicate an inequality that holds up to a multiplicative constant independent of \(n\).

\begin{assumption}\label{boundedG0}
	$\theta_1, \dots, \theta_n \sim _{iid} G_0$, where $G_0$ has bounded support on $[-R,R]$. Furthermore, $(Z_i,\tau_i)$ are independent across $i$.
	\end{assumption}

\begin{assumption}\label{boundVar}
	There exist $\omega_\ell >0$ and $\omega_u < \infty$ independent of $n$ such that $0 < \omega_\ell^2 \le \omega_i^2 \le  \omega_u^2< \infty$ for all $i$. Additionally,
	\[
	\mathbb P\left(
	\underline\tau\leq\tau_i\leq\bar\tau
	\text{ for every }i=1,\ldots,n
	\right)=1.
	\]
	Consequently, with $\sigma_\ell^2:=\frac{\omega_\ell^2}{\bar\tau}$ and $\sigma_u^2:=\frac{\omega_u^2}{\underline\tau}$, the random variances $\sigma_i^2:=\frac{\omega_i^2}{\tau_i}$	satisfy
	\[
	\mathbb P\left(
	\sigma_\ell^2\leq\sigma_i^2\leq\sigma_u^2
	\text{ for every }i=1,\ldots,n
	\right)=1.
	\]
\end{assumption}

\begin{assumption}\label{likfactor}
	The likelihood of the data collected across $n$ experiments takes the form 
	\[
	p(\D | \theta_1, \dots, \theta_n) =  c(\D) \cdot \prod_{i=1}^n \frac{1}{\sigma_i} \varphi\left(\frac{Z_i - \theta_i}{\sigma_i}\right), 
	\]
for some $c(\D)$ independent of $\theta_1, \dots, \theta_n$.  Furthermore, there exist $\bar{c} < \infty$ such that 
$$
\sup_{z_i, \tau_i \in [\underline{\tau}, \bar \tau]} \left\{ \sqrt{2\pi} \sigma_i  e^{z_i^2/2\sigma_i^2} \cdot p(z_i,\tau_i | \theta_i = 0) \right\} \le \bar{c}.
$$
Here $p(z, \tau | \theta)$ denotes the density with respect to $m(z)\otimes\nu(\tau)$, where $m(\cdot)$ denotes the Lebesgue measure, and $\nu(\tau)$ is some probability measure over the support of $\tau$.
\end{assumption}

\begin{assumption}\label{tailbound}
	Let $\bar Z_n = \underset{i}{\max} |Z_i| \vee 1$, then
		$\E[\bar Z_n^4] \lesssim (\ln n)^2$ and $\P(\bar Z_n \geq M_n) \lesssim \frac{1}{n^2}$

with	 $M_n = \sqrt{\kappa \ln n}$ for some $\kappa>0$. 
\end{assumption}

\begin{assumption}\label{likbound}
	Let $q_n = \Big(\frac{e \sqrt{2\pi}}{n^2}\Big) \land 1$ and $\kappa_n = \frac{1}{n} \ln \frac{1}{q_n}$. For all sufficiently large $n$, the NPMLE $\hat G_n$ satisfies, almost surely,
\[
\underset{G \in \mathcal{G}(\mathbb{R})}{\sup} \ell_n(G)  - \ell_n(\hat G_n) \leq \kappa_n.
\]
\end{assumption}



\begin{assumption}\label{support-geometry}
		Let $r_i(z,\tau) = p(z,\tau | \theta_i = 0)$. For every $i$ and for $\nu$-almost every $\tau$, there is an open set
	$\mathcal S_{i,\tau}\subseteq\mathbb R$ such that
	\[
	r_i(z,\tau)>0 \quad\text{for }z\in\mathcal S_{i,\tau},
	\qquad
	r_i(z,\tau)=0
	\quad\text{for Lebesgue-a.e. }z\notin\mathcal S_{i,\tau}.
	\]
The set $\{(z,\tau): z \in \mathcal{S}_{i,\tau}\}$ is measurable. There exists $\ell_*>0$, independent of $n,i,\tau$, such that every
bounded connected component $I$ of $\mathcal S_{i,\tau}$ satisfies
\[
|I|\ge\ell_*.
\]
\end{assumption}

\begin{assumption}\label{density-2}
	Let $r_i(z,\tau) = p(z,\tau | \theta_i = 0)$. There is a fixed integer $m\ge 2$ and constants
	$C_1,\ldots,C_m<\infty$, independent of $n$ and $i$, such that, for
	$\nu$-almost every $\tau$, the function $z\mapsto \ln r_i(z,\tau)$ is
	$m$-times continuously differentiable on each connected component $I$ of $\mathcal{S}_{i,\tau}$, and 
	\begin{equation}\label{eq:smoothness_assumption}
		\left|\partial_z^l\ln r_i(z,\tau)\right|
		\le C_l(1+|z|)^l, \qquad z \in I, 
		\qquad l=1,\ldots,m.
	\end{equation}
	
\end{assumption}

Assumption \ref{boundedG0} stipulates that $G_0$ has compact support.\footnote{This can be relaxed to a tail assumption of $G_0$ at the cost of more burdensome notations for the theoretical results to follow.} Assumption \ref{boundVar} requires the stopping time $\tau_i$ in each experiment to be bounded and bounded away from zero almost surely, uniformly over \(i\), with deterministic bounds independent of \(n\). This condition is met by many algorithms. Moreover, each $\omega_i$ is finite and bounded away from zero, a standard requirement in the empirical Bayes literature; see, e.g., \cite{jiang2020general}, \cite{soloff2024multivariate}, and \cite{chen2022empirical}.

The first part of Assumption \ref{likfactor} reiterates the likelihood principle (\ref{likelihoodPrinciple:generalization}), which we have already established in Proposition \ref{Prop:likelihood principle}. The second part of Assumption \ref{likfactor} is more substantive. It requires that the conditional density of $(Z_i, \tau_i)$ exhibit Gaussian tails when $\theta_i = 0$. Under this condition, Lemma \ref{lem:density_dominance} in Section \ref{sec:Appendix:SB} of the online supplement demonstrates that the marginal density of $(Z_i, \tau_i)$ can be bounded above by a multiplicative constant times the working marginal density which is useful in establishing the result in Theorem \ref{thm:NPMLE-consistency}. 
Lemma \ref{lemma:sufficient_condition_Asm3} in Appendix \ref{appendixC} demonstrates a bound on $\bar c$ when $\tau$ has finite support. 

Assumption \ref{likbound} requires $\hat G_n$ to be an
approximate maximizer of the likelihood over the full class
$\mathcal G(\mathbb R)$. 
It is used twice to establish the theoretical guarantee. It supplies the
true prior likelihood lower bound used to prove Proposition
\ref{HellingerAccuracy} in Section \ref{sec:Appendix:Hellinger} of the online supplement, and it also supplies the approximate generalized
MLE condition for Theorem 5 of \cite{jiang2020general}, both of which channels towards establishing Theorem \ref{thm:NPMLE-consistency}. Lemma \ref{lemma: qualityNPMLE} in the Appendix \ref{appendixC} constructs a data-dependent grid $\mathcal{A}_n$
using which the NPMLE satisfies Assumption \ref{likbound}.

Assumption \ref{support-geometry} places a mild restriction on the connected set over which the conditional density of $(Z_i, \tau_i)$ when $\theta_i=0$, denoted as $r_i(z,\tau)$, is strictly positive. If $r_i(z, \tau)$ is strictly positive in $z$ for $\nu$-almost every $\tau$, then we can take $\mathcal{S}_{i,\tau} = \mathbb{R}$. In this case there are no bounded connected components, so Assumption \ref{support-geometry} always holds. If instead $r_i(z,\tau)$ vanishes on part of the $z$-spce, the assumption prevents any of its bounded connected  components from being arbitrarily short as $n, i,$ or $\tau$ varies. The lower-length conditon is needed to keep a lower-order term in the Gagliardo-Nirenberg interpolation inequality uniformly controlled when applied to each connected sub-component; for details, see the proof of Lemma \ref{GNinequality}.

Assumption \ref{density-2} concerns the smoothness properties of $r_i(z,\tau)$. In particular, it requires polynomial control of its first
$m$ log-derivatives wherever that density is positive. The integer $m$ then governs the rate at which the Bayes regret converges to 0. 

Finally to show all these assumption can hold for adaptive experiments, as a demonstrating example, we verify all the Assumptions for group sequential trials with two stages in Appendix \ref{sec:two-stage-verification}.

Under the above assumptions, we show that the plug-in Bayes estimator \( \bm{\delta}^{NPEB} \) is regret-consistent. 
	
\begin{thm} \label{thm:NPMLE-consistency}
		Under Assumptions \ref{boundedG0}-\ref{density-2}, for some $m \geq 2$, 
		\[
		\mathcal{R}(\bm \delta^{NPEB}, G_0) \lesssim (\ln n) \cdot\Big( \frac{(\ln n)^2}{n} \Big)^{(1-\frac{1}{m})}.
		\]
\end{thm} 

\subsubsection{Proof sketch}
By Proposition \ref{HellingerAccuracy} in Section \ref{sec:Appendix:Hellinger} of the online supplement, which extends the findings of \cite{jiang2009general} and \cite{jiang2020general} to the adaptive framework, we obtain
\begin{equation}\label{eq:Hellinger_rate}
\P\!\left(
\bar h^2(p_{\widehat G_n},p_{G_0})\geq t^2\epsilon_n^2
\right)
\leq3n^{-t^2} \quad \forall\ t\ge 1,
\end{equation}
where $\epsilon_n^2=C_{\epsilon}(\ln n)^2/n$. 

Because the common factor $c_i$ in
$p_{G,i}=c_i f_{G,\sigma_i}$ drops out of the density ratios, Tweedie's
formula, together with the bound $\sigma_i \le \sigma_u$ for all $i$, implies
$$
\mathcal R(\bm\delta^{NPEB},G_0)
=
\E\!\left[
\frac1n\sum_i\sigma_i^4
\left\{
\partial_z\ln \frac{p_{G_0,i}}{p_{\widehat G_n,i}}
(Z_i,\tau_i)
\right\}^2
\right] 
\lesssim 
D_f\left(p_{G_0,i}, p_{\hat G_n,i} \right),
$$
where $D_f(p,q) = E_p[(\nabla \ln(p/q))^2]$ denotes the Fisher divergence between $p$ and $q$. As in \cite{jiang2009general}, \cite{jiang2020general}, and \cite{soloff2024multivariate}, the central task is to relate the Fisher divergence to the Hellinger distance. However, the techniques in the aforementioned papers cannot be applied verbatim here since they crucially rely on properties particular to Gaussian convolutions. In our setting, the marginal density takes the form
\[
p_{G,i}(z,\tau)
=c_i(z,\tau)f_{G,\sigma_i(\tau)}(z).
\]
While the common factor $c_i$ vanishes in the density ratio, it remains present in the Hellinger distance and thereby affects the smoothness properties of $p_{G,i}(z,\tau)$. To handle this, we employ a different and more general approach based on Gagliardo–Nirenberg interpolation inequalities.

Set $u_i=\sqrt{p_{G_0,i}}$ and $v_i=\sqrt{p_{G,i}}$ for some fixed $G$.  Using $\nabla \ln p = 2\, (\nabla \sqrt{p}/\sqrt{p})$ and some straightforward algebra,
\begin{align*}
D_f\left(p_{G_0,i}, p_{G,i} \right) 
 & \lesssim \int\left(\partial_z u_i- \partial_z v_i\right)^{2} 
    + \int\left(\partial_z\ln p_{G_0, i}\right)^{2}\left(u_i-v_i\right)^{2}\\
 & \lesssim \int\left(\partial_z u_i- \partial_z v_i\right)^{2}
    +\left\Vert \partial_z\ln p_{G_0, i}\right\Vert _{\infty}^2 h^{2}(p_{G_0,i},p_{G,i}).
\end{align*}
Assume for simplicity that $\mathcal{S}_{i,\tau}$ in Assumption \ref{support-geometry} coincides with $\mathbb{R}$. Then, the Gagliardo-Nirenberg interpolation inequality states that
\begin{equation}\label{eq:GN-inequality}
\left\Vert \nabla g\right\Vert _{2} \lesssim \left\Vert g\right\Vert _{2}^{1-1/m}\left\Vert \partial^{m}g\right\Vert _{2}^{1/m},
\end{equation}
for any function $g$ and $m \geq 2$. Applying this inequality yields
\begin{equation} \label{eq:Fisher_Hellinger_bound}
\begin{aligned}
D_f\left(p_{G_0,i}, p_{G,i} \right)
& \lesssim \left\Vert \partial^{m}(u_i-v_i)\right\Vert_{2}^{1/m}  [ h^{2}(p_{G_0,i},p_{G,i})]^{1-1/m} \\
& \quad + \left\Vert \partial\ln p_{G_0, i}\right\Vert_{\infty}^2 h^{2}(p_{G_0,i},p_{G,i}).
\end{aligned}    
\end{equation}
Under Assumption \ref{density-2}, we can verify that both $\left\Vert \partial^{m}(u_i-v_i)\right\Vert_{2}$ and $\left\Vert \partial_z\ln p_{G_0, i}\right\Vert _{\infty}$ are bounded up to some $\ln n$ factors.

Suppose for the moment that we were allowed to replace $G$ with $\hat{G}_n$ in (\ref{eq:Fisher_Hellinger_bound}). In that case, combining this substitution with (\ref{eq:Hellinger_rate}) would yield a proof of Theorem \ref{thm:NPMLE-consistency}.  
However, since $\widehat G_n$ is a random quantity estimated from the very same data used to evaluate its score, the deterministic inequality (\ref{eq:Fisher_Hellinger_bound}) cannot be invoked directly. To address this, one must instead rely on covering-number arguments and truncation of $Z_i$ to a compact set, as in \citet[Lemma 10]{soloff2024multivariate}. The detailed steps are provided in the proof of Theorem \ref{thm:NPMLE-consistency} in Section \ref{sec:Appendix:S} of the online suppment.

\subsubsection{Discussion}


The exponent in Theorem \ref{thm:NPMLE-consistency} reflects the smoothness order $m$ appearing in Assumption \ref{density-2}. In the standard
Gaussian normal-means model with fixed exogenous scales, the marginal
density of $Z_i$ is $f_{G,\sigma_i}$. The Gaussian base density satisfies
Assumption \ref{density-2} for any fixed finite $m$. Hence, for every
fixed $m \ge 2$, Theorem \ref{thm:NPMLE-consistency} yields
\[
\mathcal R(\bm\delta^{NPEB},G_0)
\lesssim
n^{-1+1/m}(\ln n)^{3-2/m}.
\]
Therefore, for any fixed $\varepsilon>0$, choosing
 $m$ sufficiently large
(but still fixed) leads to the rate
$O\!\left(n^{-1+\varepsilon}(\ln n)^3\right)$.

In the classical Gaussian setting, \cite{jiang2009general} and
 \cite{soloff2024multivariate} obtain near-parametric NPMLE
regret upper bounds of order $n^{-1}$, up to logarithmic factors, under
compact-support or related tail assumptions. While the formal limit
$m\to\infty$ of our rate has the same $n^{-1}$ order (up to log factors), this limit does not
rigorously follow
from Theorem \ref{thm:NPMLE-consistency}, since the
constants 
$\{C_l\}_{l=1}^m$ in Assumption \ref{density-2}, as well as those
entering the Gagliardo--Nirenberg interpolation inequalities, may depend
on $m$.
Obtaining a bound that is uniform over $m \in [2,\infty)$, and
clarifying whether the
rate for finite $m$ is optimal, is left to future
research.


\section{Extensions}

\subsection{Parametric models and local asymptotics}\protect \label{subsec:local asymptotics}

Thus far, our $g$-modeling framework has relied heavily on the assumption of Gaussian outcomes. However, as we demonstrate below, the Gaussian likelihood naturally emerges as an approximation to the true likelihood within a local asymptotic regime.

Consider the multi-arm sequential experiment setup described in Section \ref{subsec: Generalizing}, but now assume that the outcomes from each arm follow a parametric distribution:  
\[
Y_{j,i} \sim P_{\mu_i},
\]  
where \( P_{\mu} \) is a known parametric family and \( \mu_i \) is a scalar experiment-specific parameter. For instance, in A/B testing applications, outcomes are often binary, in which case \( Y_{j,i} \sim \text{Bernoulli}(\mu_i) \). As before, let \( \tau_i \) denote the proportion of times arm \( i \) was sampled, relative to the total sample size \( N \).

Following the standard local asymptotic framework, we assume that each \( \mu_i \) represents a local perturbation of a reference parameter \( \mu_0 \):  
\[
\mu_i = \mu_0 + \frac{\theta_i}{\sqrt{N}}.
\]
Furthermore, as in the earlier sections, we also suppose that \( \theta_1, \dots, \theta_n \) are independent and identically distributed (i.i.d.) draws from a common prior \( G_0 \).

We emphasize that the reparametrization of \( \mu_i \) in terms of \( \theta_i \) serves primarily to provide a theoretical justification for our procedures. In practice, knowledge of \( \mu_0 \) is not required, as placing a prior on \( \mu \) is equivalent to placing a prior on \( \theta \).

To ensure a well-posed asymptotic analysis, we impose the standard assumption that the parametric family \( P_{\mu} \) is quadratic mean differentiable (QMD) around the reference point \( \mu_0 \). This condition guarantees the existence of a well-defined score function \( \psi(Y_{j,i}) \) and an associated inverse Fisher information given by  $\omega_i^2 := \mathbb{E}_0[\psi^2]^{-1}$, where \( \mathbb{E}_0[\cdot] \) denotes expectation under \(P_0 := P_{\mu_0} \). 

We define the (normalized) average score in experiment \( i \) as  
\[
Z_i = \frac{\sqrt{N} \, \omega_i^2 }{\lfloor N \tau_i \rfloor} \sum_{j=1}^{\lfloor N \tau_i \rfloor} \psi(Y_{j,i}).
\]  
Let \( p(\mathcal{D} \mid \theta_1, \dots, \theta_n) \) denote the likelihood of the data \( \mathcal{D} \) given the parameter vector \( (\theta_1, \dots, \theta_n) \). By Lemma 2 in \cite{adusumilli2021risk}, under the QMD property and provided that \( \max_i \tau_i \) is bounded, the likelihood ratio of \( (\theta_1, \dots, \theta_n) \) relative to the reference \( (0,\dots,0) \) is approximated as  
\[
\ln \frac{p(\mathcal{D} \mid \theta_1, \dots, \theta_n)}{p(\mathcal{D} \mid 0, \dots, 0)} = \sum_i \left\{\frac{1}{\sigma_i^2} Z_i \theta_i  - \frac{1}{2\sigma_i^2} \theta_i^2 \right\} + o_{P_0}(1),
\]
where \( \sigma_i^2 := \omega_i^2/\tau_i \). 

Consequently, up to an asymptotically negligible error, the log-likelihood can be expressed as
\begin{align*}
\ln p(\mathcal{D} \mid \theta_1, \dots, \theta_n)  &= \ln c(\mathcal{D}) + \sum_{i=1}^n \ln \left\{ \frac{1}{\sigma_i} \varphi \left(\frac{Z_i - \theta_i}{\sigma_i} \right) \right\}  + o_{P_0}(1), \quad \text{where} \\
c(\mathcal{D}) &:= p(\mathcal{D} \mid 0, \dots, 0) \cdot \prod_{i=1}^n \sqrt{2\pi \sigma_i^2} \exp \left\{\frac{ Z_i^2}{2\sigma_i^2} \right\}.
\end{align*}
Thus, up to a proportionality constant that does not depend on \( (\theta_1, \dots, \theta_n) \), the log-likelihood of the data \( p(\mathcal{D} \mid \theta_1, \dots, \theta_n) \) is well-approximated by the working log-likelihood  
\[
\sum_{i=1}^n \ln \left\{\frac{1}{\sigma_i} \varphi \left(\frac{Z_i - \theta_i}{\sigma_i} \right) \right\}.
\]
This indicates that our $g$-modeling framework remains fundamentally unchanged in the parametric setting, provided \( Z_i \) is redefined as the average score in each experiment \( i \). 

\subsection{Alternative loss functions}

While our theoretical results have focused on the MSE loss function, our methodology easily generalizes to other losses. For instance, \cite{chen2022empirical} describes a loss function geared towards selecting $\theta_i$ that are larger than 0:
$$
L(\bm{\delta}, \bm{\theta}) = \frac{-1}{n} \sum_{i=1}^n \delta_i \theta_i.
$$
Here $\bm{\delta} := (\delta_1, \dots, \delta_n)$ and each $\delta_i \in \{0,1\}$ is a binary action, with $\delta_i = 1$ indicating selection. For a given prior $G$, the optimal Bayes decision is $\delta_i = \mathbb{I}\{\theta_{i,G} \ge 0\}$, where $\theta_{i,G}$ is the posterior mean of $\theta_i$ under that prior. The Empirical Bayes (EB) analogue of this strategy would be to replace $G$ with its EB estimate $\hat{G}$; in terms of our previous notation, this corresponds to using $\hat{\delta}_i^{\textrm{EB}}$ or $\hat{\delta}_i^{\textrm{NPMLE}}$ in place of $\theta_{i,G}$, depending on which $g$-modeling procedure is employed.  

The selection problem could be generalized to top-$m$ selection, where the aim is to determine the best $m$ values of $\theta_i$. Again, the optimal Bayes decision, given a prior $G$, would be to select the top $m$ values based on the posterior means $\theta_{i,G}$. The EB strategy again simply replaces $G$ with $\hat{G}$. 

For threshold selection, if $\widehat\theta_i$ and
$\theta_{i,G_0}$ denote the estimated and oracle posterior means,
the conditional regret of the plug-in decision $\mathbb{I}\{\hat{\theta}_i \ge 0\}$ is bounded by
\[
\frac1n\sum_i
|\theta_{i,G_0}|\,
1\{\operatorname{sign}(\widehat\theta_i)
\ne\operatorname{sign}(\theta_{i,G_0})\}
\leq
\frac1n\sum_i|\widehat\theta_i-\theta_{i,G_0}|.
\]
Taking expectations and invoking the Cauchy–Schwarz inequality shows that this quantity is upper bounded by the square root of the associated MSE regret. For top-$m$ selection, the reasoning in \cite{chen2022empirical} leads to an analogous result. Consequently, our regret guarantees for MSE loss, established in earlier sections, extend to these alternative loss functions as well. 

\section{Simulations}

We now present simulation results evaluating the performance of our methods for treatment effect estimation using adaptively generated data from common algorithms. For simplicity, we focus on one-armed bandit experiments employing Thompson Sampling (TS) and the Upper Confidence Bound (UCB) algorithm, and we compare various shrinkage estimators. The maximum sample size of each experiment is set to be 50. Each experiment yields a pair of summary statistics \((Z_i, \sigma_i)\), corresponding to the sample mean and sample standard deviation.  

In the first data-generating process (DGP), the prior distribution of \(\theta_1, \dots, \theta_n\) is assumed to be Gaussian: \(\theta_i \sim N(m_0, s_0^2)\) with \(m_0 = 0\) and \(s_0^2 = 1/4\). The outcome distribution for each experiment follows a Gaussian with mean \(\theta_i\) and unit variance. This setup favors linear shrinkage estimators.  Table \ref{tab:gaussianMSE} reports the mean squared error (MSE) for estimating \(\theta_i\). The Oracle assumes perfect knowledge of \(G_0\) when constructing the Bayes rule, while NPMLE estimates \(G_0\) non-parametrically. Several intermediate linear shrinkage estimators are also considered: 
\begin{itemize}
    \item L-Marginal estimates the prior parameters \((m_0, s_0^2)\) using marginal moments of \((Z_i, \sigma_i)\). The linear shrinkage estimator takes the form:  
  \[
  \hat{m}_0 + (Z_i - \hat{m}_0) \frac{\hat{s}_0^2}{\sigma_i^2 + \hat{s}_0^2}.
  \]  
  This corresponds to the standard empirical Bayes approach (e.g., \cite{kane2008does}, \cite{chetty2014measuringa}), where \((m_0, s_0^2)\) is estimated from the marginal moments of \((Z_i, \sigma_i)\).\footnote{A commonly used estimator is \(\hat{m}_0 = \frac{1}{n} \sum_i Z_i\) and \(\hat{s}_0^2 = \widehat{\text{Var}}(Z_i) - \hat{\E}(\sigma_i^2)\), based on the identities \(\E[Z_i] = \E[\theta_i] = m_0\) and \(\text{Var}[Z_i] = \E[\text{Var}[Z_i | \theta_i]] + \text{Var}[\E[Z_i | \theta_i]]\).}   

  \item L-Posterior follows the same linear shrinkage formula but estimates \((m_0, s_0^2)\) via posterior moment matching.

  \item L-LOO is similar to L-Posterior, except that prior parameter estimation is performed using a leave-one-out approach, as discussed earlier. 
\end{itemize}  
For comparison, we also report the performance of the MLE estimator for \(\theta_i\) as a benchmark. 

Under the Gaussian prior, both L-Posterior and L-LOO closely track the performance of the Oracle estimator. NPMLE is also highly competitive, particularly for larger sample sizes. While its performance lags behind for small samples (\(n = 100\)), its performance catches up with the L-posterior as $n$ increases.  In contrast, L-Marginal performs worse due to a mismatch between the working likelihood and the true marginal density of \((Z_i, \sigma_i)\). As a result, marginal moments of \((Z_i, \sigma_i)\) fail to provide consistent estimates of the prior parameters.  

Table \ref{tab:momentGaussian} further illustrates this point by reporting the mean squared error (MSE) for estimating \((m_0, s_0^2)\) under the various methods. The columns labeled ``Marginal'' correspond to estimates based on marginal moments, which exhibit high MSE due to significant bias. In contrast, the ``Posterior'' approach, which estimates \((m_0, s_0^2)\) via posterior moment matching, and NPMLE, which reports the mean and variance of \(\hat{G}_n\), both achieve substantially lower MSE across all algorithms. 

\begin{table}
  \centering
  \small
  \caption{One armed bandit, Gaussian prior} \label{tab:gaussianMSE}
  \renewcommand{\arraystretch}{1.5} 
  \begin{tabular}{rrrrrrrr}
    \hline
    n & Oracle & NPMLE & L-marginal & L-loo & L-posterior & MLE \\
    \hline
    \multicolumn{7}{c}{Thompson Sampling }\\
    $100$  & 0.0560 & 0.0686 & 0.0646 & 0.0574 & 0.0572 & 0.1744 \\
    $500$  & 0.0561 & 0.0605 & 0.0641 & 0.0564 & 0.0564 & 0.1790 \\
    $1000$ & 0.0562 & 0.0586 & 0.0637 & 0.0564 & 0.0564 & 0.1770 \\
    $5000$ & 0.0562 & 0.0570 & 0.0635 & 0.0563 & 0.0563 & 0.1775 \\
    \hline
    \multicolumn{7}{c}{UCB algorithm }\\
    $100$  & 0.0605 & 0.0717 & 0.0720 & 0.0623 & 0.0622 & 0.1959 \\
    $500$  & 0.0605 & 0.0647 & 0.0715 & 0.0608 & 0.0608 & 0.1976 \\
    $1000$ & 0.0605 & 0.0628 & 0.0716 & 0.0606 & 0.0606 & 0.1986 \\
    $5000$ & 0.0607 & 0.0613 & 0.0716 & 0.0607 & 0.0607 & 0.1983 \\
    \hline
  \end{tabular}\\
  \begin{flushleft}
  \setstretch{1.2}
  \footnotesize{Notes: One-armed bandit experiments. $\theta_i \sim N(0, \tfrac{1}{4})$. Mean squared error for $\theta_i$ estimation. Results are based on 500 simulation repetitions.}
  \end{flushleft}
\end{table}

\begin{table}
    \centering
    \small
    \caption{One armed bandit, Gaussian prior} \label{tab:momentGaussian}
    \renewcommand{\arraystretch}{1.5} 
	\begin{tabular}{rrrr|rrr}
		\hline
	n	& Marginal & Posterior & NPMLE & Marginal & Posterior & NPMLE \\ 
		\hline
			\multicolumn{7}{c}{MSE of estimator for Prior Mean}\\
		\hline 
		\multicolumn{1}{c}{}& \multicolumn{3}{c}{Thompson Sampling }& \multicolumn{3}{c}{UCB algorithm}\\
		\hline 
		100 & 0.0222 & 0.0038 & 0.0039 & 0.0336 & 0.0036 & 0.0038 \\ 
		500 & 0.0195 & 0.0007 & 0.0007 & 0.0297 & 0.0007 & 0.0008 \\ 
		1000& 0.0185 & 0.0004 & 0.0004 & 0.0298 & 0.0004 & 0.0004 \\ 
		5000 & 0.0181 & 0.0001 & 0.0001 & 0.0295 & 0.0001 & 0.0001 \\ 
		\hline \hline 
	\multicolumn{7}{c}{MSE of estimator for Prior Variance}\\
	\hline 
	\multicolumn{1}{c}{}& \multicolumn{3}{c}{Thompson Sampling }& \multicolumn{3}{c}{UCB algorithm}\\
	\hline 
		100 & 0.0218 & 0.0023 & 0.0048 & 0.0203 & 0.0027 & 0.0048 \\ 
		500 & 0.0190 & 0.0005 & 0.0009 & 0.0185 & 0.0005 & 0.0010 \\ 
		1000 & 0.0179 & 0.0002 & 0.0005 & 0.0180 & 0.0003 & 0.0005 \\ 
		5000 & 0.0176 & 0.0000 & 0.0001 & 0.0176 & 0.0001 & 0.0001 \\ 
		\hline\hline 
	\end{tabular}\\
    \begin{flushleft}
    \setstretch{1.2}
    \justify \footnotesize{Notes: One-arm bandit experiments with Thompson Sampling or UCB algorithm. Mean-squared error for the estimation of prior mean and variances. The true prior is $G_0 = N(0, 1/4)$. `Marginal' corresponds to learning the moments of $G_0$ based on $(Z_i,\sigma_i)$, `Posterior' corresponds to learning the moments of $G_0$ through posterior moment matching, `NPMLE' corresponds to the moments implied by NPMLE estimate. Results are based on 500 simulation repetitions.}
    \end{flushleft}
\end{table}

For the second DGP, we consider a discrete $G_0$ with two atoms. This is a setting that naturally favors NPMLE. Methods assuming a Gaussian prior in this case fail to recover the correct posterior distribution and therefore cannot replicate the optimal Bayes rule. We should emphasize that Stein's result - that MLE is inadimissible - no longer holds under adaptive sampling. Nevertheless, we observe that the MLE still exhibits higher MSE compared to linear shrinkage methods. Among all the estimators, NPMLE clearly demonstrates the best performance and also tracks the Oracle very closely. It mimics the optimal Bayes rule and is able to learn $G_0$ with remarkable accuracy even in moderate sample sizes. The detailed results are reported in Tables \ref{tab:twoGroupMSE} and \ref{tab:momentTwo_point}.

\begin{table}
\centering
\small
    \caption{One armed bandit, two-point prior} \label{tab:twoGroupMSE}
    \renewcommand{\arraystretch}{1.5}
	\begin{tabular}{rrrrrrr}
		\hline
		& Oracle & NPMLE & L-marginal & L-loo & L-posterior & MLE \\ 
			\hline 
						\multicolumn{7}{c}{Thompson Sampling  }\\
		$n = 100$ & 0.0000 & 0.0139 & 0.1267 & 0.1227 & 0.1243 & 0.2096 \\ 
		$n = 500$& 0.0000 & 0.0029 & 0.1222 & 0.1196 & 0.1199 & 0.2037 \\ 
		$n = 1000$ & 0.0000 & 0.0016 & 0.1227 & 0.1202 & 0.1204 & 0.2047 \\ 
		$n = 5000$ & 0.0000 & 0.0004 & 0.1228 & 0.1205 & 0.1205 & 0.2047 \\ 
		\hline
			\multicolumn{7}{c}{UCB algorithm }\\
		$n = 100$  & 0.0000 & 0.0104 & 0.1039 & 0.1013 & 0.1023 & 0.1654 \\ 
			$n = 500$ & 0.0000 & 0.0032 & 0.1049 & 0.1030 & 0.1033 & 0.1683 \\ 
		$n = 1000$ & 0.0000 & 0.0018 & 0.1050 & 0.1032 & 0.1033 & 0.1684 \\ 
			$n = 5000$  & 0.0000 & 0.0004 & 0.1046 & 0.1029 & 0.1029 & 0.1682 \\ 
		\hline
	\end{tabular}
    \begin{flushleft}
    \setstretch{1.2}
    \justify \footnotesize{Notes: One-arm bandit experiments. $\theta_i \sim \frac{1}{2}\delta_{-1} + \frac{1}{2} \delta_3$. Mean squared error for $\theta_i$ estimation. Results are based on 500 simulation repetitions.}
    \end{flushleft}
\end{table}

\begin{table}
\centering
\small
\caption{One armed bandit, two-point prior} \label{tab:momentTwo_point}
    \renewcommand{\arraystretch}{1.5}
	\begin{tabular}{rrrr|rrr}
			\hline
	n	& Marginal & Posterior & NPMLE & Marginal & Posterior & NPMLE \\ 
	\hline
	\multicolumn{7}{c}{MSE of estimator for Prior Mean}\\
	\hline 
	\multicolumn{1}{c}{}& \multicolumn{3}{c}{Thompson Sampling }& \multicolumn{3}{c}{UCB algorithm}\\
	\hline 
			100 & 0.0662 & 0.0529 & 0.0451 & 0.0582 & 0.0476 & 0.0431 \\ 
		500 & 0.0205 & 0.0102 & 0.0077 & 0.0212 & 0.0100 & 0.0083 \\ 
		1000 & 0.0169 & 0.0064 & 0.0042 & 0.0167 & 0.0054 & 0.0041 \\ 
		5000 & 0.0134 & 0.0029 & 0.0008 & 0.0122 & 0.0017 & 0.0009 \\ 
				\hline \hline 
		\multicolumn{7}{c}{MSE of estimator for Prior Variance}\\
		\hline 
		\multicolumn{1}{c}{}& \multicolumn{3}{c}{Thompson Sampling }& \multicolumn{3}{c}{UCB algorithm}\\
		\hline 
		100 & 0.2578 & 0.1781 & 0.0366 & 0.1567 & 0.1181 & 0.0348 \\ 
		500 & 0.2395 & 0.1579 & 0.0056 & 0.1614 & 0.1146 & 0.0061 \\ 
		1000 & 0.2426 & 0.1600 & 0.0028 & 0.1617 & 0.1144 & 0.0032 \\ 
		5000 & 0.2426 & 0.1601 & 0.0005 & 0.1642 & 0.1149 & 0.0006 \\ 
		\hline
	\end{tabular}
    \begin{flushleft}
    \setstretch{1.2}
    \justify \footnotesize{Notes: One-armed bandit with Thompson Sampling or UCB algorithm. Mean-squared error for the estimation of prior mean and variances, when true prior is $G_0 = \frac{1}{2} \delta_{-1} + \frac{1}{2}\delta_3$. `Marginal' corresponds to learning the moments of $G_0$ based on $(Z_i,\sigma_i)$, `Posterior' corresponds to learning the moments of $G_0$ through posterior moment matching assuming normal prior, `NPMLE' corresponds to the moments implied by NPMLE. Results are based on 500 simulation repetitions. }
    \end{flushleft}
\end{table}

\section{Empirical Illustration \protect \label{sec:Empirical illustration}}

In this section, we utilize the $g$-modeling framework to estimate the prior distribution and treatment effects using data from the ASOS digital experiments dataset. This dataset comprises results from $n = 61$ A/B tests conducted under adaptive stopping rules. In each experiment, the treatment and control groups were sampled in equal proportions. When analyzing these experiments using NPMLE, we make the implicit assumption that the experiments are independent of each other, so that Assumption 1 in Section \ref{subsec:regret_consistency_NPMLE} holds. 

The relevant terminology for this analysis is defined in Section \ref{subsec:Motivating example}. In this context,
$$
Z_i = \sqrt{N}  \left(\bar{Y}^{(1)}_i - \bar{Y}^{(0)}_i \right)
$$ 
represents the difference in sample means between the treatment and control groups, scaled by the mean sample size $\sqrt{N}$, while $\tau_i$ represents the number of times each arm was sampled divided by $N$. Also, 
$$\sigma_i^2 = \frac{\textrm{Var}[Y_{j,i}^{(1)}] + \textrm{Var}[Y_{j,i}^{(0)}]}{\tau_i}.$$ 

A key distinction from the discussion in Section \ref{subsec:Motivating example} is that the outcomes \( Y_{j,i}^{(1)} \) and \( Y_{j,i}^{(0)} \) are binary, following a Bernoulli distribution, rather than the normal distribution assumed in that section. Nevertheless, as demonstrated in Section \ref{subsec:local asymptotics}, the working likelihood  $\prod_i \frac{1}{\sigma_i} \varphi\left(\frac{Z_i - \theta_i}{\sigma_i}\right)$  
remains a valid approximation to the true likelihood, up to a proportionality constant that does not depend on the parameters \( \theta_1, \dots, \theta_n\). 

Our objective is to estimate the scaled experiment-specific treatment effects 
\[
\theta_i = \sqrt{N} \cdot \E[Y_{j,i}^{(1)} - Y_{j,i}^{(0)}],
\]  
along with the common prior distribution, $G_0$, from which these effects are generated. For this purpose, we implement the $g$-modeling procedures outlined in Section \ref{sec:Empirical Bayes}. 

These procedures rely on knowledge of $\sigma_i^2$, which is a function of $\textrm{Var}[Y_{j,i}^{(1)}]$ and $\textrm{Var}[Y_{j,i}^{(0)}]$, respectively. Although these variances are unknown, they can be approximated using standard sample variance estimators. Notably, these estimators remain consistent in adaptive experimental settings, even though their rates of convergence are slower compared to classical experiments.

Figure \ref{fig: histogram} plots the histogram of $Z_i$, the raw estimates of scaled estimated treatment effects. Most of the treatment effect estimates are fairly small, even after scaling, although we do see 6 experiments with noticeably larger positive treatment effects. The stopping time across the 61 experiments also is heterogeneous and its relationship with $Z_i$ is plotted in the left panel of Figure \ref{fig: close}. We observe a slight positive correlation. Experiments with a largest $Z_i$ tends to have larger $\sigma_i$, which is evidence of earlier stopping. Using our proposed method, we estimate $G_0$ using the NPMLE with the prescribed working model for $Z_i$ and the estimator is presented in the right panel of Figure \ref{fig: close}. Interestingly, the four experiments with the largest $Z_i$ are all given their own mass points. The rest of the mass points cluster around zero.

The estimator of $\hat G_n$ drives the shrinkage pattern for the nonparametric EB estimator. Figure \ref{fig:shrinkage} shows the amount of shrinkage instructed by the nonparametric along with the linear empirical Bayes method. The two methods agree on experiments with more extreme $Z_i$'s. The four experiments associated with the largest $Z_i's$ are not shrunk much by either method.  Two experiments with $Z_i$'s close to 5 are shrunk aggressively, more so under the nonparametric method to the extent that their relative rankings are also changed. A close inspection shows these two experiments are associated with extremely large variances. On the other hand, experiments with their effects close to zero are adjusted more towards zero under nonparametric EB, while the linear method aligns them closely on the 45 degree line, hence leaving the MLE unrevised.

\begin{figure} 
\includegraphics[scale = 0.35]{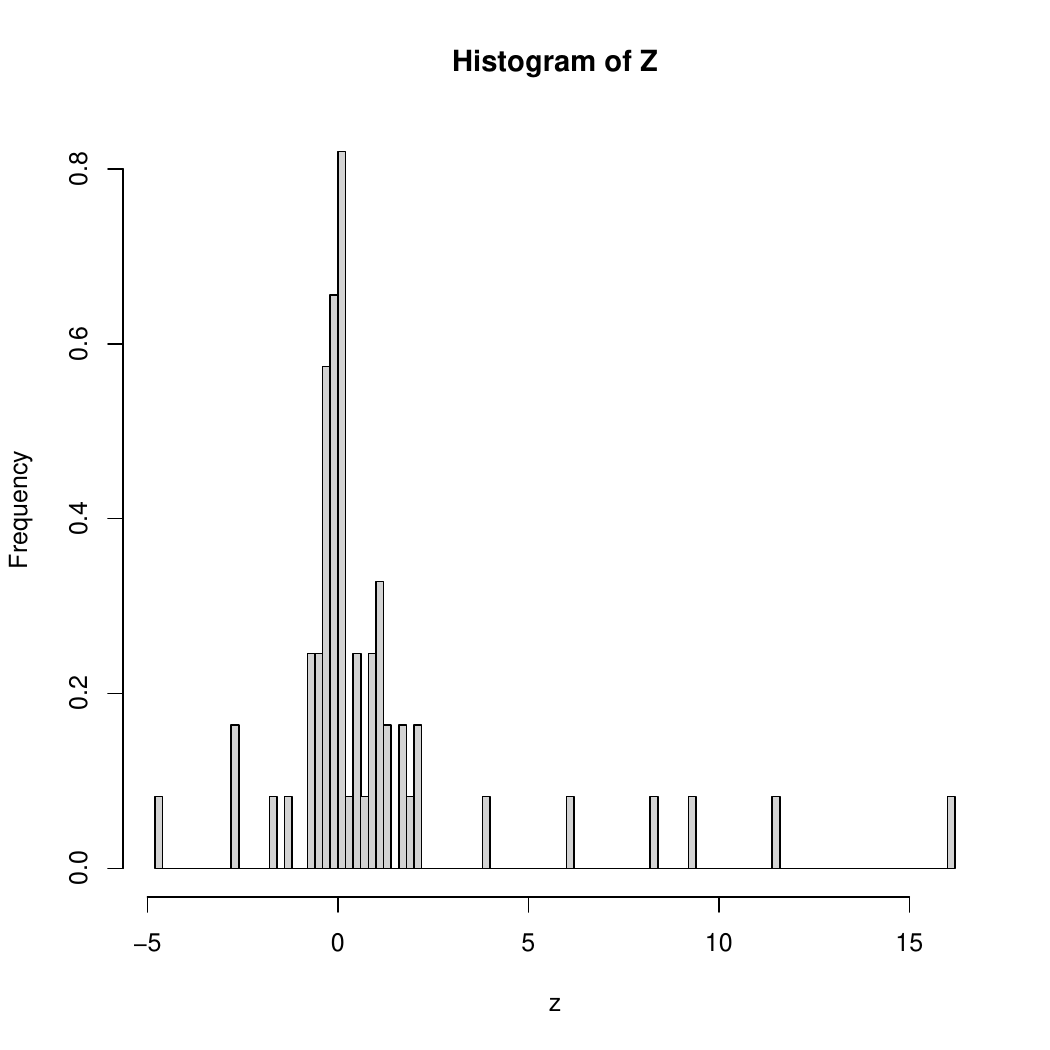}
\caption{Histogram of $Z_i$ across 61 A/B tests}
\label{fig: histogram}
\end{figure}

\begin{figure}
\includegraphics[scale = 0.45]{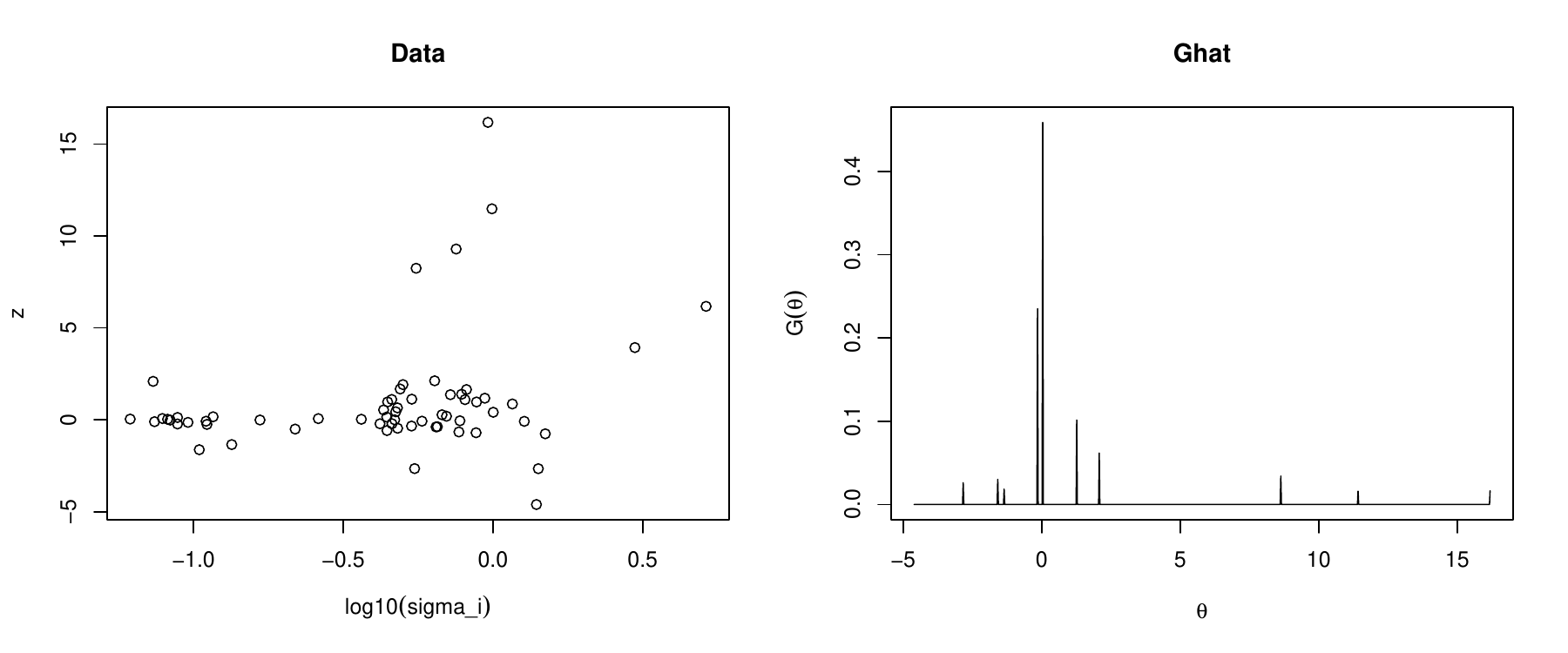}
\caption{The left panel plots the $\log_{10}(\sigma_i)$ on the x-axis and $Z_i$ on the y-axis. The right panel plots the NPMLE $\hat G_n$ of $G_0$ using the working likelihood model for $Z_i$.}
\label{fig: close}
\end{figure}
	
\begin{figure}
\includegraphics[scale = 0.45]{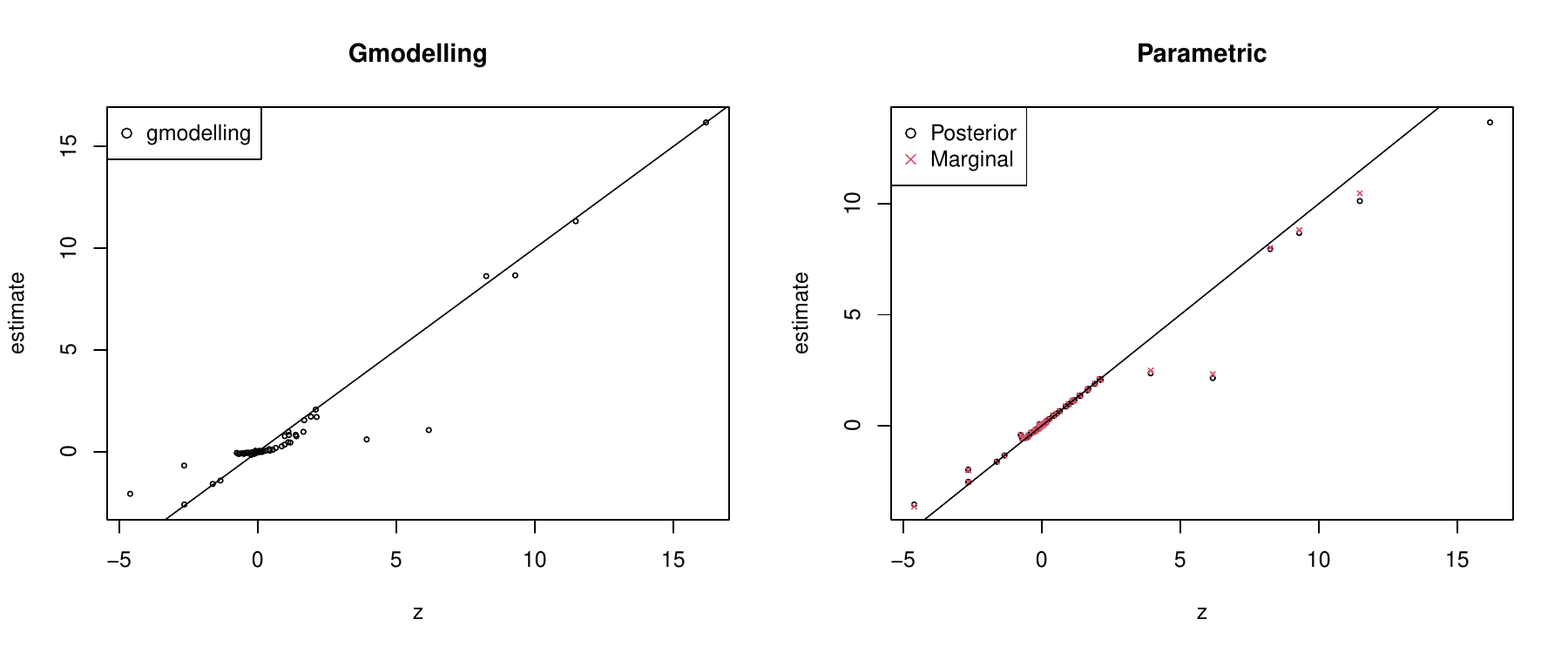}
\caption{The left panel compared g-modeling based nonparametric empirical Bayes estimator with the MLE estimator for $\theta_i$ (the latter is the 45 degree line). The right panel compares the linear shrinkage estimator, using either marginal moments matching or posterior moments matching, with the MLE.}
\label{fig:shrinkage}
\end{figure}
	
	\section{Conclusion}
Adaptive data collection need not make Empirical Bayes practice
algorithm-specific.  As long as the sampling rules and stopping times depend on
the data only through the observed history, the likelihood factors into the
familiar Gaussian working likelihood and a term that is parameter free.  
Standard $g$-modeling therefore recovers the common distribution
of treatment effects, and with it, the posterior estimates, without any
knowledge of the sampling or stopping rules.  Our moment-matching
interpretation explains this robustness and, at the same time, why naive
$f$-modeling fails when adaptivity distorts the marginal distribution of the
observed sample means.

We establish identification and regret guarantees for both parametric and
nonparametric $g$-modeling, and our simulations and the ASOS application show
that the method remains effective when the sampling algorithms
are complex or unknown. Two key assumptions still remain: 
that the experiments are comparable enough to
share a common prior and, for the theoretical analysis of NPMLE, that they are independent.
Relaxing these requirements is a natural next step.  More broadly, we envision that our results
can turn collections 
of adaptive experiments into a resource not only for improving past estimates
but also for designing better experimentation strategies for future experiments.

\bibliographystyle{IEEEtranSN}
\bibliography{EB_adaptive_experiments}

\appendix

\section{Proofs \protect \label{sec:Appendix:A}}

\subsection{Proof of Proposition \ref{Prop:likelihood principle}}
We envision the data collection as unfolding over a sequence of stages $k = 1,2,\dots$. At each stage, an arm $i$ is `pulled', based on information collected in the previous stages. 
Following \cite{lattimore2020bandit}, we employ the `stack-of-rewards' notation for the outcomes from each arm $i$. Specifically, we imagine that prior to the experimenter starting the data collection, nature determined a sequence of outcomes $\bm{y}_i := \{Y_{1,i},Y_{2,i},\dots \}$ for each arm $i$ as i.i.d draws from $\N(\tilde \theta_i, \omega_i^2)$. When the experimenter samples arm $i$ for the $j$-th time, she observes the outcome $Y_{j,i}$ at the top of the stack; this outcome is then removed from the stack.

Let $\mathcal{F}_k$ denote the information set at stage $k$: this includes data on all the previous outcomes and all past policy randomziations, the latter encoded by exogenous Uniform$[0,1]$ random variables $U_{1}, \dots, U_{k-1}$. The experimenter employs a policy rule
$$
\pi_{k} \equiv (\pi_k(1),\dotsm\pi_k(n)):\mathcal{F}_k \to [0,1]^n
$$
that specifies the probability $\pi_{k}(i)$ of arm $i$ being sampled given the information set $\mathcal{F}_k$. Formally, the choice of which arm, $i_k \in \{1,\dots ,n\}$, is sampled is determined as 
$$
\{i_{k} = i\} \iff \left\{U_{k} \in  \left[1\{i \neq 1\}\sum_{a = 1}^{i-1} \pi_k(a), \sum_{a = 1}^{i} \pi_k(a) \right]\right\},
$$
implying $i_k$ is a deterministic function of the past information and $U_k$. As required by the statement of Proposition \ref{Prop:likelihood principle}, the policy rule cannot depend on the parameter vector $(\theta_1, \dots, \theta_n)$, these quantities are not included in $\mathcal{F}_k$. Furthermore, let $A_k =1$ indicate stopping after stage $k$, and let $s_k = Pr(A_k = 1|\mathcal{F}_k)$. At each stage $k$, let $j_i(k)$ denote the number of times arm $i$ was previously sampled. We denote $N_i$ as the final sample size of each arm $i$ when the experiment stops and $K = \sum_i N_i$. Based on the above, we can write the likelihood of the augmented data, $\D^{aug}$, which in addition to the observed outcomes $Y$, also include $\{U_k, A_k\}_{k \in [K]}$, as 
\begin{align}  
	 	p(\D^{aug} | \theta_1, \dots, \theta_n) & = \prod_{k}  \left[ p(U_k) \cdot \prod_i p\left(Y_{j_i(k)+1,i}| \theta_i \right)^{\mathbb{I}\{i_k = i\}}  s_k^{A_k}(1-s_k)^{1-A_k}\right ] \nonumber \\
	 	& = a(\mathcal{D}^{aug})\cdot \prod_i  \left[ \prod_{k}    p\left(Y_{j_i(k)+1,i}| \theta_i \right)^{\mathbb{I}\{i_k = i\}} \right] = a(\mathcal{D}^{aug}) \cdot  \prod_i  \left[ \prod_{j=1}^{N_i}    p\left(Y_{j,i}| \theta_i \right) \right], \nonumber
\end{align}
Now, by the properties of the normal distribution, with $\sigma_i,  Z_i$ as defined in Section \ref{subsec: Generalizing}, after integrating out the augmented data $\{U_k, A_k\}_{k \in [K]}$, we thus obtain
$$
p(\D | \theta_1, \dots, \theta_n) = c(\D) \cdot \prod_i \frac{1}{\sigma_i} \varphi \big(\frac{Z_i - \theta_i}{\sigma_i} \big),
$$
where $c(\D) := p(\D | 0, \dots, 0)  \cdot \prod_i  \left[ \sqrt{2\pi}\sigma_i  \exp \left\{\frac{Z_i^2}{2\sigma_i^2} \right\} \right].$
This proves the desired claim. 

\subsection{Identification} 
\begin{lem}\label{lem:prior-identification}
	
	Suppose Assumption \ref{boundVar} holds. For every Borel probability measure $G$
	on $\mathbb R$, let
	$P_{G,i}$ denote the marginal distribution of the data when
	$\theta_i\sim G$. Then
	\[
	P_{G,i}=P_{G_0,i}
	\quad\Longleftrightarrow\quad
	G=G_0.
	\]
	Moreover,
	\[
	\mathbb E_{P_{G_0,i}}\!\left[
	\ln\frac{f_{G,\sigma_i(\tau_i)}(Z_i)}
	{f_{G_0,\sigma_i(\tau_i)}(Z_i)}
	\right]
	=
	-\operatorname{KL}(P_{G_0,i}\|P_{G,i})
	\leq0,
	\]
	and equality holds if and only if $G=G_0$. Consequently, $G_0$ uniquely
	maximizes
	\[
	G\longmapsto
	\mathbb E_{P_{G_0,i}}\!\left[
	\ln f_{G,\sigma_i(\tau_i)}(Z_i)
	\right]
	\]
	over $\mathcal G(\mathbb R)$. 
	In particular, if $G_0=\mathcal N(0,\gamma_0^{-1})$, then $\gamma_0$
	is identified and uniquely maximizes the population working log
	likelihood over the zero-mean Gaussian family.
\end{lem}

	\begin{proof}
		By Proposition \ref{Prop:likelihood principle} and the discussion in , denote $P_{i,\theta}$ as the distribution of the data $\mathcal{D}_i$ conditional on $\theta_i = \theta$, we have 
		\[
		\frac{dP_{i,\theta}}{dP_{i,0}} 
		=
		\exp\left\{
		\frac{z\theta}{\sigma_i^2(\tau)}
		-\frac{\theta^2}{2\sigma_i^2(\tau)}
		\right\}.
		\]
		The likelihood ratio depends only on $(z,\tau)$. Pushforward towards the distribution of $(Z_i, \tau_i)$ and integrating over $\theta\sim G$ yields, 
		\begin{equation}\label{eq:identification-common-factor}
			p_{G,i}(z,\tau)
			=
			c_i(z, \tau)f_{G,\sigma_i(\tau)}(z),
		\end{equation}
		where $c_i$ is nonnegative and does not depend on $G$. 
		
		The implication $G=G_0\Rightarrow P_{G,i}=P_{G_0,i}$ is immediate.
		Conversely, suppose $P_{G,i}=P_{G_0,i}$. Since
		$p_{G_0,i}=c_i f_{G_0,\sigma_i}$ integrates to one and every Gaussian
		convolution is strictly positive, Fubini's theorem and Assumption
		\ref{boundVar} give a $\tau$ for which
		$0<\sigma_i(\tau)<\infty$ and
		\[
		A_\tau:=\{z:c_i(z,\tau)>0\}
		\]
		has positive Lebesgue measure. Cancelling $c_i$ in
		\eqref{eq:identification-common-factor} gives
		\[
		f_{G,\sigma_i(\tau)}(z)
		=
		f_{G_0,\sigma_i(\tau)}(z)
		\]
		for almost every $z\in A_\tau$. The two Gaussian convolutions are real
		analytic, so they agree on all of $\mathbb R$. Identification of
		Gaussian location mixtures then gives $G=G_0$.
		
		Finally, strict positivity of Gaussian convolutions and
		\eqref{eq:identification-common-factor} imply that, $P_{G_0,i}$-almost
		surely,
		\[
		\frac{dP_{G,i}}{dP_{G_0,i}}
		=
		\frac{f_{G,\sigma_i(\tau)}(z)}
		{f_{G_0,\sigma_i(\tau)}(z)}.
		\]
		Hence the population log-likelihood ratio in the statement
		equals $-\operatorname{KL}(P_{G_0,i}\|P_{G,i})$, which is nonpositive by Gibbs' inequality, with equality holding when
		$P_{G,i}=P_{G_0,i}$. Given identification established earlier, equality holds when
		$G=G_0$. This proves the unique-maximizer claim.
	\end{proof}

\section{Verification of Assumptions} \label{appendixC}

\subsection{Sufficient condition for Assumption \ref{likfactor}}

\begin{lem} \label{lemma:sufficient_condition_Asm3}
	Suppose that for each \(i\), $\tau_i$ is supported on a discrete set $\mathcal{T}$ with $\vert \mathcal{T} \vert$ = $M$. Take $\nu(\tau_i)$ to be the uniform measure over $\mathcal{T}$. Then, Assumption  \ref{likfactor} is satisfied with $\bar{c} \le M$.
\end{lem} 

\begin{proof}
	Let $Z_{i,t}$ denote the sample mean when the number of observations is exactly $t$, and define $\sigma_{i,t}^2 = \omega_i^2/t$. Then, $Z_{i,t} \sim \mathcal{N}(\theta, \sigma^2_{i,t})$. Note that $Z_i := Z_{i,\tau_i}$. Then, for any Borel set $B$ and $t \in \mathcal{T}$, 
	\begin{align*} 
		\mathbb{P}(Z_i \in B, \tau_i = t \mid \theta=0) &= \mathbb{P}(Z_{i,t} \in B, \tau_i = t \mid \theta=0) \\
		&= \int_{B} \mathbb{P}(\tau_i = t | Z_{i,t} =z,\theta =0) \phi_{\sigma^2_{i,t}}(z)dz.
	\end{align*}
	Since we take the base measure $\nu(\tau)$ to be the uniform measure over $\mathcal{T}$,
	\[
	p(z_i, \tau_i = t |\theta=0) = M \cdot \mathbb{P}(\tau_i = t | Z_{i,t} =z_i,\theta =0) \phi_{\sigma^2_{i,t}}(z_i).
	\]
Consequently, 
	\begin{align*}
		& \sup_{z_i, \tau_i \in [\underline{\tau}, \bar \tau]} \left\{ \sqrt{2\pi} \sigma_i  e^{z_i^2/2\sigma_i^2} \cdot p(z_i,\tau_i | \theta_i = 0) \right\} \\
		& = \sup_{z_i, t \in \mathcal{T}} \left\{ \sqrt{2\pi} \sigma_{i,t}  e^{z_i^2/2\sigma_{i,t}^2} \cdot M \cdot \mathbb{P}(\tau_i = t | Z_{i,t} =z_i,\theta =0) \phi_{\sigma^2_{i,t}} \right\}  \le M.
	\end{align*}
	
	Hence, we can take $\bar{c} = M$.
\end{proof}

\subsection{Verifying Assumption \ref{likbound}}

The following Lemma, which relies on Proposition 6 in \cite{soloff2024multivariate}, provides an explicit construction of the support $\mathcal{A}_n$ for the NPMLE such that Assumption \ref{likbound} holds. Proposition 6 in \cite{soloff2024multivariate} is a fixed-data approximation result. It
does not require the realized $Z_i$ to have Gaussian marginal
densities, nor does it require $\sigma_i$ to be nonrandom.
Therefore the pathwise argument remains valid when
$\sigma_i=\omega_i/\sqrt{\tau_i}$ is random and statistically
dependent on $Z_i$. Proposition 6 of \cite{soloff2024multivariate} yields the claimed likelihood bound for every fixed realization $(z_1, \dots, z_n, \sigma_1, \dots, \sigma_n)$ satisfying $\sigma_i^2 \geq \sigma_\ell^2$ for all $i$.  Assumption \ref{boundVar} guarantees that this condition holds almost surely. The resulting likelihood bound therefore holds almost surely, as required by Assumption \ref{likbound}.

\begin{lem}\label{lemma: qualityNPMLE}

	Suppose Assumption \ref{boundVar} holds. For $n\geq3$, let $\kappa_n
	=
	\frac{2\ln n-\ln(e\sqrt{2\pi})}{n}.$
	For a realization $\bm z=(z_1,\ldots,z_n)$, set $z_{\min}:=\min_{1\leq i\leq n}z_i$, $z_{\max}:=\max_{1\leq i\leq n}z_i$ and $D_n(\bm z):=z_{\max}-z_{\min}$.	If $D_n(\bm z)=0$, define $\mathcal A_n(\bm z):=\{z_{\min}\}$ and if $D_n(\bm z)>0$, let 	$\mathsf{k}_0:=\min\{1,\sigma_\ell^2\}$ and define $r_n(\bm z)
	:=
	\min\left\{
	\frac{\sqrt{3}\,\mathsf{k}_0}{4D_n(\bm z)}	\quad
	\mathsf{k}_0
	\sqrt{
		\frac{\kappa_n}
		{2D_n(\bm z)^2+1/2}
	}
	\right\}$, $J_n(\bm z)
	:=
	\left\lceil
	\frac{D_n(\bm z)}{r_n(\bm z)}
	\right\rceil$, 
	$\delta_n^{\mathrm{grid}}(\bm z)
	:=
	\frac{D_n(\bm z)}{J_n(\bm z)},$ and $	\mathcal A_n(\bm z)
	:=
	\left\{
	z_{\min}+j\delta_n^{\mathrm{grid}}(\bm z):
	j=0,\ldots,J_n(\bm z)
	\right\}.$

	Let $\widehat G_n^{\mathcal A}
	\in
	\argmax_{G\in\mathcal G(\mathcal A_n(\bm Z))}
	\ell_n(G),$ then for every realized
	$(z_1,\ldots,z_n,\sigma_1,\ldots,\sigma_n)$ satisfying
	$\sigma_i^2\geq\sigma_\ell^2$ for all $i$
	\[
	0\leq
	\sup_{G\in\mathcal G(\mathbb R)}\ell_n(G)
	-\ell_n(\widehat G_n^{\mathcal A})
	\leq\kappa_n.
	\]
	Consequently, under Assumption \ref{boundVar}, the bound
	holds almost surely. Moreover, for $n\geq3$,
	$q_n=e\sqrt{2\pi}/n^2$, so the displayed $\kappa_n$ equals
	$n^{-1}\ln(1/q_n)$ and $\widehat G_n^{\mathcal{A}}$ satisfies Assumption
	\ref{likbound}. 
	
\end{lem}

\begin{proof}
	Fix an arbitrary realization
	$(z_1,\ldots,z_n,\sigma_1,\ldots,\sigma_n)$ satisfying
	$\sigma_i^2\geq\sigma_\ell^2\geq \mathsf{k}_0$ for every $i$. The argument below is
	deterministic, conditional on this realization.
	
	First suppose $D_n(\bm z)=0$, so that
	$z_1=\cdots=z_n=z_{\min}$. For every probability measure $G$ and
	every $i$,
	\[
	f_{G,\sigma_i}(z_{\min})
	=
	\int
	\frac1{\sigma_i}
	\varphi\left(\frac{z_{\min}-\theta}{\sigma_i}\right)
	dG(\theta)
	\leq
	\frac1{\sigma_i}\varphi(0)
	=
	f_{\delta_{z_{\min}},\sigma_i}(z_{\min}).
	\]
	Hence $\delta_{z_{\min}}$ is a global maximizer of $\ell_n$, and
	the asserted likelihood gap is zero.
	
	Now suppose $D_n(\bm z)>0$. In dimension $d=1$, every
	global maximizer of $\ell_n$ over $\mathcal G(\mathbb R)$ is
	supported on
	\[
	\mathcal M_n(\bm z):=[z_{\min},z_{\max}].
	\]
	This set has diameter $D_n(\bm z)$. By construction,
	\[
	0<\delta_n^{\mathrm{grid}}(\bm z)
	\leq r_n(\bm z)
	\leq
	\frac{\sqrt{3}\,\mathsf{k}_0}{4D_n(\bm z)}
	<
	\frac{\sqrt{3}}{2}
	\frac{\mathsf{k}_0}{D_n(\bm z)}.
	\]
	The intervals
	\[
	\left[
	z_{\min}+j\delta_n^{\mathrm{grid}}(\bm z),
	z_{\min}+(j+1)\delta_n^{\mathrm{grid}}(\bm z)
	\right],
	\qquad j=0,\ldots,J_n(\bm z)-1,
	\]
	therefore form a cover of $\mathcal M_n(\bm z)$ by closed
	one-dimensional hypercubes of width
	$\delta_n^{\mathrm{grid}}(\bm z)$, and their corners are the
	points of $\mathcal A_n(\bm z)$.
	
	Apply Proposition 6 of \cite{soloff2024multivariate} with $	d=1,
	\Sigma_i=\sigma_i^2,
	\mathsf{k}=\mathsf{k}_0,
	D=D_n(\bm z).$
	Here $\mathsf{k}$ denotes the lower bound on the minimum eigenvalue
	of $\Sigma_i$. The strict mesh restriction in that proposition is $	\delta
	<
	\sqrt{\frac{3}{4d}}\,\mathsf{k}D^{-1}
	=
	\frac{\sqrt{3}}{2}
	\frac{\mathsf{k}_0}{D_n(\bm z)},$
	which was verified above. Proposition 6 of \cite{soloff2024multivariate} consequently gives
	\[
	\sup_{G\in\mathcal G(\mathbb R)}\ell_n(G)
	-\ell_n(\widehat G_n)
	\leq
	\mathsf{k}_0^{-2}
	\left(2D_n(\bm z)^2+\frac12\right)
	\{\delta_n^{\mathrm{grid}}(\bm z)\}^2.
	\]
	
	The second term in the definition of $r_n(\bm z)$ and the fact
	that $\delta_n^{\mathrm{grid}}(\bm z)\leq r_n(\bm z)$ imply $	\{\delta_n^{\mathrm{grid}}(\bm z)\}^2
	\leq
	\mathsf{k}_0^2
	\frac{\kappa_n}{2D_n(\bm z)^2+1/2}.$
	Substituting this bound into the preceding display yields $	\sup_{G\in\mathcal G(\mathbb R)}\ell_n(G)
	-\ell_n(\widehat G_n^{\mathcal A})
	\leq\kappa_n.$
\end{proof}

\subsection{Verification for a two-stage group sequential trial}
\label{sec:two-stage-verification}

The two-stage group sequential trial is a basic instance of an adaptive design that proceeds in two phases. Suppose that each stage contains $N$ observations and
\[
Y_{j,i}\sim\mathcal N(\theta_i/\sqrt N,\omega_i^2).
\]
with $\omega_i \in [\omega_\ell, \omega_u]$ for all $i$. 
Define the independent stage statistics
\[
X_{1i}:=\frac1{\sqrt N}\sum_{j=1}^{N}Y_{j,i},
\qquad
X_{2i}:=\frac1{\sqrt N}\sum_{j=N+1}^{2N}Y_{j,i}.
\]
The trial stops
after stage one if $|X_{1i}|>a\omega_i$, where $a>0$ is fixed; otherwise it
continues to stage two.  With $\tau_i=N_i/N\in\{1,2\}$,
\begin{equation}\label{eq:two-stage-statistic}
	Z_i
	=
	\frac{\sqrt N}{N_i}\sum_{j=1}^{N_i}Y_{j,i}
	=
	\begin{cases}
		X_{1i},&\tau_i=1,\\[2mm]
		(X_{1i}+X_{2i})/2,&\tau_i=2.
	\end{cases}
\end{equation}
Thus the working variance is $\sigma_i^2=\frac{\omega_i^2}{\tau_i};
\sigma_i(1)=\omega_i,
\sigma_i(2)=\frac{\omega_i}{\sqrt2}.$

Let $\nu$ be uniform on $\{1,2\}$, put $b_i:=a\omega_i$, and, for
$s>0$, write $\varphi_s(x):=s^{-1}\varphi(x/s)$.  We first derive the density $r_i(z,\tau)=p(z,\tau\mid\theta_i=0)$
with respect to $dz\,d\nu(\tau)$.  For $\tau=1$, the joint density with
respect to $dz$ is $\varphi_{\omega_i}(z)1\{|z|>b_i\}.$
Because $\nu(\{1\})=1/2$, the corresponding density with respect to
$dz\,d\nu$ is twice this expression.

For $\tau=2$, set $x=X_{1i}$ and use
$X_{2i}=2z-x$.  The Jacobian of the transformation from $X_{2i}$ to $z$
is $2$, so the joint density of $(Z_i,\tau_i=2)$ with respect to $dz$ is
\begin{align*}
	h_i^{(2)}(z)
	&=
	\int_{-b_i}^{b_i}
	2\varphi_{\omega_i}(x)\varphi_{\omega_i}(2z-x)\,dx=
	\frac1{\pi\omega_i^2}
	\exp\left(-\frac{z^2}{\omega_i^2}\right)
	\int_{-b_i}^{b_i}
	\exp\left(-\frac{(x-z)^2}{\omega_i^2}\right)dx.
\end{align*}
Again dividing by the base-measure mass $\nu(\{2\})=1/2$ gives
\begin{equation}\label{eq:two-stage-base-density}
	r_i(z,\tau)=
	\begin{cases}
		\displaystyle
		\frac{2}{\sqrt{2\pi}\,\omega_i}
		\exp\left(-\frac{z^2}{2\omega_i^2}\right)
		1\{|z|>b_i\},&\tau=1,\\[4mm]
		\displaystyle
		\frac{2}{\pi\omega_i^2}
		\exp\left(-\frac{z^2}{\omega_i^2}\right)
		\int_{-b_i}^{b_i}
		\exp\left(-\frac{(x-z)^2}{\omega_i^2}\right)dx,
		&\tau=2,\\[4mm]
		0,&\tau\notin\{1,2\}.
	\end{cases}
\end{equation}
\subsubsection*{Verification of Assumption \ref{boundVar}} Since $\tau$ can only take values $\{1,2\}$ and since $\omega_i \in [\omega_\ell, \omega_u]$ for all $i$, Assumption \ref{boundVar} holds with $\underline\tau=1$ and
$\bar\tau=2$; in particular, $\frac{\omega_\ell^2}{2}\leq\sigma_i^2(\tau_i)
\leq\omega_u^2.$

\subsubsection*{Verification of Assumption \ref{likfactor}}
For $\tau\in\{1,2\}$ and every $z$ such that
$r_i(z,\tau)>0$, the likelihood ratio relative to $\theta_i=0$ is $\frac{p_i(z,\tau\mid\theta_i)}{r_i(z,\tau)}
=
\exp\left\{
\frac{z\theta_i}{\sigma_i^2(\tau)}
-\frac{\theta_i^2}{2\sigma_i^2(\tau)}
\right\}.$
Consequently,
\[
p_i(z,\tau\mid\theta_i)
=
\frac{r_i(z,\tau)}{\varphi_{\sigma_i(\tau)}(z)}
\varphi_{\sigma_i(\tau)}(z-\theta_i).
\]
Dominance follows from Lemma \ref{lemma:sufficient_condition_Asm3}. Since $\tau \in \{1,2\}$, dominance in Assumption \ref{likfactor} holds with $\bar c = 2$. 
	

\subsubsection*{Verification of Assumption \ref{tailbound}}

Lemma \ref{lem:density_dominance}, the bound $\bar c=2$, the support
condition $|\theta|\leq R$, and
$\sigma_i(\tau_i)\leq\omega_u$ imply that, for $x>R$,
\begin{align*}
	\mathbb P(|Z_i|>x)
	&\leq
	2\sup_{\substack{|\theta|\leq R\\0<\sigma\leq\omega_u}}
	\mathbb P\{|\mathcal N(\theta,\sigma^2)|>x\}\leq
	4\exp\left\{-\frac{(x-R)^2}{2\omega_u^2}\right\}.
\end{align*}
Therefore, by the union bound,
\[
\mathbb P\left(\max_{1\leq i\leq n}|Z_i|>x\right)
\leq
4n\exp\left\{-\frac{(x-R)^2}{2\omega_u^2}\right\}.
\]
Choose $M_n=\sqrt{\kappa\ln n}$ with
$\kappa\geq24\omega_u^2$.  For all sufficiently large $n$,
$M_n\geq2R$, and the preceding display gives
\[
\mathbb P\left(\max_i|Z_i|>M_n\right)\lesssim n^{-2}.
\]
Integrating the same tail bound, after splitting the integral at a
constant multiple of $\sqrt{\ln n}$, gives
\[
\mathbb E\left[(\max_i|Z_i|\vee1)^4\right]
\lesssim(\ln n)^2.
\]
This proves Assumption \ref{tailbound}.

\subsubsection*{Verification of Assumption \ref{support-geometry}}

For $\tau=1$, take
\[
\mathcal S_{i,1}=(-\infty,-b_i)\cup(b_i,\infty),
\]
and for $\tau=2$ take $\mathcal S_{i,2}=\mathbb R$.  Since $\nu$ has finite support
and these sections are Borel sets, the corresponding set of pairs
$(z,\tau)$ is measurable.  The density
\eqref{eq:two-stage-base-density} is positive on these open sets and zero
Lebesgue-a.e. outside them.  Every connected component is either a
half-line or the full line; in particular, there are no bounded components.
The lower-length condition in Assumption \ref{support-geometry} is therefore
vacuous.  

\subsubsection*{Verification of Assumption \ref{density-2}}

For $\tau=1$, the density is positive precisely when $|z|>b_i$.  On that
set,
\[
\partial_z\ln r_i(z,1)=-\frac{z}{\omega_i^2},
\qquad
\partial_z^2\ln r_i(z,1)=-\frac1{\omega_i^2},
\qquad
\partial_z^l\ln r_i(z,1)=0\quad(l\ge3).
\]

For $\tau=2$, rewrite \eqref{eq:two-stage-base-density} as
\[
r_i(z,2)
=
\frac{2}{\pi\omega_i^2}
\exp\left(-\frac{2z^2}{\omega_i^2}\right)A_i(z),
\qquad
A_i(z)
:=
\int_{-b_i}^{b_i}
\exp\left(\frac{2zx}{\omega_i^2}-\frac{x^2}{\omega_i^2}\right)dx.
\]
The function $A_i$ is positive and analytic.  Define the probability
measure
\[
Q_{i,z}(dx)
:=
\frac{
	\exp\left(2zx/\omega_i^2-x^2/\omega_i^2\right)
	1\{|x|\le b_i\}\,dx
}{
	A_i(z)
}.
\]
Differentiating $\ln A_i(z)$ gives
\begin{align}
	\partial_z\ln r_i(z,2)
	&=
	-\frac{4z}{\omega_i^2}
	+\frac{2}{\omega_i^2}\E_{Q_{i,z}}X,
	\label{eq:two-stage-score-bound}\\
	\partial_z^2\ln r_i(z,2)
	&=
	-\frac4{\omega_i^2}
	+\frac4{\omega_i^4}\operatorname{Var}_{Q_{i,z}}(X),
	\label{eq:two-stage-second-derivative}
\end{align}
and, for $l\ge3$,
\begin{equation}\label{eq:two-stage-higher-derivatives}
	\partial_z^l\ln r_i(z,2)
	=
	\left(\frac2{\omega_i^2}\right)^l
	\kappa_l(Q_{i,z}),
\end{equation}
where $\kappa_l(Q_{i,z})$ denotes the $l$th cumulant of $X$ under
$Q_{i,z}$.

If a random variable is supported on $[-b_i,b_i]$, the
moment--cumulant formula gives $|\kappa_l|\le K_lb_i^l$
for a constant $K_l$ depending only on $l$.  Since
$b_i=a\omega_i$ and $\omega_i\ge\omega_\ell$, the preceding displays imply
\begin{align*}
	|\partial_z\ln r_i(z,2)|
	&\le
	\frac{4}{\omega_\ell^2}|z|+\frac{2a}{\omega_\ell},\\
	|\partial_z^2\ln r_i(z,2)|
	&\le
	\frac{4(1+a^2)}{\omega_\ell^2},\\
	|\partial_z^l\ln r_i(z,2)|
	&\le
	K_l\left(\frac{2a}{\omega_\ell}\right)^l,
	\qquad l\ge3.
\end{align*}
Together with the derivatives for $\tau=1$, these inequalities show that,
for every fixed $m\ge2$, there are constants $C_1,\ldots,C_m$ independent
of $n$ and $i$ such that
\[
|\partial_z^l\ln r_i(z,\tau)|
\le C_l(1+|z|)^l,
\qquad l=1,\ldots,m,
\]
at every point where $r_i(z,\tau)>0$.  This proves Assumption \ref{density-2}.

\clearpage
\thispagestyle{empty}

\begin{center}
	{\Large\bfseries Online Supplementary Material}
\end{center}

\vspace{0.5cm}


\let\thesectionWithoutS\thesection 
\renewcommand\thesection{S.\thesectionWithoutS}

\setcounter{prop}{0}
\setcounter{lem}{0}

\renewcommand{\theprop}{S.\arabic{prop}}
\renewcommand{\thelem}{S.\arabic{lem}}
	\section{Proofs}\label{sec:Appendix:S}

\subsection{Proof of Theorem \ref{thm:NPMLE-consistency}}

Recall Tweedie's formula (\ref{Tweedie's formula}): 
\[
\hat \theta_{i, G} :=	\E_{G}[\theta_i | \D] = Z_i + \sigma_i^2 \frac{f_{G, \sigma_i}'(Z_i)}{f_{G,\sigma_i}(Z_i)}.
\]
Using this formula, we can express Bayes regret as 
\[
\mathcal{R}(\bm \delta^{NPEB}, G_0) = \E\Big [\frac{1}{n} \sum_i \Big(\sigma_i^2 \frac{f_{\hat G_n, \sigma_i}'(Z_i)}{f_{\hat G_n,\sigma_i}(Z_i)} - \sigma_i^2 \frac{f_{G_0, \sigma_i}'(Z_i)}{f_{G_0,\sigma_i}(Z_i)}\Big)^2\Big] 
\]

Since $f_{G,\sigma_i}$ appears in the denominator, we need to prevent it from being too small. Fix some $\rho > 0$, and define the regularized estimator 
\[
\hat \theta_{i,G,\rho} = Z_i + \sigma_i^2 \frac{f_{G, \sigma_i}'(Z_i)}{f_{G,\sigma_i}(Z_i) \vee \frac{\rho}{\sigma_i}},
\]
where $a \vee b = \max \{a,b\}$. 

Assumption \ref{likbound} and Theorem 5 of
\cite{jiang2020general} imply the pointwise lower bounds at the observed
sample points
\[
f_{\widehat G_n,\sigma_i}(Z_i)
\geq\frac{q_n}{\sqrt{2\pi}en\sigma_i},
\qquad i=1,\ldots,n.
\]
Set $\rho_n=q_n/(\sqrt{2\pi}en)$.  For all sufficiently large $n$,
$q_n=e\sqrt{2\pi}/n^2$, so $\rho_n=n^{-3}$.  Consequently,
\[
f_{\widehat G_n,\sigma_i}(Z_i)\geq\rho_n/\sigma_i
\quad\text{and hence}\quad
\widehat\theta_{i,\widehat G_n,\rho_n}(Z_i)
=\widehat\theta_{i,\widehat G_n}(Z_i).
\]
Decompose $\| \hat{\bm{\theta}}_{\hat G_n} - \hat{\bm{\theta}}_{G_0}\| $ by the triangle inequality:
\[
\| \hat{\bm{\theta}}_{\hat G_n} - \hat{\bm{\theta}}_{G_0}\| = \| \hat{\bm{\theta}}_{\hat G_n, \rho_n} - \hat{\bm{\theta}}_{G_0}\|  \le \| \hat{\bm{\theta}}_{\hat G_n, \rho_n} - \hat{\bm{\theta}}_{G_0, \rho_n}\|  + \| \hat{\bm{\theta}}_{ G_0, \rho_n} - \hat{\bm{\theta}}_{G_0}\|. 
\]
As a result,  
\begin{eqnarray}
	\mathcal{R}(\bm \delta^{NPEB}, G_0)& \lesssim  \E\Big [\frac{1}{n} \sum_i (\hat \theta_{i, \hat G_n, \rho_n} - \hat \theta_{i, G_0, \rho_n})^2\Big ] \label{regretTerm1}\\
	& \ +\ \E\Big [\frac{1}{n} \sum_i (\hat \theta_{i, G_0, \rho_n} - \hat \theta_{i, G_0})^2\Big ] \label{regretTerm2} .
\end{eqnarray} 
\subsubsection*{Bounding \eqref{regretTerm2}}
Observe that
\begin{align} \label{eq1:bound_on_A_4}
	& \E\Big [\frac{1}{n} \sum_i (\hat \theta_{i, G_0, \rho_n} - \hat \theta_{i, G_0})^2\Big ] = \E \Big[ \frac{1}{n} \sum_i \Big(\sigma_i^2 \frac{f'_{G_0,\sigma_i}(Z_i)}{f_{G_0,\sigma_i}(Z_i)} - \sigma_i^2 \frac{f'_{G_0,\sigma_i}(Z_i)}{f_{G_0,\sigma_i}(Z_i) \vee \frac{\rho_n}{\sigma_i}}\Big)^2\Big] \nonumber\\
	&\overset{(1)}{ \leq} \bar c \frac{1}{n} \sum_i \int \Big(\frac{\omega_i^2}{\tau} \frac{f'_{G_0,\frac{\omega_i}{\sqrt{\tau}}}(z)}{f_{G_0,\frac{\omega_i}{\sqrt{\tau}}}(z)} - \frac{\omega_i^2}{\tau}  \frac{f'_{G_0,\frac{\omega_i}{\sqrt{\tau}}}(z)}{f_{G_0,\frac{\omega_i}{\sqrt{\tau}}}(z) \vee \frac{\rho_n}{\omega_i/\sqrt{\tau}}}\Big)^2 f_{G_0,\frac{\omega_i}{\sqrt{\tau}}}(z) dz d\nu(\tau) \nonumber\\
	& \leq \sup_i \bar c  \int \Big(\frac{\omega_i^2}{\tau} \frac{f'_{G_0,\frac{\omega_i}{\sqrt{\tau}}}(z)}{f_{G_0,\frac{\omega_i}{\sqrt{\tau}}}(z)} \Big)^2\Big( 1-  \frac{f_{G_0,\frac{\omega_i}{\sqrt{\tau}}}(z)}{f_{G_0,\frac{\omega_i}{\sqrt{\tau}}}(z) \vee \frac{\rho_n}{\omega_i/\sqrt{\tau}}}\Big)^2 f_{G_0,\frac{\omega_i}{\sqrt{\tau}}}(z) dz d \nu(\tau) \nonumber\\	
	&	\overset{(2)}{ \leq} \sup_{i} \bar c  \sqrt{\int \Big(\frac{\omega_i^2}{\tau} \frac{f'_{G_0,\frac{\omega_i}{\sqrt{\tau}}}(z)}{f_{G_0,\frac{\omega_i}{\sqrt{\tau}}}(z)}\Big)^4 f_{G_0,\frac{\omega_i}{\sqrt{\tau}}}(z) dz d\nu(\tau) } \sqrt{\int 1\{f_{G_0,\frac{\omega_i}{\sqrt{\tau}}}(z) \leq  \frac{\rho_n}{\omega_i/\sqrt{\tau}}\} f_{G_0,\frac{\omega_i}{\sqrt{\tau}}}(z) dz d\nu(\tau)},
\end{align}
where inequality (1) is due to Lemma \ref{lem:density_dominance}, and inequality (2) employs the Cauchy-Schwartz inequality.

For the first term in (\ref{eq1:bound_on_A_4}), fix $\sigma_i= \omega_i/\sqrt{\tau}$ and assume $X$ is a random variable such that $X |\theta \sim N(\theta, \sigma_i^2)$ and $\theta \sim G_0$. Then, $X$ has marginal density $f_{G_0,\sigma_i}$ and $\sigma_i^2 \frac{f_{G_0,\sigma_i}'(X)}{f_{G_0,\sigma_i}(X) }= \E_{G_0}[\theta | X,\sigma_i]- X$; hence, under Assumption \ref{boundVar}, 
\begin{align*}
	& \int \Big(\sigma_i^2 \frac{f'_{G_0,\sigma_i}(z)}{f_{G_0,\sigma_i}(z)}\Big)^4 f_{G_0,\sigma_i}(z) dz =\E[ (\E_{G_0}[\theta | X, \sigma_i] - X)^4] = \E[ (\E_{G_0}[\theta - X| X, \sigma_i] )^4]\\
	& \leq \E[\E_{G_0}[(\theta - X)^4 | X, \sigma_i]] = \E[(\theta - X)^4]\leq 3 \sigma_{u}^4,
\end{align*} 
where the first inequality is due to Jensen's inequality. Thus, the first term in (\ref{eq1:bound_on_A_4}) is bounded for all $i$. 

For the second term in (\ref{eq1:bound_on_A_4}), under Assumption \ref{boundedG0} and \ref{boundVar}, apply Lemma \ref{lemmaBound3} and  plug-in the choice $\rho_n = 1/n^3$; then, involving the same $X$ defined above, 
\[
\int 1\{f_{G_0,\sigma_i}(z) \leq \rho_n/\sigma_i\} f_{G_0,\sigma_i}(z) dz  \leq 3 \Big( \frac{\rho_n}{\sigma_i}\Big)^{2/3} \textrm{Var}[X]^{1/3} \lesssim \frac{1}{n^2} \nonumber 
\]
Putting everything together,  we have 
\begin{equation} \label{regretboundP1}
	\E\Big [\frac{1}{n} \sum_i (\hat \theta_{i, G_0, \rho_n} - \hat \theta_{i, G_0})^2\Big ] \lesssim  \frac{1}{n}. 
\end{equation}

\subsubsection*{Bounding \eqref{regretTerm1}}
We now bound the term $ \E\Big [\frac{1}{n} \sum_i (\hat \theta_{i, \hat G_n, \rho_n} - \hat \theta_{i, G_0, \rho_n})^2\Big ]$. Letting $\bar Z_n = \max_i |Z_i| \vee 1$, we decompose 
\begin{eqnarray}
	& \E\Big [\frac{1}{n} \sum_i (\hat \theta_{i, \hat G_n, \rho_n} - \hat \theta_{i, G_0, \rho_n})^2\Big ] \nonumber \\
	&= \E\Big  [\frac{1}{n} \sum_i (\hat \theta_{i, \hat G_n, \rho_n} - \hat \theta_{i, G_0, \rho_n})^2 1\{\bar Z_n > M_n\}\Big ] \label{Term1} \\ 
	& \ +\ \E\Big  [\frac{1}{n} \sum_i (\hat \theta_{i, \hat G_n, \rho_n} - \hat \theta_{i, G_0, \rho_n})^2 1\{\bar Z_n \leq M_n\}\Big ], \label{Term2} 
\end{eqnarray}
with $M_n = \sqrt{\kappa \ln n}$ for some $\kappa > 0$ as defined in Assumption \ref{tailbound}.

To bound \eqref{Term1}, note that by the Cauchy-Schwarz inequality, 
\begin{align*}
	& \E\Big  [\frac{1}{n} \sum_i (\hat \theta_{i, \hat G_n, \rho_n} - \hat \theta_{i, G_0, \rho_n})^2 1\{\bar Z_n > M_n\}\Big ]\\
	&  \leq \sqrt{\E\Big[ \Big( \frac{1}{n} \sum_i (\hat \theta_{i,\hat G_n, \rho_n}  - \hat \theta_{i,G_0,\rho_n} )^2\Big)^2\Big]}\sqrt{\P(\bar Z_n > M_n)}.
\end{align*}

Furthermore, as argued earlier $\hat \theta_{i,\hat G_n,\rho_n} (z)= \hat \theta_{i,\hat G_n}(z).$ Fix $\sigma_i = \omega_i/\sqrt{\tau}$, $\hat \theta_{i, \hat G_n}(z) = \E_{\hat G_n}[\theta | X = z]$
where $X  | \theta \sim N(\theta, \sigma_i^2)$ and $\theta \sim \hat G_n$. The construction in Lemma \ref{lemma: qualityNPMLE} ensures that the support of $\hat G_n$ is contained in $[-\bar Z_n, \bar Z_n]$, hence $|\hat \theta_{i, \hat G_n}| \leq \bar Z_n$. Furthermore, put
\[
D_{0,i}(z):=f_{G_0,\sigma_i}(z)\vee(\rho_n/\sigma_i),
\qquad
\lambda_{0,i}(z):=
\frac{f_{G_0,\sigma_i}(z)}{D_{0,i}(z)}\in(0,1].
\]
then we can rewrite:
\begin{align*}
	\widehat\theta_{i,G_0,\rho_n}(z)
	&=z+\sigma_i^2\frac{f'_{G_0,\sigma_i}(z)}{D_{0,i}(z)}=\{1-\lambda_{0,i}(z)\}z
	+\lambda_{0,i}(z)\widehat\theta_{i,G_0}(z).
\end{align*}
Since $\operatorname{supp}(G_0)\subseteq[-R,R]$, the posterior
mean satisfies $|\widehat\theta_{i,G_0}(z)|\leq R$.  Hence
\[
|\widehat\theta_{i,G_0,\rho_n}(z)|
\leq\{1-\lambda_{0,i}(z)\}|z|+\lambda_{0,i}(z)R
\leq |z|+R.
\]
Thus, at $z=Z_i$,
$|\widehat\theta_{i,\widehat G_n,\rho_n}(Z_i)-
\widehat\theta_{i,G_0,\rho_n}(Z_i)|
\leq2\bar Z_n+R$.
Putting these bounds together, Assumption \ref{tailbound} gives
\[
\E\Big[ \Big( \frac{1}{n} \sum_i (\hat \theta_{i,\hat G_n, \rho_n}  - \hat \theta_{i,G_0,\rho_n} )^2\Big)^2\Big] \leq \E[(2\bar Z_n + 2R)^4] \lesssim \E[\bar Z_n^4] \lesssim (\ln n)^2. 
\]

Furthermore, Assumption \ref{tailbound} also implies 
\[
\P(\bar Z_n \geq M_n) \lesssim \frac{1}{n^2}.
\]

Combining the above, we conclude 
\begin{equation}\label{regretboundP2}
	\E\Big  [\frac{1}{n} \sum_i (\hat \theta_{i, \hat G_n, \rho_n} - \hat \theta_{i, G_0, \rho_n})^2 1\{\bar Z_n > M_n\}\Big ] \lesssim \frac{\ln n}{n}.
\end{equation} 

To bound \eqref{Term2}, let
$A_n:=\{\bar Z_n\leq M_n\}$ and
$B_n:=\{\bar h^2(p_{\hat G_n},p_{G_0})\leq\epsilon_n^2\}$,
where $\epsilon_n^2=C_\epsilon(\ln n)^2/n$ is as in Proposition
\ref{HellingerAccuracy} in Section \ref{sec:Appendix:Hellinger}.  On $A_n$, the construction in
Lemma \ref{lemma: qualityNPMLE} gives
\[
\operatorname{supp}(\widehat G_n)
\subseteq\mathcal A_n(\bm Z)
\subseteq[\min_i Z_i,\max_i Z_i]
\subseteq[-M_n,M_n].
\]
On $B_n$, by definition,
$\bar h^2(p_{\widehat G_n},p_{G_0})\leq\epsilon_n^2$.
Consequently, on $A_n\cap B_n$,
\[
\widehat G_n\in\mathcal H_n
:=
\left\{
G:\operatorname{supp}(G)\subseteq[-M_n,M_n],\quad
\bar h^2(p_G,p_{G_0})\leq\epsilon_n^2
\right\}.
\]
Set
\[
\delta_n:=\frac{\rho_n}{n},\qquad
\Lambda_n:=\ln\frac1{2\pi\rho_n^2},\qquad
\eta_n^*:=
\left(\sigma_u^3+\sigma_u^2\sqrt{\Lambda_n}\right)
\frac{\delta_n}{\rho_n}.
\]
For all sufficiently large $n$,
$\rho_n=n^{-3}\leq(2\pi e)^{-1/2}$ and
$\delta_n=\rho_n/n\leq\eta_0$.  Part~(ii) of Lemma
\ref{lem:uniform-mixture-entropy}, applied with
$M=M_n$, $\rho=\rho_n$, $\eta=\delta_n$, and
$\mathcal H=\mathcal H_n$, supplies
deterministic centers
$G^{(1)},\ldots,G^{(N)}\in\mathcal H_n$ such that
\begin{equation}\label{eq:theorem1-score-cover}
	\sup_{G\in\mathcal H_n}\min_{j\leq N}
	d_{M_n,\rho_n}(G,G^{(j)})\leq\eta_n^*,
	\qquad
	(\eta_n^*)^2\lesssim\frac{\ln n}{n^2},
	\qquad
	\ln N\lesssim(\ln n)^2.
\end{equation}
with $d_{M,\rho}$ defined in Lemma \ref{lem:uniform-mixture-entropy}. 

Now following \cite{soloff2024multivariate}, we decompose $1\{A_n\} \| \hat{\bm{\theta}}_{\hat G_n, \rho_n} - \hat{\bm{\theta}}_{G_0,\rho_n}\| $ into the following four terms via triangle inequalities: 
\[
1\{A_n\} \| \hat{\bm{\theta}}_{\hat G_n, \rho_n} - \hat{\bm{\theta}}_{G_0,\rho_n}\| \leq \sum_{k=1}^4 \zeta_k,
\]
with 
\begin{align*}
	\zeta_1 &= 1\{A_n \cap B_n^c\} \| \hat{\bm{\theta}}_{\hat G_n, \rho_n} - \hat{\bm{\theta}}_{G_0,\rho_n}\|\\
	\zeta_2 & = 1\{A_n \cap B_n\} \Big(\|\hat{\bm{\theta}}_{\hat G_n, \rho_n} - \hat{\bm{\theta}}_{G_0,\rho_n}\|- \underset{j\in [N]}{\max} \| \hat{\bm{\theta}}_{G^{(j)}, \rho_n} - \hat{\bm{\theta}}_{G_0, \rho_n}\|\Big)_{+}\\
	\zeta_3 & = \underset{j\in [N]}{\max}\;
	\left(
	\|\tilde V_j\|-\E\|\tilde V_j\|
	\right)_+ \\
	\zeta_4&= \underset{j\in [N]}{\max} \; \E \|\tilde V_j \|,
\end{align*}
where $\tilde V_{ij}
=
\left|
\hat\theta_{i,G^{(j)},\rho_n}
-
\hat\theta_{i,G_0,\rho_n}
\right|
1\{|Z_i|\le M_n\}$. 
Then \eqref{Term2} can be bounded, using the Cauchy-Schwartz inequality $(a_1 + a_2 + a_3 + a_4)^2 \leq 4(a_1^2 + a_2^2 + a_3^2+a_4^2)$, as:
\[
\E \Big[ \frac{1}{n} \sum_i (\hat \theta_{i,\hat G_n, \rho_n} - \hat \theta_{i,G_0,\rho_n})^21\{A_n\}\Big] \leq  \frac{4}{n} (\E \zeta_1^2  + \E \zeta_2^2 + \E \zeta_3^2 + \E \zeta_4^2).
\]

\subsubsection*{Bounding $\E \zeta_1^2$}
Using Lemma F.1 of \cite{saha2020nonparametric}, for any $G \in \mathcal{G}(\mathbb{R})$ and for every $z \in \R$, we have that for $\rho \in (0, \frac{1}{\sqrt{2\pi e}}]$,
\[
\frac{\vert f'_{G,1}(z)\vert}{f_{G,1}(z) \vee \rho} \leq \sqrt{\ln \frac{1}{2 \pi \rho^2}}.
\]
Combined with the scaled Tweedie formula (Lemma \ref{lemma: scaledTweedie}), this implies that for any $z_i \in \R$, and with $\tilde z_i := z_i/\sigma_i$,
we use $G_{i0}$ for the law of $\theta/\sigma_i$ when
$\theta\sim G_0$, $\widehat G_{in}$ for the corresponding law when
$\theta\sim\widehat G_n$, and $G_i^{(j)}$ for the corresponding
law when $\theta\sim G^{(j)}$. Then
\[
(\hat \theta_{i,\hat G_n, \rho_n}- \hat \theta_{i,G_0,\rho_n})^2 = \sigma_i^2 \Big( \frac{f'_{\hat G_{in},1}(\tilde z_i)}{f_{\hat G_{in},1}(\tilde z_i) \vee \rho_n} - \frac{f'_{G_{i0},1}(\tilde z_i)}{f_{G_{i0},1}(\tilde z_i) \vee \rho_n}\Big)^2 \leq 4 \sigma_i^2 \ln \left(\frac{1}{2\pi \rho_n^2}\right).
\]
Consequently, under Assumption \ref{boundVar} and using $\rho_n = 1/n^3$,
\begin{equation} \label{regretboundP3}
	\frac{1}{n} \E\zeta_1^2 \lesssim \ln\left(\frac{1}{\rho_n^2}\right) \P(A_n \cap B_n^c) \lesssim  \frac{\ln n}{n},
\end{equation} 
where the last inequality holds because 
$\mathbb P(A_n\cap B_n^c)\leq\mathbb P(B_n^c)\leq3n^{-1}$ by
Proposition \ref{HellingerAccuracy} with $t=1$.

\subsubsection*{Bounding $\E\zeta_2^2$}
Under the event $A_n\cap B_n$, \eqref{eq:theorem1-score-cover}
provides a $j$ such that
\[
d_{M_n,\rho_n}(\hat G_n,G^{(j)})\leq\eta_n^*.
\]
Therefore,

\begin{align*}
	\frac{1}{n} \zeta_2^2 & \leq \frac{1}{n}1\{A_n \cap B_n\} \underset{j\in [N]}{\min} \Big (  \| \hat{\bm{\theta}}_{\hat G_n, \rho_n} - \hat{\bm{\theta}}_{G_0,\rho_n}\| - \|\hat{\bm{\theta}}_{G^{(j)},\rho_n} - \hat{\bm{\theta}}_{G_0,\rho_n}\| \Big)^2\\
	& \leq  \frac{1}{n}1\{A_n\cap B_n\} \underset{j\in [N]}{\min} \|\hat{\bm{\theta}}_{\hat G_n,\rho_n} - \hat{\bm{\theta}}_{G^{(j)},\rho_n}\|^2\\
	& = \frac{1}{n}1\{A_n \cap B_n\} \underset{j\in [N]}{\min} \sum_{i=1}^n \Big(\sigma_i^2 \frac{f^{'}_{\hat G_n,\sigma_i}(Z_i)}{f_{\hat G_n,\sigma_i}(Z_i) \vee \frac{\rho_n}{\sigma_i}} -\sigma_i^2 \frac{f^{'}_{G^{(j)},\sigma_i}(Z_i)}{f_{G^{(j)},\sigma_i}(Z_i)\vee \frac{\rho_n}{\sigma_i}}  \Big)^2 \\
	& \leq (\eta_n^*)^2.
\end{align*}

Taking expectations and using \eqref{eq:theorem1-score-cover} gives
\begin{equation}\label{regretboundP4}
	\frac{1}{n} \E \zeta_2^2
	\leq(\eta_n^*)^2
	\lesssim \frac{\ln n}{n^2}.
\end{equation}

\subsubsection*{Bounding $\E \zeta_3^2$}
Fix $j \in [N]$, for all $z_i \in \R$ with $\tilde z_i = z_i/\sigma_i$. Lemma  \ref{lemma: scaledTweedie} then implies 
\[
\tilde V_{ij}^2 \leq \sigma_i^2 \left( \frac{f'_{G_i^{(j)},1}(\tilde z_i)}{f_{G_i^{(j)},1}(\tilde z_i) \vee \rho_n} - \frac{f'_{G_{i0},1}(\tilde z_i)}{f_{G_{i0},1}(\tilde z_i) \vee \rho_n}\right)^2 \leq 4\sigma_i^2 \ln \left(\frac{1}{2 \pi \rho_n^2}\right).
\]
Hence,  
\[
\underset{ij}{\max} |\tilde V_{ij}| \leq 2 \sigma_u \sqrt{\ln \left(\frac{1}{2 \pi \rho_n^2} \right)} := K_n.
\]
Define $\bar V_{ij} = \frac{\tilde V_{ij}}{K_n}$; then, $(\bar V_{1j}, \dots, \bar V_{nj})$ are  random variables in $[0,1]$. We would like to establish a tail bound for $\|\tilde V_j\| - \E\|\tilde V_j\|$. 

Under Assumption \ref{boundedG0}, $\tilde V_{1j}, \dots, \tilde V_{nj}$ are independent. An application of the Lipschitz concentration inequality (see Theorem 6.10 of \citealt{inequalities}) then gives
\[
\P(\|\tilde V_j \| > \E \|\tilde V_j\| + t) = \P(\|\bar V_j\| > \E \|\bar V_j\| + t/K_n) \leq \exp(-t^2/2K_n^2),
\] 
which in turn implies, by the union bound, that
\[
\P(\zeta_3^2 > x) \leq N \exp\left(-\frac{x}{2K_n^2}\right).
\]
Consequently, 
\begin{align*}
	\frac{1}{n} \E \zeta_3^2 & = \frac{1}{n} \int_0^{\infty} \P(\zeta_3^2 \geq x) dx \\
	& \leq \frac{1}{n} \int_0^{+\infty} \min \Big(1, N(\exp(-x/2K_n^2))\Big) dx \\
	& = \frac{1}{n} \Big( \int_0^{2K_n^2 \ln N } dx + \int_{2K_n^2 \ln N}^{+\infty} N \exp(-x/2K_n^2)dx\Big) \\
	& \leq \frac{1}{n} \Big(2K_n^2 \ln N + 2K_n^2 N \exp(-\ln N) \Big) =  \frac{1}{n} \Big(2K_n^2 \ln N + 2K_n^2\Big).
\end{align*}

The conclusion in
\eqref{eq:theorem1-score-cover} combined with $K_n^2\lesssim\ln n$ leads to
\begin{equation}\label{regretboundP5}
	\frac{1}{n} \E\zeta_3^2 \lesssim \frac{(\ln n)^3}{n}.
\end{equation}

\subsubsection*{Bounding $\E\zeta_4^2$}
For a fixed $j \in [N]$,
\begin{align*}
	&\frac{1}{n}\big(\E\|\tilde V_j\|\big)^2\\
	&\quad=
	\frac{1}{n} \left( \E \left[ \sqrt{\sum_i 1\{|Z_i|\leq M_n\}\Big( \sigma_i^2 \frac{f'_{G^{(j)}, \sigma_i}(Z_i)}{f_{G^{(j)}, \sigma_i}(Z_i) \vee \frac{\rho_n}{\sigma_i}} - \sigma_i^2  \frac{f'_{G_0, \sigma_i}(Z_i)}{f_{G_0, \sigma_i}(Z_i) \vee \frac{\rho_n}{\sigma_i}} \Big)^2} \ \right]\right)^2\\
	& \leq \frac{1}{n} \sum_i \E \left[ 1\{|Z_i| \leq M_n\}  \Big( \sigma_i^2 \frac{f'_{G^{(j)}, \sigma_i}(Z_i)}{f_{G^{(j)}, \sigma_i}(Z_i) \vee \frac{\rho_n}{\sigma_i}} - \sigma_i^2  \frac{f'_{G_0, \sigma_i}(Z_i)}{f_{G_0, \sigma_i}(Z_i) \vee \frac{\rho_n}{\sigma_i}} \Big)^2 \right] \\
	& \leq \sigma_u^4 \frac{1}{n} \sum_i \E\left[ 1\{|Z_i| \leq M_n\} \Big(  \frac{f'_{G^{(j)}, \sigma_i}(Z_i)}{f_{G^{(j)}, \sigma_i}(Z_i) \vee \frac{\rho_n}{\sigma_i}} -    \frac{f'_{G_0, \sigma_i}(Z_i)}{f_{G_0, \sigma_i}(Z_i) \vee \frac{\rho_n}{\sigma_i}} \Big)^2 \right] 
\end{align*}
where the first inequality is due to Jensen's inequality and the second inequality holds under Assumption \ref{boundVar}. Denote 
\[
S_{G,\rho,\sigma_i} = \frac{f'_{G,\sigma_i}}{f_{G,\sigma_i} \vee \rho}, \quad S_{G,\sigma_i} = \frac{f'_{G,\sigma_i}}{f_{G,\sigma_i}} 
\]
then we have 
\begin{align*}
	& \frac{1}{n} \sum_i \E \left[ 1\{|Z_i| \leq M_n\}  \Big(  \frac{f'_{G^{(j)}, \sigma_i}(Z_i)}{f_{G^{(j)}, \sigma_i}(Z_i) \vee \frac{\rho_n}{\sigma_i}} -   \frac{f'_{G_0, \sigma_i}(Z_i)}{f_{G_0, \sigma_i}(Z_i) \vee \frac{\rho_n}{\sigma_i}} \Big)^2 \right] \\
	& = \frac{1}{n} \sum_i\E_{p_{G_0,i}}\Big[ 1\{|Z_i| \leq M_n\}   \Big( S_{G^{(j)}, \frac{\rho_n}{\sigma_i}, \sigma_i } - S_{G_0,\frac{\rho_n}{\sigma_i}, \sigma_i}\Big)^2\Big] \\
	& = \frac{1}{n} \sum_i\E_{p_{G_0,i}}\Big[ 1\{|Z_i| \leq M_n\}   \Big((S_{G^{(j)}, \sigma_i} - S_{G_0,\sigma_i}) + (S_{G^{(j)}, \frac{\rho_n}{\sigma_i},\sigma_i} - S_{G^{(j)}, \sigma_i}) - (S_{G_0,\frac{\rho_n}{\sigma_i},\sigma_i} - S_{G_0,\sigma_i}) \Big)^2\Big] \\
	& \leq \underbrace{3 \frac{1}{n} \sum_i\E_{p_{G_0,i}}\Big[1\{|Z_i| \leq M_n\}  \Big( \partial_z \ln \frac{p_{G_0,i}}{p_{G^{(j)},i}} \Big)^2\Big] }_{(I)}\\
	&+ \underbrace{3 \frac{1}{n} \sum_i \E_{ p_{G_0,i}}\Big[ S_{G^{(j)},\sigma_i}^2 1\{ f_{G^{(j)},\sigma_i} \leq \rho_n/\sigma_i, |Z_i| \leq M_n \} \Big]}_{(II)}\\
	& + \underbrace{3 \frac{1}{n} \sum_i\E_{p_{G_0,i}} \Big[S_{G_{0},\sigma_i}^2 1\{ f_{G_0,\sigma_i} \leq \rho_n/\sigma_i, |Z_i| \leq M_n \}\Big]}_{(III)}
\end{align*}
where the last inequality uses the fact that $\nabla \ln \frac{f_{G_0,\sigma_i}}{f_{G^{(j)},\sigma_i}} = \partial_z \ln \frac{p_{G_0,i}}{p_{G^{(j)}, i}}$ and $| S_{G,\rho,\sigma} - S_{G, \sigma}| \leq |S_{G,\sigma}| 1\{f_{G,\sigma} \leq \rho\}$. We now bound each of the three terms.

\subsubsection*{Bounding (I)} 
By construction, every cover center $G^{(j)}$ belongs to
$\mathcal H_n$.  Hence it is supported on $[-M_n,M_n]$. Lemma \ref{GNinequality}, with $B=M_n$,
therefore gives
\[
(I)
\leq 
C(1+M_n^{2})\Big( \bar h^2(p_{G_0}, p_{G^{(j)}})\Big)^{1-\frac{1}{m}}
\]
Furthermore by construction, we have $\bar h^2(p_{G_0},p_{G^{(j)}})
\leq
\epsilon_n^2$, therefore
\[
(I)
\lesssim
(\ln n) \Big (\epsilon_n^2\Big)^{1-\frac{1}{m}}.
\]

\subsubsection*{Bounding (III)} Using Lemma \ref{lem:density_dominance}, 
\begin{align*} 
	&\E_{p_{G_0,i}}\Big[\Big( \frac{f'_{G_0,\omega_i/\sqrt{\tau}}(z)}{f_{G_0, \omega_i/\sqrt{\tau}}(z)}\Big )^2 1\Big \{ f_{G_0, \omega_i/\sqrt{\tau}}(z) \leq \frac{\rho_n}{\omega_i/\sqrt{\tau}}\Big\} \Big]\\
	&	\leq \bar c \sqrt{\int \Big(  \frac{f'_{G_0, \omega_i/\sqrt{\tau}}(z)}{f_{G_0,\omega_i/\sqrt{\tau}}(z)}\Big )^4 f_{G_0, \omega_i/\sqrt{\tau}}(z) dz d\nu(\tau)} \sqrt{\int 1\{f_{G_0,\omega_i/\sqrt{\tau}}(z) \leq \frac{\rho_n}{\omega_i/\sqrt{\tau}}\} f_{G_0,\omega_i/\sqrt{\tau}}(z) dz d\nu(\tau)}
\end{align*}
Employing the same argument as that used for bounding \eqref{eq1:bound_on_A_4} and employing choice of $\rho_n = 1/n^3$, we obtain 
\begin{equation} \label{eq: III}
	(III) \lesssim \frac{1}{n}.
\end{equation}

\subsubsection*{Bounding (II)} 
By construction, every cover center $G^{(j)}$ belongs
to $\mathcal H_n$.  In particular,
$\operatorname{supp}(G^{(j)})\subseteq[-M_n,M_n]$.  The tweedie formula implies that
\[
\sigma_i^2 S_{G^{(j)},\sigma_i}(z)
=
\E_{G^{(j)}}[\theta\mid Z=z,\sigma_i]-z.
\]
then we have 
\begin{equation}\label{eq:internal-center-score-bound}
	\sup_{|z|\leq M_n}|S_{G^{(j)},\sigma_i}(z)|^2
	\leq \frac{4M_n^2}{\sigma_i^4}
	\leq \frac{4M_n^2}{\sigma_\ell^4}.
\end{equation}
Combining \eqref{eq:internal-center-score-bound} with Lemma
\ref{lemBound4} yields
\begin{align*}
	\frac{1}{3}(II)
	&\leq
	C M_n^2\frac{1}{n}\sum_i
	\mathbb P_{p_{G_0,i}}
	\left[f_{G^{(j)},\sigma_i}\leq\rho_n/\sigma_i,
	|Z_i|\leq M_n\right]\\
	&\leq
	C M_n^2
	\left\{
	\bar h^2(p_{G_0},p_{G^{(j)}})
	+\frac{\bar c\rho_nM_n}{\sigma_\ell}
	\right\}\\
	&\leq
	C' M_n^2\left(\epsilon_n^2+\rho_nM_n\right),
\end{align*}
Since
$M_n^2\asymp\ln n$, $\rho_n=n^{-3}$, and
$\epsilon_n^2=C_\epsilon(\ln n)^2/n$, this gives, uniformly over
$j\leq N$,
\begin{equation}\label{eq:theorem1-term-II}
	(II)
	\lesssim
	M_n^2\left(\epsilon_n^2+\rho_nM_n\right)
	\lesssim
	(\ln n)\epsilon_n^2
\end{equation}

Let $(I)_j$ and $(II)_j$ denote the preceding quantities for the center
$G^{(j)}$; term (III) does not depend on $j$.  By the definition of
$\zeta_4$, 
\begin{equation}\label{regretboundP6}
	\frac{1}{n}\E\zeta_4^2
	=
	\max_{j\leq N}\frac{1}{n}\big(\E\|\tilde V_j\|\big)^2
	\leq
	C\max_{j\leq N}\big\{(I)_j+(II)_j+(III)\big\}
	\lesssim
	(\ln n)\left(\epsilon_n^2\right)^{1-1/m}.
\end{equation}
where the final inequality uses the fact that $\epsilon_n^2\leq1$ for all sufficiently
large $n$, so that
$\epsilon_n^2\leq(\epsilon_n^2)^{1-1/m}$.  

\subsubsection{Putting together} 
Putting together (\ref{regretboundP1}), (\ref{regretboundP2}), (\ref{regretboundP3}), (\ref{regretboundP4}), (\ref{regretboundP5}) and (\ref{regretboundP6}) gives the desired result: 
\[
\mathcal{R}(\bm \delta^{NPEB}, G_0) \lesssim (\ln n)\cdot \Big(\frac{(\ln n)^{2}}{n}\Big) ^{(1-\frac{1}{m})}.
\]

\section{Average Hellinger Accuracy}\label{sec:Appendix:Hellinger}
Let $p_{G_0, i}(z_i, \tau_i)$ denote the true marginal density of $(Z_i,\tau_i)$ under the prior $G_0$ with $p_{G_0,i}(z_i,\tau_i) = c_i(z_i,\tau_i) f_{G_0, \omega_i/\sqrt{\tau_i}}(z_i)$ and $p_{\hat G_n, i}(z_i,\tau_i)$ be its estimator where we replace $G_0$ by its estimator $\hat G_n$. Define 
\[
\bar{h}^2(p_{\hat G_n}, p_{G_0}) = \frac{1}{n} \sum_i h^2(p_{\hat G_n,i}, p_{G_0,i})
\]
where $h^2(p,q) = \frac{1}{2}\int(\sqrt{p}-\sqrt{q})^2$ is the squared Hellinger distance between a pair of density $p,q$. 
\begin{prop}\label{HellingerAccuracy}
	Under Assumptions \ref{boundedG0}--\ref{likbound}, there exists a
	constant $C_\epsilon<\infty$ such that, with
	\[
	\epsilon_n^2
	:=
	C_\epsilon\frac{(\ln n)^2}{n},
	\]
	for every $t\geq1$ and all sufficiently large $n$,
	\[
	\mathbb P\left\{
	\bar h^2(p_{\widehat G_n},p_{G_0})
	\geq t^2\epsilon_n^2
	\right\}
	\leq 3n^{-t^2}.
	\]
\end{prop}
\textbf{Proof.}

Define $\Delta_{\mathrm{opt},n}
:=
\sup_{G\in\mathcal G(\mathbb R)}\ell_n(G)
-\ell_n(\widehat G_n).$
Since $G_0\in\mathcal G(\mathbb R)$,
\begin{align*}
	\ell_n(\widehat G_n)-\ell_n(G_0)
	&=
	\left\{
	\sup_{G\in\mathcal G(\mathbb R)}\ell_n(G)
	-\ell_n(G_0)
	\right\}
	-\Delta_{\mathrm{opt},n}\geq-\Delta_{\mathrm{opt},n}.
\end{align*}
Assumption \ref{likbound} therefore implies
\begin{align}
	\prod_{i=1}^n
	\frac{f_{\widehat G_n,\sigma_i}(Z_i)}
	{f_{G_0,\sigma_i}(Z_i)}
	&=
	\exp\left[
	n\{\ell_n(\widehat G_n)-\ell_n(G_0)\}
	\right]\nonumber\\
	&\geq
	\exp(-n\Delta_{\mathrm{opt},n})
	\geq
	\exp(-n\kappa_n)
	=q_n.
	\label{eq:Hellinger-likelihood-lower-bound}
\end{align}

Fix constants $0<\beta<\alpha<1$. For $n\geq3$, $-\ln q_n
=
2\ln n-\ln(e\sqrt{2\pi})
\leq2\ln n.$ 
Since
$nt^2\epsilon_n^2
=
C_\epsilon t^2(\ln n)^2$ and $t\geq1$, for all sufficiently large $n$, we have $-\ln q_n
\leq
(\alpha-\beta)nt^2\epsilon_n^2$. 
Consequently, $q_n
\geq
\exp\left\{
(\beta-\alpha)nt^2\epsilon_n^2
\right\}$. 
Combining this inequality with
\eqref{eq:Hellinger-likelihood-lower-bound} gives
\begin{equation}
	\prod_{i=1}^n
	\frac{f_{\widehat G_n,\sigma_i}(Z_i)}
	{f_{G_0,\sigma_i}(Z_i)}
	\geq
	\exp\left\{
	(\beta-\alpha)nt^2\epsilon_n^2
	\right\}.
	\label{eq:Hellinger-working-likelihood-bound}
\end{equation}

For $M>0$, define the deterministic pseudo-metric
\[
d_M^{(0)}(G,H)
:=
\sup_{\sigma_\ell\leq\sigma\leq\sigma_u}
\sup_{|z|\leq M}
|f_{G,\sigma}(z)-f_{H,\sigma}(z)|.
\]
Let $H_1,\ldots,H_{N_{\mathcal F}}$ be the deterministic
$\eta$-net supplied by part~(i)
of Lemma \ref{lem:uniform-mixture-entropy}. Let
\[
\mathcal P_n
:=
\{(p_{G,1},\ldots,p_{G,n}):G\in\mathcal G(\mathbb R)\}.
\]
By Lemma \ref{lem:density_dominance}, we have
\[
\max_{1\leq i\leq n}
\sup_{\substack{|z|\leq M\\
		\tau\in[\underline\tau,\bar\tau]}}
|p_{G,i}(z,\tau)-p_{H,i}(z,\tau)|
\leq \bar c\,d_M^{(0)}(G,H),
\]
hence $H_1,\ldots,H_{N_{\mathcal F}}$ induce a
$\bar c\eta$-net for $\mathcal P_n$ on $|z|\leq M$.
Based on this, for each $j$, let $H_{0,j}$ be a distribution that satisfies 
\[
\bar h(p_{H_{0,j}}, p_{G_0}) \geq t \epsilon_n, \quad\text{and}\quad
\max_{i\leq n}\sup_{\substack{|z|\leq M\\
		\tau\in[\underline\tau,\bar\tau]}}
|p_{H_{0,j},i}(z,\tau)-p_{H_j,i}(z,\tau)|
\leq\bar c\eta,
\]
assuming such a $H_{0,j}$ exists, and let $J = \{j \leq N_\mathcal{F}: H_{0,j} \text{ exists}\}$. Then for any distribution $G$ with $\bar h(p_G,p_{G_0}) \geq t \epsilon_n$, there exists $j \in J$ such that for $i = 1, \dots, n$, 
\[
p_{G, i}(z,\tau) \leq \begin{cases} p_{H_{0,j}, i}(z,\tau) + 2\bar c \eta = p_{H_{0,j}, i}(z,\tau) + 2h^*(z) & |z|\leq M\\
	c_i(z,\tau)/(\sqrt{2\pi} \omega_i/\sqrt{\tau}) & |z|>M
\end{cases} 
\]
where $h^*(z) = \bar c \eta 1\{|z|\leq M\} + \frac{\bar c \eta M^2}{z^2} 1\{|z|>M\}$. Note that by construction, $\int_{-\infty}^{+\infty} h^*(z) dz = 4 \bar c \eta M$.

On the event of $\bar h(p_{\hat G_n}, p_{G_0}) \geq t \epsilon_n$, we have
\begin{align*}
	\exp((\beta-\alpha)t^2n \epsilon_n^2)&\leq 	\displaystyle \prod_{i=1}^{n} \frac{f_{\hat G_n, \omega_i/\sqrt{\tau_i}}(Z_i) }{f_{G_0,\omega_i/\sqrt{\tau_i}}(Z_i)} = \prod_{i=1}^n \frac{p_{\hat G_n,i}(Z_i,\tau_i)}{p_{G_0,i}(Z_i,\tau_i)} \\
	&= \prod_{|Z_i|\leq M} \frac{p_{\hat G_n, i}(Z_i,\tau_i) }{p_{G_0,i}(Z_i,\tau_i)} \prod_{|Z_i|>M} \frac{p_{\hat G_n, i}(Z_i,\tau_i) }{p_{G_0,i}(Z_i,\tau_i)} \\
	& \leq \sup_{j \in J} \prod_{|Z_i|\leq M} \frac{p_{H_{0,j}, i}(Z_i,\tau_i) + 2h^*(Z_i)  }{p_{G_0,i}(Z_i,\tau_i)} \prod_{|Z_i|>M} \frac{c_i(Z_i,\tau_i)/(\sqrt{2\pi} \sigma_i)}{p_{G_0,i}(Z_i,\tau_i)}\\
	& \leq \sup_{j \in J} \prod_{i=1}^n \frac{p_{H_{0,j}, i}(Z_i,\tau_i) + 2h^*(Z_i)  }{p_{G_0,i}(Z_i,\tau_i)} \prod_{|Z_i|>M} \frac{\bar c /(\sqrt{2\pi} \sigma_i)}{2h^*(Z_i)}		
\end{align*}
where the last inequality used Assumption \ref{likfactor} and the result in Lemma \ref{lem:density_dominance}.   

It follows that
\begin{align}
	&		\mathbb{P}(\bar h(p_{\hat G_n}, p_{G_0}) \geq t \epsilon_n) \nonumber \\
	& \leq \mathbb{P}\Big (\sup_{j \in J} \prod_{i=1}^n \frac{p_{H_{0,j}, i}(Z_i,\tau_i) + 2h^*(Z_i)  }{p_{G_0,i}(Z_i,\tau_i)} \prod_{|Z_i|>M} \frac{\bar c /(\sqrt{2\pi} \sigma_i)}{2h^*(Z_i)}		 \geq 	\exp((\beta-\alpha)t^2n \epsilon_n^2)\Big ) \nonumber \\
	&  \leq \mathbb{P}\Big(\sup_{j \in J} \displaystyle \prod_{i=1}^n \frac{p_{H_{0,j}, i}(Z_i,\tau_i) + 2h^*(Z_i)  }{p_{G_0,i}(Z_i,\tau_i)} \geq \exp(-\alpha t^2 n \epsilon_n^2) \Big)\label{INEQ1}\\
	&\ + \mathbb{P}\Big (\displaystyle \prod_{|Z_i|>M} \frac{\bar c /(\sqrt{2\pi} \sigma_i)}{2h^*(Z_i)}		  \geq  \exp(\beta t^2 n\epsilon_n^2)\Big ). \label{INEQ2}
\end{align}
We bound the two probabilities separately.
\subsubsection*{Bounding \eqref{INEQ1}}
%
By the union bound and Markov's inequality, 
\begin{align*}
	&\mathbb{P}\Big (\sup_{j \in J} \prod_{i=1}^n \frac{p_{H_{0,j}, i}(Z_i,\tau_i) + 2h^*(Z_i)  }{p_{G_0,i}(Z_i,\tau_i)}   \geq   \exp(-\alpha t^2 n \epsilon_n^2)\Big )\\
	& \leq \exp\left(\frac{\alpha}{2} t^2 n \epsilon_n^2\right) \sum_{j\in J} \E \left[ \prod_i^n \left(\frac{p_{H_{0,j}, i}(Z_i,\tau_i) + 2h^*(Z_i)  }{p_{G_0,i}(Z_i,\tau_i)} \right)^{1/2} \right]. 
\end{align*}
By the independence across experiments, as assumed in Assumption \ref{boundedG0}, 
\begin{align*}
	&\E \left[ \prod_{i=1}^n\left( \frac{p_{H_{0,j}, i}(Z_i,\tau_i) + 2h^*(Z_i)  }{p_{G_0,i}(Z_i,\tau_i)} \right)^{1/2} \right] =  \prod_{i=1}^n \mathbb{E}_{p_{G_0,i}}\Big[ \Big(\frac{p_{H_{0,j}, i}(Z_i,\tau_i) + 2h^*(Z_i)  }{p_{G_0,i}(Z_i,\tau_i)} \Big)^{1/2}\Big] \\
	& \overset{(1)}{\leq} \exp \left( \sum_i  \E_{p_{G_0,i}} \left[\left( \frac{p_{H_{0,j}, i}(Z_i, \tau_i) + 2  h^*(Z_i)}{p_{G_0,i }(Z_i, \tau_i)}\right)^{1/2} - 1\right]\right) \\
	& \overset{(2)}{\leq}  \exp \left( \sum_i  \E_{p_{G_0,i}} \left[\left( \frac{p_{H_{0,j}, i}(Z_i, \tau_i)}{p_{G_0,i}(Z_i, \tau_i)}\right)^{1/2} - 1\right]  + \sum_i \E_{p_{G_0,i}} \left[\left( \frac{2 h^*(Z_i)}{p_{G_0, i}(Z_i, \tau_i)}\right)^{1/2} \right] \right) \\
	& = \exp \left( - \sum_i h^2(p_{H_{0,j}, i}, p_{G_0, i}) \right)  \exp \left( \sum_i \int \sqrt{2h^*(z_i)} \sqrt{p_{G_0, i}(z_i, \tau_i)} dz_i d\nu(\tau_i) \right)  \\
	& \overset{(3)}{\leq} \exp\Big( - n \bar h^2(p_{H_{0,j}}, p_{G_0})\Big) \exp\Big( n \sqrt{\int 2h^*(z) dz} \Big)\\
	& \overset{(4)}{=} \exp\Big( - n \bar h^2(p_{H_{0,j}}, p_{G_0})\Big) \exp\Big( 2 n \sqrt{2\eta \bar{c} M} \Big), 
\end{align*}
where (1) is due to $\ln x \leq x - 1$,  (2) is due to $\sqrt{a+b} \leq \sqrt{a} + \sqrt{b}$, (3) is due to the Cauchy-Schwartz inequality, and equality (4) holds due to $\int  h^*(z)  = 4\bar{c} \eta M$. 

Now, since $|J| \leq N_\mathcal{F}$, we have 
\begin{align*}
	& 	\P\Big(\sup_{j \in J} \prod_{i}^n \frac{p_{H_{0,j}, i}(Z_i, \tau_i) + 2h^*(Z_i)  }{p_{G_0, i}(Z_i, \tau_i)} \geq \exp(-\alpha t^2 n \epsilon_n^2) \Big)\\
	& \leq   \exp \left(-(1-\alpha/2) n t^2 \epsilon_n^2 + \ln N_\mathcal{F} + 2n \sqrt{2\eta \bar{c} M} \right).
\end{align*}

Choose $\eta:=n^{-2}$ and $M:=\sigma_u\sqrt{8\ln n},$ For all sufficiently large $n$, part~(i) of Lemma
\ref{lem:uniform-mixture-entropy} yields
\[
\ln N_{\mathcal F}
\leq
C_{\mathrm{ent}}(\ln n)^2
\]
for some constant $C_{\mathrm{ent}}<\infty$.  
Furthermore,
\[
2n\sqrt{2\bar c\eta M}
=
2\sqrt{2\bar cM}
\leq
C_{\mathrm{rem}}(\ln n)^2
\]
for all sufficiently large $n$ and some constant
$C_{\mathrm{rem}}<\infty$. Thus, if $(1-\alpha/2)C_\epsilon
\geq
C_{\mathrm{ent}}+C_{\mathrm{rem}}+1,$
\begin{equation}
	\eqref{INEQ1}
	\leq
	\exp\{-t^2(\ln n)^2\}.
	\label{eq:INEQ1-final}
\end{equation}

\subsubsection*{Bounding \eqref{INEQ2}}

Continue to use $\eta = 1/n^2$. By the definition of $h^*(z)$, we have $\frac{\bar c /(\sqrt{2\pi}\sigma_i)}{2h^*(Z_i)} = \frac{(nZ_i)^2}{2\sqrt{2\pi} \sigma_i M^2}$ for $|Z_i|>M$. Then, under Assumption \ref{boundVar} and using Markov's inequality,
\begin{align*}
	&  \mathbb{P}\left(\prod_{|Z_i|>M} \frac{\bar c /(\sqrt{2\pi} \sigma_i)}{ 2h^*(Z_i)} \geq  \exp(\beta t^2 n\epsilon_n^2)\right)\\
	& \leq \exp\left(-\frac{\beta t^2 n \epsilon_n^2}{2\ln n}\right) \E \left[ \prod_{i=1}^n  \left|\frac{nZ_i}{\sqrt{\sigma_l} M}\right|^{{\mathbb{I}\{|Z_i| > M\}}/\ln n}  \right].
\end{align*}
Writing out the expectation explicitly, 
\begin{align*}
	&\E \left[ \prod_{i=1}^n  \left|\frac{nZ_i}{\sqrt{\sigma_l} M}\right|^{{\mathbb{I}\{|Z_i| > M\}}/\ln n}  \right] 
	= \prod_{i=1}^n \E \left[\left|\frac{nZ_i}{\sqrt{\sigma_l} M}\right|^{{\mathbb{I}\{|Z_i| > M\}}/\ln n}  \right]  \\
	& \le \prod_{i=1}^n \E \left[ 1 + \left(\frac{n}{\sqrt{\sigma_{\ell}} M} \right)^{1/\ln n} |Z_i|^{1/\ln n}  \mathbb{I}\{|Z_i | > M \} \right]\\
	& \overset{(1)}{\leq} \exp \left( \sum_{i=1}^n \left(\frac{n}{\sqrt{\sigma_{\ell}} M} \right)^{1/\ln n} \int_{|z|>M} |z|^{1/\ln n} p_{G_0,i}(z,\tau) dz d\nu(\tau) \right) \\
	& \overset{(2)}{\leq} \exp \left( \sum_{i=1}^n \left(\frac{n}{\sqrt{\sigma_{\ell}} M} \right)^{1/\ln n} \bar{c} \int_{|z|>M} |z|^{1/\ln n} f_{G_0, \sigma_i}(z) dz \right) \\
	& \overset{(3)}{\leq} \exp\left(\sum_{i=1}^n \left(\frac{n}{\sqrt{\sigma_{\ell} }M}\right)^{1/\ln n} 4 \sigma_u M^{\frac{1}{\ln n} - 1} \frac{\bar{c}}{n \sqrt{2\pi}}\right) \\
	& =  \exp\left( \left(\frac{n}{\sqrt{\sigma_{\ell} }}\right)^{1/\ln n}  4 \sigma_u \frac{\bar{c}}{M\sqrt{2\pi}}\right)\\
	& = \exp\Big( \frac{e}{\sigma_{\ell}^{1/2\ln n}}\frac{\bar{c}}{\sqrt{\pi \ln n}}\Big). 
\end{align*}
In the above expression, inequality (1) uses the fact that $\log x \leq x - 1$. Inequality (2) follows from Assumption \ref{likfactor} and Lemma \ref{lem:density_dominance}. Inequality (3) uses Lemma \ref{lemma: tailmeanbound} after choosing $n$ large enough such that $M = \sigma_u \sqrt{8\ln n} \geq 2R$. Finally, the last equality is obtained by plugging in the choice of $M$. 

Based on the above, we obtain
\begin{align*}
	&  \mathbb{P}\Big (\prod_{|Z_i|>M} \frac{\bar c /\sqrt{2\pi} \sigma_i}{ 2h^*(Z_i)} \geq  \exp(\beta t^2 n\epsilon_n^2)\Big )\\
	& \leq  \exp\left(-\frac{\beta t^2 n \epsilon_n^2}{2\ln n} \right)\exp\left( \frac{e}{\sigma_{\ell}^{1/2\ln n}}\frac{\bar{c}}{\sqrt{\pi \ln n}}\right) \\
	& \leq \exp\Big(-\frac{C_\epsilon\beta t^2}{2} \ln n + \frac{e}{\sigma_{\ell}^{1/2\ln n}}\frac{\bar{c}}{\sqrt{\pi \ln n}}\Big)  \leq 2 \exp(-t^2 \ln n),
\end{align*}
where the last inequality holds when we plug in $\epsilon_n^2 = C_\epsilon(\ln n)^2/n$ with a constant $C_\epsilon$ such that $C_\epsilon\beta / 2 \geq 1$, and when $n$ is large enough that $\frac{e}{(\sigma_{\ell})^{1/2\ln n}}\frac{\bar{c}}{\sqrt{\pi \ln n}} \leq \ln 2$. 

Thus
\begin{equation}
	\eqref{INEQ2}
	\leq
	2\exp(-t^2\ln n)
	=
	2n^{-t^2}.
	\label{eq:INEQ2-final}
\end{equation}

Combining \eqref{eq:INEQ1-final} and
\eqref{eq:INEQ2-final} yields

\begin{align*}
	\mathbb P\left\{
	\bar h^2(p_{\widehat G_n},p_{G_0})
	\geq t^2\epsilon_n^2
	\right\}
	&=
	\mathbb P\left\{
	\bar h(p_{\widehat G_n},p_{G_0})
	\geq t\epsilon_n
	\right\}\\
	&\leq
	\exp\{-t^2(\ln n)^2\}
	+
	2\exp(-t^2\ln n)\\
	&\leq
	3\exp(-t^2\ln n)
	=
	3n^{-t^2}.
\end{align*}
This proves the proposition.

\section{Auxiliary Results\protect \label{sec:Appendix:SB}}

\begin{lem} \label{lem:density_dominance}
	Let $p_{G,i}(z_i, \tau_i)$ denote the marginal density of $Z_i, \tau_i$ under some prior $G$. Under Assumption \ref{likfactor},
	$$
	p_{G,i}(z_i, \tau_i) \le  \bar{c} \cdot f_{G,\sigma_i}(z_i).
	$$
	and furthermore for any two priors $G$ and $H$, 
	$$
	|p_{G, i}(z,\tau) - p_{H,i}(z,\tau)| \leq \bar c \cdot  |f_{G,\sigma_i}(z) - f_{H,\sigma_i}(z)|
	$$
\end{lem}

\begin{proof}
	The proof requires a bit of preparation in order to properly define the marginal density.
	
	Denote the observed outcomes and actions by $(A_{1,i},Y_{1,i},A_{2,i},Y_{2,i},\dots,A_{j,i},Y_{j,i},\dots)$
	in the sequence they are observed. We can interpret
	the sampling algorithm as mapping past data and some exogenous randomization
	$\bm{U}:=(U_{1,i},U_{2,i},\dots U_{j,i},\dots)$ to the set of actions $A_{j,i}$
	over the course of the experiment. It is without loss of generality
	to take $\bm{U}_{i}:= (U_{1,i}, U_{2,i},\dots)$ to be iid Uniform.  
	
	The maximal possible sample size is $\bar{N}$. It is useful to imagine
	that the total set of observations is always this fixed number. For
	open ended experiments, we can always achieve this by imagining that
	if $N_i < \bar{N}$, the remaining observations are exogenously
	drawn from a known distribution, say $\mathcal{N}(0,1)$.
	These observations are obviously ancillary to estimation of $\theta_{i}$
	and can be ignored. However, fixing the number allows us to define
	the likelihood densities with respect to a base measure in a consistent manner.
	We denote the full set of outcomes, with the addition of this auxiliary
	data, by $\bm{y}_i$. The observed data is therefore $\bm{y}_i,\bm{U}_i$. 
	
	Define the base measure
	\[
	\nu(\mathcal{D}_i)=m(\bm{y}_i)\otimes\gamma(\bm{U}_i),
	\]
	where $m(\bm{y}_i)$ denotes the $\bar{N}$ dimensional Lebesgue
	measure over $\bm{y}_i$, and $\gamma(\bm{U}_i)$ denotes the joint
	probability distribution over $\bm{U}_i$. 
	
	For any $\theta_i$, let $p(\mathcal{D}_i\vert\theta_i)$
	denote the likelihood with respect to the base measure $\nu(\mathcal{D}_i)$. Note
	that in this definition, the likelihood includes the auxiliary observations.
	Consider the transformation $\delta_{z,\tau}:(\bm{y}_i,\bm{U}_i)\to(Z_i,\tau_i)$,
	and define $\nu_{z,\tau}(\cdot)$ to be the $(\delta_{z,\tau},m(z)\otimes\nu(\tau))$-disintegration
	of $\nu(\cdot)$, where $m(z)$ denotes the univariate Lebesgue
	density, and $\nu(\tau)$ is the probability measure over $\tau$ defined in Assumption 3. The transformation $\delta_{z,\tau}$ depends on
	the algorithms used to generate the data (so the transformation is
	unknown, but all we need to know is that one exists). The function
	would not depend on the auxiliary observations. Based on the disintegration,
	we can define the marginal density of $(z,\tau)$ given
	$\theta_i$ as 
	\begin{align*}
		p(z,\tau\vert \theta_i) & :=\int p(\mathcal{D}_i \vert \theta_i)d\nu_{z,\tau}(\mathcal{D}_i).
	\end{align*}
	Intuitively, $\nu_{z,\tau}(\mathcal{D}_i)$ is the conditional distribution of $\nu(\mathcal{D}_i)$
	given $(z,\tau)$, multiplied by a Random-Nikodym factor to make it a density with respect to $m(z)\otimes\nu(\tau)$. 
	
	We know that for a given $\theta_i,\theta_i^\prime$ 
	\begin{align*}
		\frac{p(\mathcal{D}_i\vert \theta_i)}{p(\mathcal{D}_i \vert \theta_i^\prime)} 
		& = \left(\exp\left\{ \frac{Z_{i}\theta_{i}^{\prime}}{\sigma_{i}^{2}}-\frac{\theta_{i}^{\prime2}}{2\sigma_{i}^{2}}\right\} \right)^{-1}\exp\left\{ \frac{Z_{i}\theta_{i}}{\sigma_{i}^{2}}-\frac{\theta_{i}^{2}}{2\sigma_{i}^{2}}\right\} \\
		& = \frac{\sqrt{2\pi}}{\sigma_{i}}\varphi\left(\frac{Z_{i}-\theta_{i}}{\sigma_{i}}\right) \cdot \sigma_{i}\left(\exp\left\{ \frac{Z_{i}\theta_{i}^{\prime}}{\sigma_{i}^{2}}-\frac{\theta_{i}^{\prime2}}{2\sigma_{i}^{2}}\right\} \right)^{-1}\exp\left\{ \frac{Z_{i}^{2}}{2\sigma_{i}^{2}}\right\} .
	\end{align*}
	Taking $\theta_{i}^{\prime}=0$, we get
	\[
	p(\mathcal{D}_i \vert \theta_i)= \sqrt{2\pi} \sigma_i \exp\Big\{ \frac{Z_i^2}{2\sigma_i^2}\Big \}\cdot  \frac{1}{\sigma_{i}} \varphi\Big( \frac{Z_i - \theta_i}{\sigma_i}\Big)\cdot p(\mathcal{D}_i\vert \theta_i = 0).
	\]
	
	Based on the construction of $p(z,\tau \vert \theta_i)$, we can write 
	\begin{align*}
		p(z_i,\tau_i \vert \theta_i) & = \sqrt{2\pi}\sigma_i\exp\Big\{ \frac{z_i^2}{2\sigma_i^2}\Big \}\cdot  \frac{1}{\sigma_{i}} \varphi\Big( \frac{z_i - \theta_i}{\sigma_i}\Big)\cdot \int p(\mathcal{D}_i \vert \theta_i = 0) d\nu_{z,\tau}(\mathcal{D}_i)\\
		& = \sqrt{2\pi}\sigma_i\exp\Big\{ \frac{z_i^2}{2\sigma_i^2}\Big \}\cdot  \frac{1}{\sigma_{i}} \varphi\Big( \frac{z_i - \theta_i}{\sigma_i}\Big)\cdot p(z_i,\tau_i \vert \theta_i =0).
	\end{align*}
	Integrating the previous display with respect to the prior $G$, we get 
	\begin{equation}\label{eq:bounding_p_G}
		p_{G,i}(z_i,\tau_i) = f_{G}(z_i,\tau_i) \cdot \Big\{\sqrt{2\pi} \sigma_i \exp\Big\{ \frac{z_i^2}{2\sigma_i^2} \Big\} p(z_i,\tau_i \vert \theta_i = 0)\Big \}, 
	\end{equation}
	where 
	$$
	f_{G}(z_i, \tau_i) := \int \frac{1}{\sigma_i} \varphi\Big( \frac{z_i - \theta}{\sigma_i}\Big) dG(\theta) = f_{G,\sigma_i}(z_i). 
	$$
	The claim thus follows from Assumption \ref{likfactor}.
	The second statement holds by observing the fact that $c_i(z_i,\tau_i):= \Big\{\sqrt{2\pi} \sigma_i \exp\Big\{ \frac{z_i^2}{2\sigma_i^2} \Big\} p(z_i,\tau_i \vert \theta_i = 0)\Big \}$ does not depend on $\theta$, hence for any prior distribution $G$, we have 
	$$
	|p_{G,i}(z,\tau) - p_{H,i}(z,\tau)| = |c_i(z,\tau)| \cdot |f_{G,\sigma_i}(z) - f_{H,\sigma_i}(z)| \leq \bar c \cdot |f_{G,\sigma_i}(z) - f_{H,\sigma_i}(z)|.
	$$
\end{proof}

\begin{lem}\label{lemma: tailmeanbound} 
	Suppose that $X_i \sim N(0, \sigma_i^2)$, with  $\theta \sim G_0$ with support $[-R,R]$. Let $W_i  = X_i + \theta$ and denote the marginal density of $W_i$ as $f_{G_0,\sigma_i}$. Then, provided $M =  \sigma_u \sqrt{8\ln n} \geq 2R$, 
	\[
	\int_{|w|> M} |w|^{1/\ln n} f_{G_0,\sigma_i}(w) dw \leq 4 \sigma_u M^{\frac{1}{\ln n} - 1} \frac{1}{n \sqrt{2\pi}}.
	\]
\end{lem}

\begin{proof} 
	We can write 
	\begin{align*}
		&	\int_{|w|> M} |w|^{1/\ln n} f_{G_0,\sigma_i}(w) dw = \E\left[ |X_i + \theta|^{\frac{1}{\ln n}} 1\{|X_i + \theta |>M\} \right] \\
		& \leq \E\Big[ |2X_i|^{1/\ln n} 1\{|X_i| > M/2\}\Big] + \E\Big[ |2\theta|^{1/\ln n} 1\{|\theta| > M/2\}\Big] .		
	\end{align*}
	For $n \geq 3, (1/\ln n) - 1< 0$, so the first term can be bounded by 
	\[
	\E\Big[ |2X_i|^{1/\ln n} 1\{|X_i| > M/2\}\Big] \leq 2 M^{\frac{1}{\ln n} -1} \E\Big[ |X_i| 1\{|X_i|>M/2\}\Big].
	\]
	For the second term, since $\theta \in [-R,R]$ and given the choice $M = 4\sigma_u \sqrt{\ln n} \geq 2R$ which holds for large enough $n$, $|\theta|\leq M/2$ a.s. Based on this, we have 
	\begin{align*}
		&	\int_{|w|> M} |w|^{1/\ln n} f_{G_0,\sigma_i}(w) dw \\
		& \leq 2M^{\frac{1}{\ln n} - 1} \sigma_i \E\Big[ \left|\tfrac{X_i}{\sigma_i} \right| 1\left\{|\tfrac{X_i}{\sigma_i}| > \tfrac{M}{2\sigma_i}\right\}\Big] \\
		& \leq 4M^{\frac{1}{\ln n}-1} \sigma_u \int_{M/2\sigma_i}^{\infty} z \varphi(z) dz  \\
		& = 4M^{\frac{1}{\ln n}-1} \sigma_u\frac{1}{\sqrt{2\pi}} \exp(-M^2/8\sigma_i^2) \leq 4M^{\frac{1}{\ln n}-1} \sigma_u\frac{1}{n\sqrt{2\pi}},
	\end{align*}
	where $\varphi(\cdot)$ is the standard normal density and the last equality makes use of the fact that $\int_a^{\infty} z \varphi(z) dz =\varphi(a) =  \frac{1}{\sqrt{2\pi}} \exp(-a^2/2)$.
\end{proof} 

\begin{lem} \label{lemmaBound3}
	Let $X$ denote a random variable with density $f$ and variance $V_X$. Then, for any positive valued $M$ and $t$,
	\[
	\int 1\{f(x) \leq t\} f(x) dx \leq \frac{V_X}{M^2} + 2Mt.
	\]
	Specifically, if we pick $M = t^{-1/3} V_X^{1/3}$, 
	\[
	\int 1\{f(x) \leq t\} f(x) dx \leq 3 t^{2/3} V_X^{1/3}.
	\]
\end{lem} 

\begin{proof} 
	Since we can always employ a change of variable from $x$ to $x - c$ without changing the value of the integral, we assume without loss of generality that $\int x f(x)dx = 0$. Then, 
	\begin{align*}
		& \int 1\{f(x) \leq t\} f(x) dx = \int 1\{f(x) \leq t, |x|<M\} f(x) dx + \int 1\{f(x) \leq t, |x|>M\}f(x) dx \\
		& \leq \int_{-M}^M t dx + \P(|X| > M) \leq 2Mt + \frac{V_X}{M^2},
	\end{align*}
	where the last step is due to $\E[X] = 0$ and the Chebyshev inequality. 
\end{proof}

\begin{lem}\label{lemma: scaledTweedie}
	Let $G$ be any Borel probability measure and let $G_{i}$
	be the law of $\theta/\sigma_i$ when $\theta\sim G$. For
	$\widetilde z=z/\sigma_i$,
	\[
	f_{G,\sigma_i}(z)
	=
	\frac{1}{\sigma_i}f_{G_i,1}(\widetilde z),
	\qquad
	f'_{G,\sigma_i}(z)
	=
	\frac{1}{\sigma_i^2}f'_{G_i,1}(\widetilde z).
	\]
	Consequently,
	\[
	\widehat\theta_{i,G,\rho}(z)
	=
	\sigma_i\left\{
	\widetilde z+
	\frac{f'_{G_i,1}(\widetilde z)}
	{f_{G_i,1}(\widetilde z)\vee\rho}
	\right\}.
	\]
\end{lem}

\begin{proof}
	The change of variables $\widetilde\theta=\theta/\sigma_i$ gives
	\[
	f_{G,\sigma_i}(z)
	=
	\frac{1}{\sigma_i}
	\int\varphi(\widetilde z-\widetilde\theta)
	\,dG_i(\widetilde\theta)
	=
	\frac{1}{\sigma_i}f_{G_i,1}(\widetilde z).
	\]
	Differentiating with respect to $z$ and using
	$\partial_z\widetilde z=1/\sigma_i$ gives the derivative identity.
	Finally,
	\[
	f_{G,\sigma_i}(z)\vee\frac{\rho}{\sigma_i}
	=
	\frac{1}{\sigma_i}
	\left\{f_{G_i,1}(\widetilde z)\vee\rho\right\},
	\]
	and substitution into the definition of
	$\widehat\theta_{i,G,\rho}(z)$ proves the last display.
\end{proof}

\begin{lem} \label{GNinequality}
	Suppose Assumption \ref{boundedG0}--\ref{likfactor} and  Assumption \ref{support-geometry}-\ref{density-2} hold. Let $m$ be the fixed integer in Assumption \ref{density-2}. Fix $C_s < \infty$ and consider a prior $G$ with $\textrm{supp}(G) \subseteq [-B,B]$ such that $R \leq B \leq C_s M_n$. Then, 
	\[
	\frac{1}{n} \sum_i \E_{p_{G_0,i}}\Big[ 1\{|Z_i|<M_n\}\Big( \partial_z \ln \frac{p_{G_0,i}}{p_{G,i}}\Big)^2\Big] \leq 
	C(1+M_n^2) \Big[ \bar h^2(p_{G_0}, p_{G})\Big]^{1-\frac{1}{m}} 
	\]
	for some constant $C$ not depending on $n, i, G$. 
\end{lem} 

\begin{proof}
	For each $i$, let $u_i = \sqrt{p_{G_0, i}}$ and $v_i = \sqrt{p_{G,i}}$ and $g_i = u_i - v_i$. By \eqref{eq:bounding_p_G} we have $p_{G,i}(z,\tau) = r_i(z,\tau) \Lambda_{G,i}(z,\tau)$ with $\Lambda_{G,i}(z,\tau)= f_{G,\sigma_i}(z) \{\sqrt{2\pi} \sigma_i \exp(\frac{z^2}{2\sigma_i^2})\} > 0$, so $p_{G,i}$ is strictly positive on the same components as $r_i$. On every such component, 
	\begin{align*}
		&		1\{|z|< M_n\} \Big( \partial_z \ln \frac{p_{G_0,i}}{p_{G,i}}\Big)^2\\
		& \leq 	1\{|z|< M_n\}  \cdot 4 \Big (\frac{\partial_z u_i}{u_i}-\frac{\partial_z v_i}{v_i}\Big)^2\\
		& = 1\{|z|< M_n\} \cdot 4 \Big( \frac{\partial_z u_i - \partial_z v_i}{u_i} + \partial_z v_i \Big(\frac{1}{u_i}-\frac{1}{v_i}\Big)\Big)^2\\
		& \leq 1\{|z|< M_n\}  \cdot \Big(8 \frac{1}{u_i^2} (\partial_z u_i - \partial_z v_i)^2 + 8 \Big( \frac{\partial_z v_i}{v_i}\Big)^2 \frac{1}{u_i^2} [u_i - v_i]^2\Big),
	\end{align*}
	and so, 
	\begin{equation} \label{eq:Lemma5:1}
		\begin{aligned}[b]
			& \E_{p_{G_0,i}}\Big[ 1\{|Z_i|<M_n\}\Big( \partial_z \ln \frac{p_{G_0,i}}{p_{G,i}}\Big)^2\Big] \\
			&\leq 8 \int \int_{|z|<M_n,z \in \mathcal{S}_{i,\tau}} |\partial_z u_i - \partial_z v_i|^2 dz d\nu(\tau)
			+8 \int \int_{|z|<M_n, z \in \mathcal{S}_{i,\tau}} \Big(\frac{\partial_z  v_i}{v_i}\Big)^2 [u_i - v_i]^2 dz d\nu(\tau)
		\end{aligned}
	\end{equation}
	
	By Lemma \ref{Assumption_6_result_1}, on $|z|\leq M_n$, 
	\[
	\Big |\frac{\partial_z  v_i}{v_i}\Big| = \frac{1}{2} |\partial_z \ln p_{G,i}| \leq \frac{1}{2} \Big( C_1(1+M_n) + \frac{B}{\sigma_\ell^2}\Big)
	\]
	and since $\int [u_i- v_i]^2 = 2 h^2(p_{G_0,i}, p_{G,i})$, the second term of \eqref{eq:Lemma5:1} is at most 
	\begin{equation}
		4\Big( C_1(1+M_n) + \frac{B}{\sigma_\ell^2}\Big)^2 h^2(p_{G_0,i}, p_{G,i}).
	\end{equation}
	Since $C_1$ and $\sigma_\ell$ are fixed, this is at most
	$C(1+M_n^2+B^2)h^2(p_{G_0,i},p_{G,i})$ for some constant $C$.

	We now analyze the first term in (\ref{eq:Lemma5:1}). For brevity, we employ the notation 
	\[
	\| h\|_{2,\mathcal{S}_i}^2 :=\int \Big\{ \int_{\mathcal{S}_{i,\tau}}|h(z,\tau)|^2 dz \Big\}d\nu(\tau)
	\]
	Because $\mathcal{S}_{i,\tau}$ is an open subset of
	$\mathbb R$, it is the disjoint union of at most countably many connected
	components, each of which is either a bounded open interval, a half-line, or the entire real line
	$\mathbb R$.  We apply the Gagliardo--Nirenberg inequality separately on
	each component. For
	bounded $I$, apply the Gagliardo-Nirenberg interpolation inequality for a bounded-domain
	(see Remark~5 of \cite{nirenberg1959elliptic}) of $[0,1]$ and rescale affinely to $I$.  This gives 
	\begin{equation}\label{eq:GN-bounded-component}
		\int_I|g'|^2
		\le C_m\left[
		\left(\int_I|g^{(m)}|^2\right)^{1/m}
		\left(\int_I|g|^2\right)^{1-1/m}
		+|I|^{-2}\int_I|g|^2
		\right].
	\end{equation}
	For an unbounded component $I$ (a half-line or $\mathbb R$),
	the corresponding inequality reads 
	\begin{equation}\label{eq:GN-unbounded-component}
		\int_I|g'|^2
		\le C_m
		\left(\int_I|g^{(m)}|^2\right)^{1/m}
		\left(\int_I|g|^2\right)^{1-1/m}.
	\end{equation}
	Under Assumption \ref{support-geometry}, we thus obtain the following unified bound for any component $I$, 
	\begin{equation}\label{eq:GN-unified-component}
		\int_I|g'|^2
		\le C_{m,\ell_*}\left[
		\left(\int_I|g^{(m)}|^2\right)^{1/m}
		\left(\int_I|g|^2\right)^{1-1/m}
		+\int_I|g|^2
		\right].
	\end{equation}
	
	By Lemma \ref{Assumption_6_result_2} and Fubini's theorem, the
	componentwise inequalities below apply for $\nu$-almost every $\tau$.
	Fix such a $\tau$ and enumerate the connected components of
	$\mathcal S_{i,\tau}$ as $\{I_k\}_{k\in K_\tau}$. Define
	\[
	A(\tau):=\int_{\mathcal S_{i,\tau}}
	|\partial_z^m g_i(z,\tau)|^2\,dz,
	\qquad
	B(\tau):=\int_{\mathcal S_{i,\tau}}|g_i(z,\tau)|^2\,dz,
	\]
	and
	\[
	D(\tau):=\int_{\mathcal S_{i,\tau}}
	|\partial_z g_i(z,\tau)|^2\,dz.
	\]
	Summing \eqref{eq:GN-unified-component} over finite collections of
	components, applying H\"older's inequality with exponents $m$ and
	$m/(m-1)$, and then passing to all components by monotone convergence
	gives
	\[
	D(\tau)\le C_{m,\ell_*}
	\left[A(\tau)^{1/m}B(\tau)^{1-1/m}+B(\tau)\right].
	\]
	Finally, integrating over $\tau$ and applying H\"older's inequality once
	more,
	\[
	\int A(\tau)^{1/m}B(\tau)^{1-1/m}\,d\nu(\tau)
	\le
	\left(\int A\,d\nu\right)^{1/m}
	\left(\int B\,d\nu\right)^{1-1/m}.
	\]
	Therefore,
	\begin{equation}\label{eq:componentwise-GN-derived}
		\|\partial_zg_i\|_{2,\mathcal S_i}^2
		\le C_{m,\ell_*}\left[
		\|\partial_z^mg_i\|_{2,\mathcal S_i}^{2/m}
		\|g_i\|_{2,\mathcal S_i}^{2(1-1/m)}
		+\|g_i\|_{2,\mathcal S_i}^2
		\right].
	\end{equation}
	The same lemma, applied separately to $G$ and $G_0$, gives
	\[
	\|\partial_z^mv_i\|_{2,\mathcal S_i}^2\le C(1+B^{2m}),
	\qquad
	\|\partial_z^mu_i\|_{2,\mathcal S_i}^2\le C(1+R^{2m}).
	\]
	Therefore, using $|a-b|^2\le2a^2+2b^2$ and $R\le B$,
	\[
	\|\partial_z^mg_i\|_{2,\mathcal S_i}^2
	\le C(1+B^{2m}).
	\]
	Moreover, $p_{G,i}=r_i\Lambda_{G,i}$ with $\Lambda_{G,i}>0$, and
	Assumption \ref{support-geometry} makes $r_i$ vanish almost everywhere
	outside $\mathcal S_{i,\tau}$.  Hence both mixture densities vanish
	there and
	\[
	\|g_i\|_{2,\mathcal S_i}^2
	=2h^2(p_{G_0,i},p_{G,i}).
	\]
	Substituting the above in \eqref{eq:componentwise-GN-derived}, we get
	\begin{align*}
		&8\int\!\int_{\substack{|z|<M_n\\z\in\mathcal S_{i,\tau}}}
		|\partial_z u_i-\partial_z v_i|^2\,dz\,d\nu(\tau)\\
		&\qquad\le
		C(1+B^{2m})^{1/m}(2h^2(p_{G_0,i},p_{G,i}))^{1-1/m}
		+2C h^2(p_{G_0,i},p_{G,i})\\
		&\qquad\le
		C'(1+B^2)(h^2(p_{G_0,i},p_{G,i}))^{1-1/m}+C' h^2(p_{G_0,i},p_{G,i}).
	\end{align*}
	
	Armed with these two bounds for both terms in (\ref{eq:Lemma5:1}), as a last step, we average over $i$ and invoke Jensen's inequality, which yields $\frac{1}{n}\sum_i (h^2(p_{G_0,i}, p_{G,i}))^{1-\frac{1}{m}} \leq \left(\frac{1}{n}\sum_i h^2(p_{G_0,i}, p_{G,i})\right)^{1-\frac{1}{m}}$, to obtain
	\begin{align*}
		&\frac{1}{n} \sum_i \E_{p_{G_0,i}}\Big[ 1\{|Z_i|<M_n\}\Big( \partial_z \ln \frac{p_{G_0,i}}{p_{G,i}}\Big)^2\Big]\\
		&\leq 	C'(1+B^2)
		\{\bar h^2(p_{G_0},p_G)\}^{1-1/m}
		+C'(1+M_n^2+B^2)\bar h^2(p_{G_0},p_G)\\
		&\leq C''(1+M_n^2) \{\bar h^2(p_{G_0},p_G)\}^{1-1/m}
	\end{align*}
	The final inequality follows from
	$B\le C_sM_n$ and
	$x\le x^{1-1/m}$ for $0\le x\le1$ since $\bar h^2 \leq 1$.
\end{proof}

\begin{lem}\label{lemBound4}
	Under Assumption \ref{boundedG0} - \ref{likfactor}, let $\bar h^2(p_{G_0}, p_G) = \frac{1}{n} \sum_i h^2(p_{G_0,i}, p_{G,i})$, we have for any $G \in \mathcal{G}(\mathbb{R})$ and any $\rho > 0$ and $M > 0$, 
	\[
	\frac{1}{n}\sum_i \mathbb{P}_{p_{G_0,i}}[f_{G, \sigma_i}\leq \rho/\sigma_i, |Z|\leq M] \lesssim 4 \bar c \frac{\rho M}{\sigma_{\ell}} + 4\bar h^2(p_{G_0}, p_G).
	\]
\end{lem} 

\begin{proof}
	For any event $A$, and any two probability distributions $P$ and $Q$ with densities $p$ and $q$ respectively, we have 
	\[
	\Big	|\sqrt{P(A) }- \sqrt{Q(A)}\Big | \leq \sqrt{\int (\sqrt{p}-\sqrt{q})^2 1(A)  }\leq  \sqrt{2}h (p,q). 
	\]
	Therefore, $P(A)\leq \Big (\sqrt{Q(A)}+\sqrt{2}h(p,q)\Big )^2$, and it follows that 
	\[
	P(A) \leq 2Q(A)+ 4 h^2 (p,q).
	\]
	Applying this inequality, we get 
	\[
	\frac{1}{n}\sum_i \mathbb{P}_{p_{G_0,i}}[f_{G, \sigma_i}\leq \rho/\sigma_i, |Z_i|\leq M]  \leq 4 \bar h^2(p_{G_0}, p_{G}) + \frac{2}{n} \sum_i \mathbb{P}_{p_{G,i}}(f_{G, \sigma_i}\leq \rho/\sigma_i, |Z_i|\leq M)
	\]
	Now, by Assumption \ref{likfactor},
	\[
	\mathbb{P}_{p_{G,i}}(f_{G, \sigma_i}\leq \rho/\sigma_i, |Z_i|\leq M) \leq \bar c  \int_{|z|<M} 1\{f_{G,\sigma_i}<\rho/\sigma_i\} f_{G, \sigma_i} dz \leq 2\bar c \frac{\rho M}{\sigma_i} \leq 2\bar c \frac{\rho M}{\sigma_\ell}.
	\]
	Put together, we conclude
	\[
	\frac{1}{n}\sum_i \mathbb{P}_{p_{G_0,i}}[f_{G, \sigma_i}\leq \rho/\sigma_i, |Z_i|\leq M]  \leq 4 \bar h^2(p_{G_0}, p_{G}) + 4\bar c \rho M/\sigma_{\ell}. 
	\]
\end{proof} 

\begin{lem} \label{Assumption_6_result_1}
	Suppose Assumptions \ref{boundVar}, \ref{likfactor},  \ref{support-geometry} and 
	\ref{density-2} hold.  If 
	$\operatorname{supp}(G)\subseteq[-B,B]$, then
	\begin{equation}\label{eq:local-mixture-score}
		\operatorname*{ess\,sup}_{\substack{
				\tau,\ z\in\mathcal S_{i,\tau}\\ |z|\le M_n}}
		|\partial_z\ln p_{G,i}(z,\tau)|
		\le C_1(1+M_n)+\frac{B}{\sigma_\ell^2}.
	\end{equation}
\end{lem}

\begin{proof}
	On every component where $r_i(\cdot)$ is positive, exact likelihood factorization gives
	\[
	p_{G,i}(z,\tau)=r_i(z,\tau)\Lambda_{G,i}(z,\tau),
	\qquad
	\Lambda_{G,i}(z,\tau)
	=\int\exp\left\{
	\frac{z\theta}{\sigma_i^2}-\frac{\theta^2}{2\sigma_i^2}
	\right\}\,dG(\theta).
	\]
	Because $\Lambda_{G,i}>0$, $p_{G,i}$ and $r_i$ are positive on the same
	components.  Moreover,
	\[
	\partial_z\ln p_{G,i}
	=\partial_z\ln  r_i
	+\frac{\E_{\Pi_{G,i}(\cdot\mid z,\tau)}[\theta]}{\sigma_i^2},
	\]
	where
	\[
	\Pi_{G,i}(d\theta\mid z,\tau)
	\propto
	\exp\left\{\frac{z\theta}{\sigma_i^2}
	-\frac{\theta^2}{2\sigma_i^2}\right\}dG(\theta).
	\]
	On $|z|\le M_n$, Assumption \ref{density-2} with $l=1$ gives
	$|\partial_z\ln r_i(z,\tau)|\le C_1(1+M_n)$.  Since
	$|\E_{\Pi_{G,i}}\theta|\le B$ and $\sigma_i\ge\sigma_\ell$, the result
	follows.
\end{proof}

\begin{lem} \label{Assumption_6_result_2}
	
	Suppose Assumptions \ref{boundedG0}-\ref{likfactor} and
	\ref{support-geometry}-\ref{density-2} hold. Let $m$ be the fixed
	integer in Assumption \ref{density-2}. If
	$\operatorname{supp}(G)\subseteq[-B,B]$, then, for
	$j=1,\ldots,m$,
	\begin{equation}\label{eq:mixture-Sobolev-bound}
		\int\left\{\int_{\mathcal S_{i,\tau}}
		\left|\partial_z^j\sqrt{p_{G,i}(z,\tau)}\right|^2
		\,dz\right\}d\nu(\tau)
		\le C_j(1+B^{2j}).
	\end{equation}
	Here $C_j$ may depend on $j$ and the fixed constants in the assumptions,
	but not on $n,i,G$, or $B$.
	
\end{lem} 

\begin{proof}
	Fix $i$ and work on a connected component of
	$\mathcal S_{i,\tau}$, where Assumption \ref{density-2} supplies the
	required derivatives.  By Assumption \ref{support-geometry}, this set
	agrees with $\{z:r_i(z,\tau)>0\}$ up to a null set. Exact likelihood factorization gives
	\[
	p_{G,i}(z,\tau)
	=
	r_i(z,\tau)\Lambda_{G,i}(z,\tau),
	\]
	where
	\[
	\Lambda_{G,i}(z,\tau)
	:=
	\int
	\exp\left\{
	\frac{z\theta}{\sigma_i^2}
	-\frac{\theta^2}{2\sigma_i^2}
	\right\}dG(\theta)>0.
	\]
	Put
	\[
	\ell_{G,i}(z,\tau):=\ln p_{G,i}(z,\tau)
	=
	\ln r_i(z,\tau)+\ln \Lambda_{G,i}(z,\tau).
	\]
	
	Define
	\[
	\Pi_{G,i}(d\theta\mid z,\tau)
	:=
	\frac{
		\exp\left\{
		z\theta/\sigma_i^2-\theta^2/(2\sigma_i^2)
		\right\}dG(\theta)
	}{
		\Lambda_{G,i}(z,\tau)
	}.
	\]
	For $l\ge1$,
	\[
	\partial_z^l\ln \Lambda_{G,i}(z,\tau)
	=
	\sigma_i^{-2l}
	\kappa_l\left(\Pi_{G,i}(\cdot\mid z,\tau)\right),
	\]
	where $\kappa_1$ is the posterior mean and, for $l\ge2$,
	$\kappa_l$ is the $l$th posterior cumulant.
	
	Because $\Pi_{G,i}(\cdot\mid z,\tau)$ is supported on $[-B,B]$, the
	moment--cumulant formula gives
	\[
	\left|
	\kappa_l\left(\Pi_{G,i}(\cdot\mid z,\tau)\right)
	\right|
	\le K_lB^l.
	\]
	Consequently,
	\begin{equation}\label{eq:log-Lambda-derivative-bound}
		\left|
		\partial_z^l\ln\Lambda_{G,i}(z,\tau)
		\right|
		\le K_l\sigma_\ell^{-2l}B^l,
		\qquad l=1,\ldots,m.
	\end{equation}
	
	On the other hand, Assumption \ref{density-2} gives
	\begin{equation}\label{eq:log-r-derivative-bound}
		\left|
		\partial_z^l\ln r_i(z,\tau)
		\right|
		\le C_l(1+|z|)^l,
		\qquad l=1,\ldots,m.
	\end{equation}
	Combining \eqref{eq:log-Lambda-derivative-bound} and
	\eqref{eq:log-r-derivative-bound}, we obtain
	\begin{equation}\label{eq:mixture-log-derivative-bound}
		\left|
		\partial_z^l\ell_{G,i}(z,\tau)
		\right|
		\le
		C_l'\left\{(1+|z|)^l+B^l\right\}
		\le
		C_l''(1+|z|+B)^l.
	\end{equation}

	Because
	\[
	\sqrt{p_{G,i}}=\exp(\ell_{G,i}/2),
	\]
	Faà di Bruno's formula gives, for $j=1,\ldots,m$,
	\[
	\partial_z^j\sqrt{p_{G,i}}
	=
	\sqrt{p_{G,i}}
	\sum_{\substack{k_1,\ldots,k_j\ge0\\
			\sum_{l=1}^jlk_l=j}}
	c_{j,\bm k}
	\prod_{l=1}^j
	\{\partial_z^l\ell_{G,i}\}^{k_l}.
	\]
	For each monomial in this sum, $\sum_{l=1}^jlk_l=j$, it therefore follows from
	\eqref{eq:mixture-log-derivative-bound} that
	\[
	\left|
	\prod_{l=1}^j
	\{\partial_z^l\ell_{G,i}\}^{k_l}
	\right|
	\le
	C_j(1+|z|+B)^j.
	\]
	Since the number of partitions depends only on $j$,
	\begin{equation}\label{eq:sqrt-density-pointwise-bound}
		\left|
		\partial_z^j\sqrt{p_{G,i}(z,\tau)}
		\right|^2
		\le
		C_jp_{G,i}(z,\tau)(1+|z|+B)^{2j}.
	\end{equation}
	
	Integrating \eqref{eq:sqrt-density-pointwise-bound} over the
	positive-density region and applying Lemma
	\ref{lem:density_dominance} gives
	\begin{align*}
		&\int\left\{
		\int_{{\mathcal S_{i,\tau}}}
		\left|
		\partial_z^j\sqrt{p_{G,i}(z,\tau)}
		\right|^2dz
		\right\}d\nu(\tau)\\
		&\qquad\le
		\bar c \cdot C_j\int\!\!\int
		f_{G,\sigma_i}(z)(1+|z|+B)^{2j}
		\,dz\,d\nu(\tau).
	\end{align*}
	
	For fixed $\tau$, the density $f_{G,\sigma_i}$ is the density of
	\[
	Z=\Theta+\sigma_i\varepsilon,
	\qquad
	\Theta\sim G,\quad
	\varepsilon\sim N(0,1),
	\]
	where $|\Theta|\le B$ and $\sigma_i\le\sigma_u$. Hence
	\[
	1+|Z|+B
	\le
	1+2B+\sigma_u|\varepsilon|.
	\]
	Using the elementary inequality
	$
	(a+b+c)^{2j}
	\le C_j(a^{2j}+b^{2j}+c^{2j}),
	$
	we obtain
	\begin{align*}
		\int f_{G,\sigma_i}(z)(1+|z|+B)^{2j}\,dz
		&=
		\E(1+|\Theta+\sigma_i\varepsilon|+B)^{2j}\\
		&\le
		C_j\left\{
		1+B^{2j}
		+\sigma_u^{2j}\E|\varepsilon|^{2j}
		\right\}\\
		&\le C_j'(1+B^{2j}).
	\end{align*}
	The bound is uniform in $\tau$, and $\nu$ is a probability measure.
	Integrating over $\tau$ proves
	\eqref{eq:mixture-Sobolev-bound}.	
\end{proof}

\begin{lem}
	\label{lem:uniform-mixture-entropy}
	Fix $0<\sigma_\ell\leq\sigma_u<\infty$.  There are constants
	$c_0>1$, $C_0<\infty$, and $\eta_0\in(0,1)$, depending only on
	$(\sigma_\ell,\sigma_u)$, such that the following hold for every
	$M>0$ and $0<\eta\leq\eta_0$.  For $0<x\leq\eta_0$, write
	$L_x:=\ln(c_0/x)$.  
	
	\emph{(i) Density cover.}
	There are deterministic priors
	$H_1,\ldots,H_{N_0}\in\mathcal G(\mathbb R)$ such that
	\begin{equation}\label{eq:uniform-density-cover}
		\begin{aligned}
			\sup_{G\in\mathcal G(\mathbb R)}\min_{j\leq N_0}
			\sup_{\sigma_\ell\leq\sigma\leq\sigma_u}
			\sup_{|z|\leq M}
			|f_{G,\sigma}(z)-f_{H_j,\sigma}(z)|
			&\leq\eta,\\
			\ln N_0
			&\leq
			C_0\left(1+\frac{M}{\sqrt{L_\eta}}\right)L_\eta^2.
		\end{aligned}
	\end{equation}
	
	\emph{(ii) Regularized-score cover.}
	For $0<\rho\leq(2\pi e)^{-1/2}$, define
	\[
	T_{G,\rho,\sigma}(z)
	:=
	\sigma^2\frac{f'_{G,\sigma}(z)}
	{f_{G,\sigma}(z)\vee(\rho/\sigma)},
	\qquad
	\Lambda_\rho:=\ln\frac1{2\pi\rho^2},
	\]
	and
	\[
	d_{M,\rho}(G,H)
	:=
	\sup_{\sigma_\ell\leq\sigma\leq\sigma_u}
	\sup_{|z|\leq M}
	|T_{G,\rho,\sigma}(z)-T_{H,\rho,\sigma}(z)|.
	\]
	For every nonempty deterministic subclass
	$\mathcal H\subseteq\mathcal G(\mathbb R)$, there are deterministic
	priors
	$Q^{\mathcal H}_1,\ldots,Q^{\mathcal H}_{N_{\mathcal H}}\in\mathcal H$
	such that
	\begin{equation}\label{eq:subclass-score-cover}
		\begin{aligned}
			\sup_{G\in\mathcal H}
			\min_{j\leq N_{\mathcal H}}
			d_{M,\rho}(G,Q^{\mathcal H}_j)
			&\leq
			\left(\sigma_u^3+\sigma_u^2\sqrt{\Lambda_\rho}\right)
			\frac{\eta}{\rho},\\
			\ln N_{\mathcal H}
			&\leq
			C_0\left(1+\frac{M}{\sqrt{L_\eta}}\right)L_\eta^2.
		\end{aligned}
	\end{equation}
	
\end{lem}

\begin{proof}
	Set $\underline k=\sigma_\ell^2$ and $\bar k=\sigma_u^2$.  We use
	Lemmas~5--6 and Appendix~D.3.5 of the online supplement of
	\cite{soloff2024multivariate}, specialized to $d=1$.  Those lemmas are
	stated for a fixed finite list of covariance matrices.  Their proofs,
	however, first approximate the mixing distribution by moment matching
	and then replace its atoms and weights by elements of fixed finite
	nets.  These nets depend only on the approximation accuracy, the
	spatial set, and $(\underline k,\bar k)$, and every approximation bound
	in the construction uses only
	$\underline k\leq\Sigma\leq\bar k$.  Consequently, the same finite
	collection of candidate priors works simultaneously for every
	$\Sigma=\sigma^2$ with
	$\sigma\in[\sigma_\ell,\sigma_u]$.  Thus the maxima over the finite
	covariance list in those proofs may be replaced by the supremum over
	$\sigma\in[\sigma_\ell,\sigma_u]$.
	
	For part~(i), take $S=\{0\}$ and localization radius $M$ in Lemma~5
	of the online supplement of \cite{soloff2024multivariate}, so that $S^M=[-M,M]$.  Its enlargement
	parameter satisfies $a_\eta\asymp\sqrt{L_\eta}$, while
	$S^{M+a_\eta}=[-M-a_\eta,M+a_\eta]$.  In one dimension,
	\[
	N\!\left(a_\eta,[-M-a_\eta,M+a_\eta]\right)
	\leq C\left(1+\frac{M}{a_\eta}\right)
	\leq C'\left(1+\frac{M}{\sqrt{L_\eta}}\right).
	\]
	Substitution into the entropy bound in that lemma proves
	\eqref{eq:uniform-density-cover}, after adjusting $c_0$, $C_0$, and
	$\eta_0$.
	
	We next prove part~(ii).  For $G,H\in\mathcal G(\mathbb R)$, define
	\begin{align*}
		\Delta_{0,M}(G,H)
		&:=\sup_{\sigma_\ell\leq\sigma\leq\sigma_u}
		\sup_{|z|\leq M}|f_{G,\sigma}(z)-f_{H,\sigma}(z)|,\\
		\Delta_{1,M}(G,H)
		&:=\sup_{\sigma_\ell\leq\sigma\leq\sigma_u}
		\sup_{|z|\leq M}|f'_{G,\sigma}(z)-f'_{H,\sigma}(z)|.
	\end{align*}
	Fix $\sigma\in[\sigma_\ell,\sigma_u]$, and put
	$D_G=f_{G,\sigma}\vee(\rho/\sigma)$ and
	$D_H=f_{H,\sigma}\vee(\rho/\sigma)$.  Since
	$|D_G-D_H|\leq|f_{G,\sigma}-f_{H,\sigma}|$ and
	$D_H\geq\rho/\sigma$, we have
	\begin{align*}
		|T_{G,\rho,\sigma}-T_{H,\rho,\sigma}|
		&\leq
		\frac{\sigma^2}{D_H}|f'_{G,\sigma}-f'_{H,\sigma}|
		+|T_{G,\rho,\sigma}|\frac{|D_G-D_H|}{D_H}.
	\end{align*}
	Lemma~F.1 of \cite{saha2020nonparametric} gives
	$|T_{G,\rho,\sigma}(z)|\leq\sigma\sqrt{\Lambda_\rho}$.
	Taking the supremum over $|z|\leq M$ and
	$\sigma\in[\sigma_\ell,\sigma_u]$ therefore yields
	\begin{equation}\label{eq:score-cover-comparison}
		d_{M,\rho}(G,H)
		\leq
		\frac{\sigma_u^2\sqrt{\Lambda_\rho}}{\rho}
		\Delta_{0,M}(G,H)
		+\frac{\sigma_u^3}{\rho}\Delta_{1,M}(G,H).
	\end{equation}
	
	Applying 
	Lemmas~5 and~6 of the online supplement of
	\cite{soloff2024multivariate}, with component accuracy $\eta/4$,
	we have deterministic priors $A_1,\ldots,A_{N_{\mathrm{den}}}$ and $B_1,\ldots,B_{N_{\mathrm{der}}}$ such that 
	\[
	\sup_{G\in\mathcal G(\mathbb R)}
	\min_{j\leq N_{\mathrm{den}}}
	\Delta_{0,M}(G,A_j)
	\leq\frac{\eta}{4},
	\qquad
	\sup_{G\in\mathcal G(\mathbb R)}
	\min_{k\leq N_{\mathrm{der}}}
	\Delta_{1,M}(G,B_k)
	\leq\frac{\eta}{4}.
	\]
	Moreover, because $L_{\eta/4}\asymp L_\eta$, the entropy bounds in
	those two lemmas imply, after increasing $C_0$ if necessary, that
	\begin{equation}\label{eq:component-cover-entropy}
		\ln N_{\mathrm{den}}+\ln N_{\mathrm{der}}
		\leq
		C_0\left(1+\frac{M}{\sqrt{L_\eta}}\right)L_\eta^2.
	\end{equation}
	For every pair $(j,k)$
	for which
	\[
	\left\{G:\Delta_{0,M}(G,A_j)\leq\eta/4\right\}
	\cap
	\left\{G:\Delta_{1,M}(G,B_k)\leq\eta/4\right\}
	\neq\varnothing,
	\]
	choose one prior $H_{jk}$ in this intersection, and let $\mathcal C$
	be the collection of all priors so chosen.  Given any
	$G\in\mathcal G(\mathbb R)$, choose density and derivative cover cells
	containing $G$.  The representative $H_{jk}$ of their intersection
	satisfies, by the triangle inequality,
	\[
	\Delta_{0,M}(G,H_{jk})\leq\eta/2,
	\qquad
	\Delta_{1,M}(G,H_{jk})\leq\eta/2.
	\]
	Writing $K_\rho:=\sigma_u^3+\sigma_u^2\sqrt{\Lambda_\rho}$,
	\eqref{eq:score-cover-comparison} consequently gives
	\[
	\sup_{G\in\mathcal G(\mathbb R)}
	\min_{H\in\mathcal C}d_{M,\rho}(G,H)
	\leq r:=K_\rho\frac{\eta}{2\rho}.
	\]
	Moreover,
	$|\mathcal C|\leq N_{\mathrm{den}}N_{\mathrm{der}}$.  Hence
	\eqref{eq:component-cover-entropy} gives
	\[
	\ln|\mathcal C|
	\leq C_0\left(1+\frac{M}{\sqrt{L_\eta}}\right)L_\eta^2.
	\]
	
	Finally, $d_{M,\rho}$ is a pseudometric and hence satisfies the
	triangle inequality.  For each $H\in\mathcal C$ whose closed
	$d_{M,\rho}$-ball of radius $r$ meets $\mathcal H$, choose one
	$Q_H\in\mathcal H$ from that intersection.  If $G\in\mathcal H$ and
	$d_{M,\rho}(G,H)\leq r$, then
	\[
	d_{M,\rho}(G,Q_H)
	\leq d_{M,\rho}(G,H)+d_{M,\rho}(H,Q_H)
	\leq2r=K_\rho\frac{\eta}{\rho}.
	\]
	There are at most $|\mathcal C|$ selected priors, proving
	\eqref{eq:subclass-score-cover}.
	
\end{proof}

\end{document}